\documentclass[prb,aps,twocolumn,superscriptaddress,10pt]{revtex4}

\usepackage{amsfonts}
\usepackage{amsmath}
\usepackage{amssymb}
\usepackage{amsthm}
\usepackage{bm}
\usepackage{enumerate}
\usepackage{esint}
\usepackage{float}
\usepackage{graphicx}
\usepackage[pdftex,colorlinks=red]{hyperref}
\usepackage{mathtools}
\usepackage{MnSymbol}
\usepackage{pdfpages}
\usepackage{psfrag}
\usepackage{subfigure}
\usepackage{xcolor}
\usepackage[all]{xy}

\newcommand{\ket}[1]{|#1\rangle}
\newcommand{\bra}[1]{\langle#1|}

\newcommand{\Tr}{\mathrm{Tr}}

\newcommand{\Z}{\mathbb{Z}}
\newcommand{\BMS}{\mathrm{BM}}
\newcommand{\RBMS}{\mathrm{RBM}}
\newcommand{\Ocal}{\mathcal{O}}

\newcommand{\Sym}{\mathrm{Sym}}

\newcommand{\Lcal}{\mathcal{L}}
\newcommand{\Rcal}{\mathcal{R}}
\newcommand{\Ecal}{\mathcal{E}}

\newtheorem{theorem}{Theorem}[section]

\newtheorem{lemma}[theorem]{Lemma}
\newtheorem{assumption}[theorem]{Assumption}

\begin{document}
\title{Restricted Boltzmann Machines and Matrix Product States of 1D Translational Invariant Stabilizer Codes}

\author{Yunqin Zheng}
\affiliation{Physics Department, Princeton University, Princeton, New Jersey 08544, USA}

\author{Huan He}
\affiliation{Physics Department, Princeton University, Princeton, New Jersey 08544, USA}

\author{Nicolas Regnault}
\affiliation{
	Laboratoire Pierre Aigrain, D\'epartement de physique de l'ENS, \'Ecole normale sup\'erieure, PSL Research University, Universit\'e Paris Diderot, Sorbonne Paris Cit\'e, Sorbonne Universit\'es, UPMC Univ. Paris 06, CNRS, 75005 Paris, France
}

\author{B. Andrei Bernevig}
\affiliation{Physics Department, Princeton University, Princeton, New Jersey 08544, USA}
\affiliation{Physics Department, Freie Universitat Berlin, Arnimallee 14, 14195 Berlin, Germany}
\affiliation{Max Planck Institute of Microstructure Physics, 06120 Halle, Germany}


\begin{abstract}
We discuss the relations between restricted Boltzmann machine (RBM) states and the matrix product states (MPS) for the ground states of 1D translational invariant stabilizer codes. A generic translational invariant and finitely connected RBM state can be expressed as an MPS, and the matrices of the resulting MPS are of rank 1. We dub such an MPS as an RBM-MPS. This provides a necessary condition for exactly realizing a quantum state as an RBM state, if the quantum state can be written as an MPS. For generic 1D stabilizer codes having a non-degenerate ground state with periodic boundary condition, we obtain an expression for the lower bound of their MPS bond dimension, and an upper bound for the rank of their MPS matrices. In terms of RBM, we provide an algorithm to derive the RBM for the cocycle Hamiltonians whose MPS matrices are proved to be of rank 1. Moreover, the RBM-MPS produced by our algorithm has the minimal bond dimension. A family of examples is provided to explain the algorithm. We finally conjecture that these features hold true for all the 1D stabilizer codes having a non-degenerate ground state with periodic boundary condition, as long as their MPS matrices are of rank 1.
\end{abstract}

\maketitle

\tableofcontents

\clearpage

\section{Introduction}
\label{Sec.Introduction}

Restricted Boltzmann machines (RBM) and more generally neural networks \cite{Carleo2017Solving,torlai2018neural,carleo2018constructing,jia2018efficient,Kaubruegger2018Chiral,Glasser2018Neural,Clark2018Unifying,Nomura2017Restricted,Rao2018Identifying,morningstar2017deep,gao2017efficient,Deng2017Quantum,Chen2018Equivalence,Aoki2016Restricted,Tubiana2017Emergence, 2018PhRvB..97c5116C,Huang2017Accelerated, 2018arXiv181002352L}, have recently gained lots of attention as numerical tools for studying quantum many-body physics, boosted by the fast paced progress in machine learning. An RBM is a restriction from a Boltzmann machine (BM). The BM is defined on a bipartite graph, whose vertices are grouped into two classes: the visible vertices and the hidden vertices.  Suppose there are $n$ visible vertices and $m$ hidden vertices, and we associate the visible variables $g\in\{0,1\}^n$ and the hidden variables $h\in \{0,1\}^m$ on the visible and hidden vertices respectively. The variables $\{g, h\}$ obey the Boltzmann distribution, 
\begin{equation}\label{eq.boltzmanndistribution}
P(g,h)=\frac{1}{Z}\exp\left(-\Ecal(g,h)\right),
\end{equation}
where $\Ecal(g,h)$ is a real function mimicking the ``energy" in the Boltzmann distribution, and $Z=\sum_{g,h}\exp\left(-\Ecal(g,h)\right)$ is the partition function. As the name suggests, only the visible variables will show up in the physical probability distribution, while the hidden variables are summed over and thus hidden. Given Eq.~\eqref{eq.boltzmanndistribution}, the BM is defined to be the marginal distribution $P(g)$ over the visible variables $g$ by summing over all the hidden variables $\{h\}$
\begin{equation}
P(g) = \sum_{h} P(g,h) = \frac{1}{Z} \sum_{h}\exp\left(-\Ecal(g,h)\right).
\end{equation}
The RBM further requires that the ``energy" function $\Ecal(g,h)$ depends linearly on $g$ and $h$. The most important property of RBM is its representing power. It has been proven\cite{le2008representational} in the machine learning context that any probability distribution $P_0(g)$ of an $n$ number of $\Z_2$ variables, i.e., $g\in \{0,1\}^n$, can be approximated arbitrarily well by an RBM $P(g)$ given enough number of hidden spins. See Ref.~\onlinecite{le2008representational} for details. 

For the purposes of the quantum physics, it is natural to change the ``energy" function $\Ecal(g,h)$ from a real function to a complex one. Then we can interpret the ``complex probability distribution" $P(g)$ as the coefficients of a quantum many-body wave function:\footnote{To obtain a normalized state, we need to rescale $P(g)$ by a common factor irrelevant of $g$. } \footnote{The construction of the quantum wave-function from a classical Hamiltonian has been discussed in Ref.~\onlinecite{PhysRevB.91.155150} following the work of Rokhsar and Kivelson\cite{PhysRevLett.61.2376}.}
\begin{equation}\label{Eq.1}
\ket{\Psi} = \sum_{g} P(g) \ket{g},
\end{equation}
where $\ket{g}$ is the basis to expand the quantum states $\ket{\Psi}$. The RBM state refers to the ansatz in Eq.~\eqref{Eq.1}. 

The ground state of a 1D gapped local Hamiltonian has entanglement entropy $S(L)$ for a subregion of length $L$ which obeys area law\cite{Hastings2007Area} and is in fact a constant, i.e., $S(L)\sim \mathrm{constant}$. Therefore, we expect that it suffices to use a constant number of hidden spins per visible spin, when we represent such a ground state by an RBM. One of the purposes of this paper is to study the representing power of RBM for 1D stabilizer code ground states using as few hidden spins as possible.

For a 1D gapped local Hamiltonian, its ground state is efficiently encoded as a matrix product state (MPS)\cite{1987PhRvL..59..799A, fannes1992, PhysRevLett.69.2863, PhysRevLett.75.3537,  Hastings2007Area, 2006quant.ph..8197P, 2008AdPhy..57..143V,  2008PhRvL.100c0504S, schollwock2011}. An MPS can be obtained either numerically (via, e.g., the density matrix renormalization group) or analytically. The core reason of this efficiency is the area law satisfied by these many body quantum states. Indeed, the entanglement of a generic MPS is upper bounded by the dimension of the MPS matrices, i.e., the bond dimension $D$. On general grounds, the RBM state and the MPS share many features in common\cite{ Chen2018Equivalence, Deng2017Quantum, 2017arXiv170106246H, 2018arXiv180810601J}. Both the virtual indices in the MPS and the hidden variables in the RBM serve as the glue between different parts of the state, and thus provide nontrivial entanglement.  In the literature, some of the relations between MPS and the RBM have already been studied. In Refs.~\onlinecite{Chen2018Equivalence} and  \onlinecite{2018arXiv181002352L}, a numerical algorithm has been proposed to convert an RBM into an MPS as well as into a projected entangled pair state (PEPS) - a generalization of MPS in higher dimension. Refs.~\onlinecite{Chen2018Equivalence} and  \onlinecite{2018arXiv181002352L} also studied how the PEPS are mapped to the RBM for some subset of the stabilizer codes whose interactions are products of either purely Pauli $X$ or purely $Z$ operators (e.g. the toric code model\cite{kitaev2003fault}). In Refs.~\onlinecite{Deng2017Quantum} and \onlinecite{PhysRevB.96.195145}, an RBM state for the 1D $ZXZ$ model was found numerically. However, their RBM state is not optimal: the MPS from their RBM state has bond dimension 4, which is larger the minimal bond dimension 2. In other words, their MPS uses too many hidden spin variables. We present another analytical construction which yields an optimal RBM state for the $ZXZ$ model. We also systematically construct the optimal RBM states for a large family of stabilizer codes.  

In this paper, we make progresses toward answering the following questions, when the stabilizer codes have one ground state with periodic boundary condition (PBC):
\begin{enumerate}
	\item How to map a translational invariant and finitely connected RBM to an MPS?
	\item Given a stabilizer code, how to find the MPS of its ground state?
	\item Given a stabilizer code, can we cast the ground state as an RBM state minimizing the number of hidden spins?
\end{enumerate}

A crucial concept in our paper is the \emph{rank} of the MPS matrix, which we define to be the rank of the matrix with the \emph{fixed} physical index. We will justify the validity of this concept in Sec.~\ref{Sec.IntroBM}. 

This paper is organized as follows: In Sec.~\ref{Sec.IntroBM}, we review the BM and RBM, and discuss how an RBM state can be expressed as an MPS which we dub as RBM-MPS. In particular, we show that the rank of the non-vanishing RBM-MPS matrices must be 1.  In Sec.~\ref{Sec.MPSSC}, we present an algorithm deriving the MPS from a stabilizer code Hamiltonian. We illustrate the algorithm through the example of $ZZXZZ$ model. We derive a lower bound for the bond dimension of the MPS matrices. In Sec.~\ref{Sec.RankMPS},  we provide an upper bound for the rank of the MPS matrices. In Sec.~\ref{Sec.BMSofQSC}, we show that the MPS matrices for the cocycle Hamiltonians, which are representative Hamiltonians of generic SPT phases, are of rank 1. We provide an algorithm to derive the optimal RBM of the cocycle Hamiltonians whose bond dimension of the RBM-MPS saturates the lower bound that is derived in Sec.~\ref{Sec.MPSSC}. We apply our algorithm to the $Z^{q-1}XZ^{q-1}$ models deriving their RBM-MPS matrices. 

\section{(Restricted) Boltzmann Machine}
\label{Sec.IntroBM}

In this section, we introduce the notion of Boltzmann machine (BM) states,  restricted Boltzmann machine (RBM) states and their connection to MPS. 

\subsection{Definitions}

A BM state is a state defined by a classical Ising model on a graph. Each vertex of the graph carries a classical Ising spin $s^r=0, 1$ where $r$ is the index of the vertex. Each edge of the graph carries a weight $\widehat{W}_{rr'}\in \mathbb{C}$ that mimics the Ising ``interaction" between $s^r$ and $s^{r'}$, and each vertex also carries a bias $\widehat{\alpha}_r\in \mathbb{C}$ that mimics ``an external magnetic field". The ``energy" for such an Ising model is:
\begin{equation}\label{eq.IsingModel}
\Ecal_{\mathrm{BM}}(\{s_r\})= \sum_{r,r'} \widehat{W}_{rr'} s^r s^{r'} + \sum_{r} \widehat{\alpha}_r s^r,
\end{equation}
where the summation runs over all spins. In turn, a BM can be efficiently represented by a graph: (1) the vertices of the graph represent the spins $\{s^r\}$; (2) the nonzero weight of $s^r$ and $s^{r'}$ is represented by the link connecting $s^r$ and $s^{r'}$. The set of spins is divided into two disjoint subsets: the visible spins whose set is denoted by $V$ and the hidden spins denoted by $H$. We denote $g^r$ the visible spins and $h^s$ the hidden spins. Using these notations, the BM state is defined as:
\begin{equation}
|\BMS\rangle=\mathcal{C}\sum_{\substack{\{g^r\}\\ r\in V}} \sum_{\substack{\{h^{s}\}\\ s\in H}} \exp\bigg(-\Ecal_{\mathrm{BM}}(\{h^s\},\{g^r\}) \bigg) |\{g^r\}\rangle,
\end{equation} 
where $\mathcal{C}$ is a normalization constant that we will drop for simplicity. The states $|\{g^r\}\rangle$ are the basis states over the visible spins, i.e., a given $|\{g^r\}\rangle$ is the direct product of Pauli $Z$ eigenstates with eigenvalues $\{(-1)^{g^r}\}$. The ``energy" terms in $\Ecal_{\mathrm{BM}}(\{h^s\},\{g^r\})$ can be split into
\begin{equation}\label{eq.EnergyTermsSplit}
\begin{split}
\Ecal_{\mathrm{BM}}(\{h^s\},\{g^r\})=& \sum_{\substack{r,r'\in V}} R_{rr'}g^r g^{r'} + \sum_{\substack{s,s'\in H}} S_{ss'} h^sh^{s'} \\
&+\sum_{\substack{r\in V\\s\in H}}W_{rs}g^r h^s + \sum_{\substack{r\in V}} \beta_{r} g^r+\sum_{\substack{s\in H}} \alpha_{s} h^s,
\end{split}
\end{equation}
where $W_{rs}, R_{rr'}, S_{ss'}\in \mathbb{C}$ are the weights between visible and hidden, visible and visible, hidden and hidden spins respectively. $\beta_{r}\in \mathbb{C}$ is the bias of the visible spin $g^r$, and $\alpha_{s}\in \mathbb{C} $ is the bias of the hidden spin $h^s$. 

A \emph{restricted} Boltzmann machine (RBM) state is a special BM state satisfying
\begin{equation}
R_{rr'}=0, \quad \forall\; r,r' \in V;	\quad
S_{ss'}=0, \quad \forall\; s,s' \in H.
\end{equation} 
Thus an RBM state reads
\begin{equation}\label{Eq.RBMstate}
\begin{split}
&|\RBMS\rangle=\sum_{\substack{\{g^r\}\\ r\in V}} \sum_{\substack{\{h^s\}\\ s\in H}} \exp\bigg(-\Ecal_{\mathrm{RBM}}(\{h^s\},\{g^r\})\bigg) |\{g^r\}\rangle
\end{split}
\end{equation}
with 
\begin{equation}
\Ecal_{\mathrm{RBM}}(\{h^s\},\{g^r\})=\sum_{\substack{r\in V\\s\in H}}W_{rs}g^r h^s +\sum_{\substack{r\in V}} \beta_{r} g^r+\sum_{\substack{s\in H}} \alpha_{s} h^s .
\end{equation}
In this article, we will consider RBM states for 1D translational invariant systems. For this reason, we use $r, s$ to label the \emph{unit cells}, and $i,a$ to label the visible spin and hidden spins within a u］nit cell (which are dubbed ``orbitals'') respectively. We further require the RBM to be finitely connected, and by properly enlarging the unit cell, we can always choose the RBM to be nearest unit cell connected. Due to the requirement of translational invariance and nearest neighbor connectivity, we label the visible spins, the hidden spins, the weights and the biases as follows:
\begin{enumerate}
	\item The visible spins within the unit cell at $r$ are labeled by $g^r_i$ where $i=1, \ldots , q$ labels the orbitals within the unit cell. $q$ is the number of visible spins within each unit cell. 
	\item The hidden spins are divided into two categories:
	\begin{enumerate}
		\item $h^r_{a}$, $a\in\{1, \ldots , M\}$, labels the hidden spins connecting to the visible spins from the unit cell at $r-1$ and those from the unit cell at $r$, i.e., $h^r_{a}$ connects to both $\{g^{r-1}_i\}$ and $\{g^r_i\}$. $M$ is the total number of such hidden spins within the unit cell. Since we assume that the RBM is nearest unit cell connected, $h^r_{a}$ does not connect to the visible spins of another unit cell. We will dub such hidden spins as type-$h$ hidden spins.
		\item $\widetilde{h}^r_{b}$, $b\in\{1, \ldots , \widetilde{M}\}$, labels the hidden spins connecting to the visible spins within the unit cell at $r$, i.e., $\widetilde{h}^r_{b}$ only connects to $\{g^r_i\}$. $\widetilde{M}$ is the total number of such hidden spins within the unit cell. We will dub such hidden spins as type-$\widetilde{h}$ hidden spins. 
	\end{enumerate}
	\item The weight connecting $h^r_{a}$ and $g^r_{i}$ is labeled by $A_{ia}$, $i\in\{1, \ldots , q\}, a\in\{1, \ldots , M\}$. 
	\item The weight connecting $h^r_{a}$ and $g^{r-1}_{i}$ is labeled by $B_{ia}$, $i\in\{1, \ldots , q\}, a\in\{1, \ldots , M\}$. 
	\item The weight connecting $\widetilde{h}^r_{b}$ and $g^{r}_{i}$ is labeled by $\widetilde{C}_{ib}$, $i\in\{1, \ldots , q\}, b\in\{1, \ldots , \widetilde{M}\}$. 
	\item The bias of the visible spin $g^r_{i}$ is $\beta_i$, $i\in\{1, \ldots , q\}$. 
	\item The bias of the hidden spin $h^r_{a}$ is $\alpha_a$, $a\in\{1, \ldots , M\}$. 
	\item The bias of the hidden spin $\widetilde{h}^r_{b}$ is $\widetilde{\alpha}_b$, $b\in\{1, \ldots , \widetilde{M}\}$. 
\end{enumerate}
Due to translational invariance, the weights $A_{ia}$, $B_{ia}$, $\widetilde{C}_{ib}$ and the biases $\beta_i$, $\alpha_a$ and $\widetilde{\alpha}_b$ are all independent of the position of the unit cell $r$. We have distinguished the hidden spins into type-$h$ and type-$\widetilde{h}$ because, as will be explained in Sec.~\ref{Sec.relationtoMPS},  the hidden spins of type-$h$ contribute to the entanglement, while those of type-$\widetilde{h}$ do not. Correspondingly, we distinguish the weights $A_{ia}$ which connect the visible spin $g^r_i$ to the hidden spins of type-$h$, i.e., $h^r_{a}$,  and $\widetilde{C}_{ib}$ which connect the visible spin $g^r_i$ to the hidden spins of type-$\widetilde{h}$, i.e., $\widetilde{h}^r_{b}$.
In Fig.~\ref{Fig.RBMGeneral}, we show an example of such an RBM state with $q=3, M=2$ and $\widetilde{M}=2$. The visible spins (i.e., $g^r_i$) are represented by red circles. The hidden spins connecting to the visible spins from the neighboring unit cells (i.e., $h^r_{a}$) are represented by the rectangles and the hidden spins connecting to the visible spins from a single unit cell (i.e., $\widetilde{h}^r_{b}$) are represented by triangles. 

\begin{figure}[t]
	\centering
	\includegraphics[width=1\columnwidth]{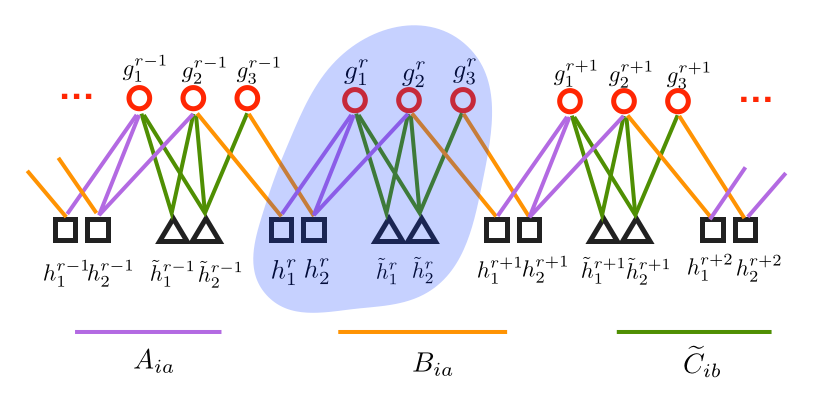}
	\caption{An example of RBM state corresponding to $q=3, M=2, \widetilde{M}=2$. The red circles represent visible spins. The black rectangles are the hidden spins connecting visible spin belonging to different unit cells, which are linked to the purple and orange lines representing the weights $A_{ia}$ and $B_{ia}$ respectively. The black triangles are the hidden spins connecting visible spins within the same unit cell, which are linked to the green lines representing the weights $\widetilde{C}_{ib}$. The blue region represents a unit cell. Notice that the nonzero weights are only between the hidden spins and the visible spins. }
	\label{Fig.RBMGeneral}
\end{figure}

With the notations introduced above, a translational invariant and nearest neighbor connected RBM state is
\begin{equation}\label{Eq.RBMState}
\begin{split}
|\mathrm{RBM}\rangle =& \sum_{\substack{\{g_i^r\}}} \sum_{\substack{\{h^r_a, \widetilde{h}^r_b\}}} \exp\bigg(-\Ecal_{\mathrm{RBM}}(\{h^r_a, \widetilde{h}^r_b\}, \{g^r_i\})\bigg)|\{g^r_i\}\rangle,
\end{split}
\end{equation}
with 
\begin{equation*}
\begin{split}
\Ecal_{\mathrm{RBM}}=&\sum_{r}\sum_{i=1}^q \left[\sum_{a=1}^{M}(A_{ia}g^r_i h^r_a+ B_{ia}g^r_i h^{r+1}_a)+\sum_{b=1}^{\widetilde{M}} \widetilde{C}_{ib}g^r_i \widetilde{h}^r_b\right]\\&+\sum_{r}\left[\sum_{i=1}^{q} \beta_{i} g^r_i+ \sum_{a=1}^{M} \alpha_{a} h^r_a+ \sum_{b=1}^{\widetilde{M}} \widetilde{\alpha}_b \widetilde{h}^r_b\right].
\end{split}
\end{equation*}

\subsection{Relation to MPS}
\label{Sec.relationtoMPS}

The RBM state defined by Eq.~\eqref{Eq.RBMState} can be cast into an MPS by mapping the hidden spins of the RBM to the virtual indices of the MPS. We name such MPS an \emph{RBM-MPS}. Specifically, Eq.~\eqref{Eq.RBMState} can be rewritten as
\begin{equation}\label{Eq.RBM2MPS}
|\mathrm{RBM}\rangle=\sum_{\substack{\{g_i^r\}}} \Tr \bigg( \prod_{r}T^{g^r_1\ldots g^r_q}\bigg) |\{g^r_i\}\rangle,
\end{equation} 
where
\begin{widetext}
\begin{equation}\label{eq.RBM2MPS}
T^{g^r_1 \ldots g^r_q}_{h^r_1 \ldots h^r_M, h^{r+1}_1\ldots h^{r+1}_M}
=e^{-\sum_{i,a=1}^{q, M}\left(A_{ia}g^r_i h^r_a+ B_{ia}g^r_i h^{r+1}_a\right) -\sum_{i=1}^q \beta_{i} g^r_i-\sum_{a=1}^M \alpha_{a} h^r_a } \sum_{\{\widetilde{h}^r_b\}} e^{- \sum_{i,b=1}^{q,\widetilde{M}} \widetilde{C}_{ib}g^r_i \widetilde{h}^r_b- \sum_{b=1}^{\widetilde{M}} \widetilde{\alpha}_b \widetilde{h}^r_b}.
\end{equation}
\end{widetext}
The bond dimension of the RBM-MPS is determined by the number of type-$h$ hidden spins, i.e., $M$. Hence only the hidden spins of type-$h$ contribute to the entanglement, while those of type-$\widetilde{h}$ do not. The optimal $M$ will be determined in Sec.~\ref{Sec.MPSSC}.  For instance, if each $h^r_a\in\{0 ,1\}$ is $\Z_2$ valued, the bond dimension is $2^{M}$. The tensor $T$ satisfies two useful properties: 
\begin{theorem}\label{theoremII.rank1ofRBM}
(a) $T^{g^r_1\ldots g^r_q}$ in Eq.~\eqref{eq.RBM2MPS} is either strictly zero or all its matrix elements are non-vanishing.
(b) If $T^{g^r_1\ldots g^r_q}$ is non-vanishing, it is of rank 1. If $T^{g^r_1\ldots g^r_q}$ vanishes, it is of rank 0. 
\end{theorem}

\begin{proof}
To prove (a), we notice that each matrix element of $T^{g^r_1\ldots g^r_q}$ is a common multiplicative factor $\sum_{\{\widetilde{h}^r_b\}} e^{- \sum_{i,b=1}^{q,\widetilde{M}} \widetilde{C}_{ib}g^r_i \widetilde{h}^r_b- \sum_{b=1}^{\widetilde{M}} \widetilde{\alpha}_b \widetilde{h}^r_b}$ independent of the hidden spins $\{h^r_a, h^{r+1}_{a}\}$ for all $a$, times a strictly nonzero expression of $\{h^r_a, h^{r+1}_{a}\}$: $e^{-\sum_{i,a=1}^{q, M}\left(A_{ia}g^r_i h^r_a+ B_{ia}g^r_i h^{r+1}_a\right) -\sum_{i=1}^q \beta_{i} g^r_i-\sum_{a=1}^M \alpha_{a} h^r_a }$. If the common multiplicative factor is zero then $T^{g^r_1\ldots g^r_q}$ vanishes. If the common multiplicative factor is nonzero, all matrix elements are non-vanishing.  

To prove (b), we observe that, when the matrix elements of $T^{g^r_1\ldots g^r_q}$ are non-vanishing, the ratio
\begin{equation}
\frac{T^{g^r_1 \ldots g^r_q}_{h^r_1 \ldots h^r_M, h^{r+1}_1\ldots h^{r+1}_M}}{T^{g^r_1 \ldots g^r_q}_{h^r_1 \ldots h^r_M, h'^{r+1}_1\ldots h'^{r+1}_M}}
\end{equation}
is independent of $h^r_1 \ldots h^r_M$, for any $h^{r+1}_1\ldots h^{r+1}_M$ and $h'^{r+1}_1\ldots h'^{r+1}_M$. Hence any two rows of the matrix $T^{g^r_1\ldots g^r_N}$ are proportional to each other, and thus the matrix is of rank 1. When $T^{g^r_1\ldots g^r_q}$ vanishes, by definition, it is of rank 0. 
\end{proof}
We emphasize that the form of the RBM depends on the visible spin basis choice. If we perform a basis transformation on the visible spins, this form, in general,  will no longer be preserved. 
Indeed, the rank of the MPS matrices after the basis rotation can be higher than 1 for a generic local unitary transformation. However, if the new MPS can be casted into an RBM with the same connection range, the rank of the RBM-MPS should be 1 by Theorem \ref{theoremII.rank1ofRBM}. This leads to a contradiction. 
Hence in this paper, we consider the RBM and the MPS under a particular choice of visible (physical) spin basis. The rank one condition refers to the MPS matrices with \emph{fixed} physical indices being of rank 1.  

Since the non-vanishing matrices of the RBM-MPS are of rank 1, it is natural to ask if the reverse statement also holds true, i.e., whether an MPS can be expressed as an RBM-MPS if the non-vanishing MPS matrices are of rank 1. In the rest of this article, we study this problem in the context of stabilizer codes. We conjecture that if the non-vanishing MPS matrices of the ground state of a translational invariant stabilizer code are of rank 1, such a ground state can also be found as an RBM state. In Sec.~\ref{Sec.MPSSC}, we first determine the condition for the non-vanishing MPS matrices of a stabilizer code to be of rank 1. In Sec.~\ref{Sec.BMSofQSC}, we give an algorithm to generate the RBM state for a large class of models (the cocycle models) whose MPS matrices are of rank 1. 

\section{Matrix Product State of a Stabilizer Code}
\label{Sec.MPSSC}

In this section, we present our algorithm to find an MPS for the ground state of a translational invariant stabilizer code. Along the way, we derive the formula which enables us to read the minimal bond dimension and the upper bound of the rank of the MPS matrices from the Hamiltonian terms. We illustrate our algorithm using an example, the $ZZXZZ$ model, and then discuss of the general case. Each steps of the algorithm is proven  in App.~\ref{app.CorrelationAndTransferMatrix}, \ref{app.Deriving0}, \ref{app.Deriving}, \ref{app.Ucommutation} and \ref{app.Linearizing}. We derive and prove our results based on the following assumptions throughout the paper:
\begin{assumption}\label{Assumption1}
We only consider the translational invariant stabilizer codes that have a unique ground state with PBC.
\end{assumption}
\begin{assumption}\label{Assumption2}
The MPS matrices of the translational invariant stabilizer codes become independent of the system size for sufficiently large system sizes. 
\end{assumption}
\begin{assumption}\label{Assumption3}
The MPS matrices are independent of the boundary condition. In other words, in the bulk, the MPS matrices for PBC are the same as those for the open boundary condition (OBC). 
\end{assumption}

We begin by stating the notations of spin chains. We  mainly consider  spin models defined on a finite chain with $L$ unit cells and PBC. Each unit cell contains $q$ spin-$\frac{1}{2}$'s. For the $i$-th  spin $(i=1, \ldots , q)$ in the $r$-th unit cell $(r=0, \ldots , L-1)$, we associate a two dimensional Hilbert space spanned by $|g^r_i\rangle$, where $g^r_i=0, 1$ corresponds to spin up and spin down respectively. $|g^r_i\rangle$ satisfies
\begin{equation}
Z^r_i \ket{g_{i}^{r}}=(-1)^{g^r_i}\ket{g_{i}^{r}},	\quad
X^r_i \ket{g_{i}^{r}}=\ket{1-g_{i}^{r}},
\end{equation} 
where $Z^r_i$ and $X^r_i$ are Pauli $Z$ and $X$ matrices acting on $\ket{g_{i}^{r}}$. 

\subsection{An Example of Stabilizer Codes: $ZZXZZ$ Model}\label{Sec.Example0}

To define the $ZZXZZ$ model, we place $3$ physical spin-$\frac{1}{2}$'s in each unit cell, i.e., $q=3$. (We choose $q=3$ since it fits naturally into the discussion of general cocycle models. See App.~\ref{app.SPT} for details.)
We introduce three sets of commuting operators $\mathcal{O}_\alpha^r$ $(\alpha=1,2,3)$ defined as
\begin{equation}\label{Eq.Ori}
\begin{split}
\Ocal_1^r&=Z^{r}_2 Z^{r}_3 X^{r+1}_1 Z^{r+1}_2 Z^{r+1}_3\\
\Ocal_2^r&=Z^{r}_3 Z^{r+1}_1 X^{r+1}_2 Z^{r+1}_3 Z^{r+2}_1\\
\Ocal_3^r&=Z^r_1 Z^r_2 X^r_3 Z^{r+1}_1 Z^{r+1}_2,
\end{split}
\end{equation}
where $r$ is defined modulo $L$ due to PBC. Using these operators, the Hamiltonian reads
\begin{equation}\label{eq.ZZXZZ}
\begin{split}
H_{ZZXZZ}=-\sum_{r=0}^{L-1} (\mathcal{O}_1^r+ \Ocal_2^r+\Ocal_3^r).
\end{split}
\end{equation}
All the terms in the Hamiltonian Eq.~\eqref{eq.ZZXZZ} mutually commute, and have eigenvalues $\pm 1$. 
Thus its ground state is the common positive eigenstate of $\Ocal^r_\alpha$ for any $r$ and $\alpha$, i.e., 
\begin{equation}\label{Eq.StabilizerEqsZZXZZ}
\Ocal^r_\alpha|\mathrm{GS}\rangle=|\mathrm{GS}\rangle,\quad \alpha=1,2,3, \; r=0, 1, \ldots , L-1.
\end{equation}
For example, one can construct $|\mathrm{GS}\rangle$ as 
\begin{equation}\label{Eq.GSZZXZZ}
|\mathrm{GS}\rangle=\prod_{r=0}^{L-1}\prod_{\alpha=1}^3\bigg(\frac{1+\Ocal^r_\alpha}{2}\bigg)|\mathbf{0}\rangle,
\end{equation}
where
\begin{equation}
\ket{\mathbf{0}}=\bigotimes_{r=0}^{L-1}\bigotimes_{i=1}^{3}\ket{0_{i}^{r}}.
\end{equation}
It is straightforward to verify that the $|\mathrm{GS}\rangle$ in Eq.~\eqref{Eq.GSZZXZZ} satisfies Eq.~\eqref{Eq.StabilizerEqsZZXZZ}.

Our goal in this section is to express the ground state $\ket{\mathrm{GS}}$ as an MPS
\begin{equation}\label{Eq.GS}
|\mathrm{GS}\rangle=\sum_{\{g^r_i\}}\Tr\bigg(\prod_{r=0}^{L-1}T^{g^r_1g^r_2g^r_3}\bigg)|\{g^r_i\}\rangle,
\end{equation}
where 
\begin{equation}
|\{g^r_i\}\rangle\equiv \bigotimes_{r=0}^{L-1} \bigotimes_{i=1}^3|g^r_i\rangle.
\end{equation} 
The matrix $T^{g^r_1g^r_2g^r_3}$ is labeled with three physical indices $g^r_1, g^r_2$ and $g^r_3$ in the $r$-th unit cell. The left and right virtual indices of $T^{g^r_1g^r_2g^r_3}$ and matrix elements will be solved in Sec.~\ref{Sec.Example}. The product of two $T$ matrices amounts to contracting the pair of virtual indices between them. The coefficient of $|\{g^r_i\}\rangle$ is determined by contracting all virtual indices with the same configuration of physical spins $\{g^r_i\}$. 
\begin{figure}
	\centering
	\includegraphics[width=0.8\columnwidth]{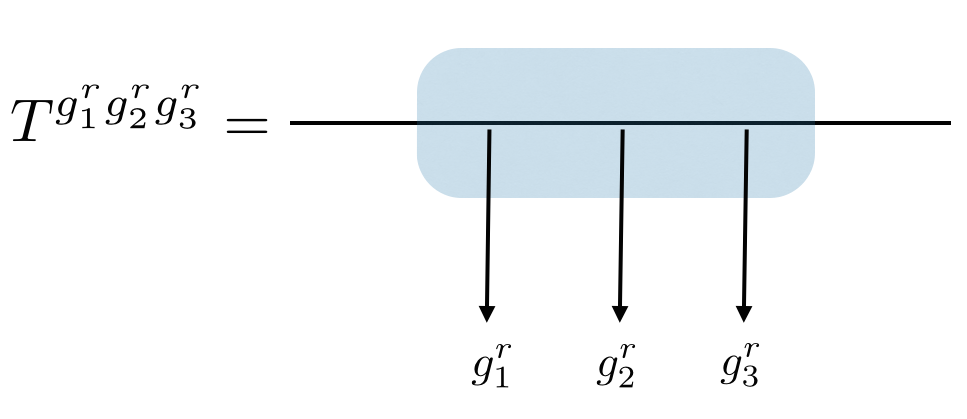}
	\caption{A graphical representation of the matrix $T^{g^r_1g^r_2g^r_3}$. We denote each physical index by an arrow. The shaded region represents a unit cell, and the virtual left and right indices are represented by the horizontal line. The virtual indices are not explicitly shown here. }
	\label{FigZZXZZMPS}
\end{figure}
\noindent The matrix $T^{g^r_1g^r_2g^r_3}$ is graphically represented in Fig.~\ref{FigZZXZZMPS}. Some notations and general properties of MPS for  stabilizer codes are given in App.~\ref{app.MPS}.

\subsection{MPS for the $ZZXZZ$ Model}
\label{Sec.Example}

To derive the MPS for the ground state $|\mathrm{GS}\rangle$ of the $ZZXZZ$ model Eq.~\eqref{eq.ZZXZZ}, we start with Eq.~\eqref{Eq.StabilizerEqsZZXZZ}. $\Ocal^r_i$ acts on the basis $|\{g^r_i\}\rangle$ in each summand of $|\mathrm{GS}\rangle$ in Eq.~\eqref{Eq.GS}. By re-arranging the summation, we derive the action of $\Ocal^r_i$ on the $T$-matrices. For example, let us consider the action of $\Ocal^r_1$,
\begin{widetext}
\begin{equation}\label{Eq.Or1T}
\begin{split}
\Ocal^r_1|\mathrm{GS}\rangle
&=\sum_{\{g^{r'}_i\}}\Tr\bigg(\prod_{r'=0}^{L-1}T^{g^{r'}_1g^{r'}_2g^{r'}_3}\bigg)\Ocal^{r}_1|\{g^{r'}_i\}\rangle\\
&=\sum_{\{g^{r'}_i\}}\Tr\bigg(\prod_{r'=0}^{L-1}T^{g^{r'}_1g^{r'}_2g^{r'}_3}\bigg)(-1)^{g^{r}_2+g^{r}_3+g^{r+1}_2+g^{r+1}_3}\bigg|\{g^{r'}_i\}|_{r'\le r},\{(1-g^{r+1}_1)g^{r+1}_2g^{r+1}_3\},\{g^{r^{\prime\prime}}_i\}|_{r^{\prime\prime}>r+1}\bigg\rangle\\
&=\sum_{\{\widehat{g}^{r'}_i\}}\Tr\bigg(\Big(\prod_{r'<r}T^{\widehat{g}^{r'}_1\widehat{g}^{r'}_2\widehat{g}^{r'}_3}\Big)\cdot T^{\widehat{g}^{r}_1\widehat{g}^{r}_2\widehat{g}^{r}_3}\cdot T^{(1-\widehat{g}^{r+1}_1)\widehat{g}^{r+1}_2\widehat{g}^{r+1}_3}
\cdot\Big(\prod_{r'>r+1}T^{\widehat{g}^{r'}_1\widehat{g}^{r'}_2\widehat{g}^{r'}_3}\Big)\bigg)(-1)^{\widehat{g}^{r}_2+\widehat{g}^{r}_3+\widehat{g}^{r+1}_2+\widehat{g}^{r+1}_3}|\{\widehat{g}^{r'}_i\}\rangle\\
&\equiv \sum_{\{\widehat{g}^{r'}_i\}}\Tr\bigg(\Big(\prod_{r'<r}T^{\widehat{g}^{r'}_1\widehat{g}^{r'}_2\widehat{g}^{r'}_3}\Big)\cdot \Ocal^r_1\circ\Big( T^{\widehat{g}^{r}_1\widehat{g}^{r}_2\widehat{g}^{r}_3}\cdot T^{\widehat{g}^{r+1}_1\widehat{g}^{r+1}_2\widehat{g}^{r+1}_3}\Big)\cdot \Big(\prod_{r'>r+1}T^{\widehat{g}^{r'}_1\widehat{g}^{r'}_2\widehat{g}^{r'}_3}\Big)\bigg)|\{\widehat{g}^{r'}_i\}\rangle.
\end{split}
\end{equation}
In the second equality, we use the definition of $\Ocal^r_1$ in Eq.~\eqref{Eq.Ori}, and in the last equality we \emph{defined} the action of $\Ocal^r_1$ on $T^{g^{r}_1g^{r}_2g^{r}_3}\cdot T^{g^{r+1}_1g^{r+1}_2g^{r+1}_3}$, via 
\begin{equation}\label{Eq.Or1TT}
\Ocal^r_1\circ\Big(T^{g^{r}_1g^{r}_2g^{r}_3}\cdot T^{g^{r+1}_1g^{r+1}_2g^{r+1}_3}\Big) 
\equiv (-1)^{g^{r}_2+g^{r}_3+g^{r+1}_2+g^{r+1}_3}\Big(T^{g^{r}_1g^{r}_2g^{r}_3}\cdot T^{(1-g^{r+1}_1)g^{r+1}_2g^{r+1}_3}\Big).
\end{equation}
For $\Ocal^r_2$ and $\Ocal^r_3$, we similarly define
\begin{equation}\label{Eq.Or23TT}
\begin{split}
&\Ocal^r_2\circ\Big(T^{g^{r}_1g^{r}_2g^{r}_3}\cdot T^{g^{r+1}_1g^{r+1}_2g^{r+1}_3}\cdot T^{g^{r+2}_1g^{r+2}_2g^{r+2}_3}\Big)
\equiv(-1)^{g^{r}_3+g^{r+1}_1+g^{r+1}_3+g^{r+2}_1} \Big(T^{g^{r}_1g^{r}_2g^{r}_3}\cdot T^{g^{r+1}_1(1-g^{r+1}_2)g^{r+1}_3}\cdot T^{g^{r+2}_1g^{r+2}_2g^{r+2}_3}\Big)	\\
&\Ocal^r_3\circ\Big( T^{g^r_1g^r_2g^r_3}\cdot T^{g^{r+1}_1g^{r+1}_2g^{r+1}_3}\Big) \equiv(-1)^{g^{r}_1+g^{r}_2+g^{r+1}_1+g^{r+1}_2}\Big( T^{g^r_1g^r_2(1-g^r_3)}\cdot T^{g^{r+1}_1g^{r+1}_2g^{r+1}_3}\Big).
\end{split}
\end{equation}
Using the definitions Eqs.~\eqref{Eq.Or1TT} and \eqref{Eq.Or23TT}, we find that a sufficient condition for the stabilizer condition Eq.~\eqref{Eq.StabilizerEqsZZXZZ} is 
\begin{equation}\label{Eq.ConstraintTmatrices}
\begin{split}
&\Ocal^r_1\circ\Big(T^{g^{r}_1g^{r}_2g^{r}_3}\cdot T^{g^{r+1}_1g^{r+1}_2g^{r+1}_3}\Big)=\Big(T^{g^{r}_1g^{r}_2g^{r}_3}\cdot T^{g^{r+1}_1g^{r+1}_2g^{r+1}_3}\Big)\\
&\Ocal^r_2\circ\Big(T^{g^{r}_1g^{r}_2g^{r}_3}\cdot T^{g^{r+1}_1g^{r+1}_2g^{r+1}_3}\cdot T^{g^{r+2}_1g^{r+2}_2g^{r+2}_3}\Big)=\Big(T^{g^{r}_1g^{r}_2g^{r}_3}\cdot T^{g^{r+1}_1g^{r+1}_2g^{r+1}_3}\cdot T^{g^{r+2}_1g^{r+2}_2g^{r+2}_3}\Big)\\
&\Ocal^r_3\circ\Big( T^{g^r_1g^r_2g^r_3}\cdot T^{g^{r+1}_1g^{r+1}_2g^{r+1}_3}\Big)=\Big(T^{g^r_1g^r_2g^r_3}\cdot T^{g^{r+1}_1g^{r+1}_2g^{r+1}_3}\Big).
\end{split}
\end{equation}
\end{widetext}
We prove in App.~\ref{app.CorrelationAndTransferMatrix} and \ref{app.Deriving0} that Eq.~\eqref{Eq.ConstraintTmatrices} is also a \emph{necessary} condition for Eq.~\eqref{Eq.StabilizerEqsZZXZZ}. A graphical representation of these equations is given in Fig.~\ref{FigZZXZZEqs}. 

\begin{figure}[t]
	\centering
	\includegraphics[width=1\columnwidth]{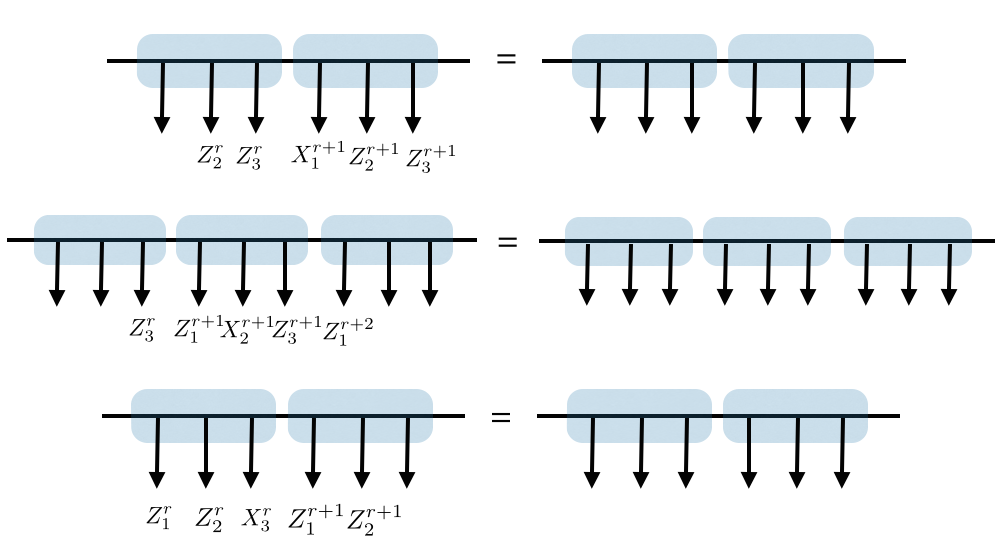}
	\caption{A graph representation of Eq.~\eqref{Eq.StabilizerEqsZZXZZ}. }
	\label{FigZZXZZEqs}
\end{figure}

We now construct a solution of $T$ matrices from Eq.~\eqref{Eq.ConstraintTmatrices}. It is difficult to solve Eq.~\eqref{Eq.ConstraintTmatrices} directly, as it is a set of nonlinear equations of the $T$ matrices. In the following, we will derive a new set of equations equivalent to Eq.~\eqref{Eq.ConstraintTmatrices}, which are linear in the $T$ matrices and only contain the matrices in the $r$-th unit cell.

The idea to derive the equations linear in the $T$ matrices is to decompose the Hamiltonian terms $\Ocal^r_\alpha$ into separate parts, where each part acts only on one unit cell, and then derive their action on a single $T$ matrix. (Notice that we are only allowed to cut in between the unit cells,  not  inside one unit cell.)  As a first step, we start by cutting the operator $\Ocal^r_\alpha$ into two parts: $\Ocal_1^r$ and $\Ocal_3^r$ can be cut into $\Ocal_1^r=\Lcal^r_{1, 1}\Rcal^r_{1, 1}$ and $\Ocal_3^r=\Lcal^r_{3, 1}\Rcal^r_{3, 1}$, such that the operators 
\begin{equation}\label{Eq.LR1}
\begin{split}
\Lcal^r_{1,1}&=I^{r}_{1}\otimes Z^{r}_2\otimes Z^{r}_3\\
\Rcal_{1,1}^r&= X^{r+1}_1 \otimes Z^{r+1}_2 \otimes Z^{r+1}_3
\end{split}
\end{equation}
and
\begin{equation}\label{Eq.LR3}
\begin{split}
\Lcal^r_{3,1}&=Z_1^r\otimes Z_2^r\otimes X_3^r\\
\Rcal_{3,1}^r&=Z^{r+1}_1 \otimes Z^{r+1}_2\otimes I^{r+1}_3
\end{split}
\end{equation}
only act on a given unit cell. The second subscript $\tau$ of $\Lcal^r_{\alpha, \tau}$ labels the position of bipartition of $\Ocal^r_\alpha$. Since $\Ocal_1^r$ and $\Ocal_3^r$ are supported on two unit cells,  there is only one way to cut them into two parts and hence $\tau$ only takes one value, i.e.,  $\tau=1$.  Such a unique bipartition is not possible for $\Ocal_2^r$ since $\Ocal^r_2$ is supported on 3 unit cells. Nevertheless, we can define two bipartitions as follows:
\begin{equation}
\Ocal_2^r=\Lcal^r_{2, 1}\Rcal^r_{2, 1},	\quad
\Ocal_2^r=\Lcal^r_{2, 2}\Rcal^r_{2, 2},
\end{equation} 
where $\Lcal^r_{2, \tau}, \Rcal^r_{2,\tau} (\tau=1,2)$ are
\begin{equation}\label{Eq.LR2}
\begin{split}
\Lcal_{2,1}^r&=I^{r}_1\otimes I^{r}_2\otimes Z^{r}_3\\
\Rcal_{2,1}^r&=Z^{r+1}_1\otimes X^{r+1}_2 \otimes Z^{r+1}_3 \otimes Z^{r+2}_1\otimes I^{r+2}_2\otimes I^{r+2}_3\\
\Lcal_{2,2}^r&=I^{r}_1\otimes I^{r}_2\otimes Z^{r}_3 \otimes Z^{r+1}_1\otimes X^{r+1}_2 \otimes Z^{r+1}_3\\ 
\Rcal_{2,2}^r&=Z^{r+2}_1\otimes I^{r+2}_2\otimes I^{r+2}_3.
\end{split}
\end{equation}
Notice that $\Rcal_{2,1}^r$ and $\Lcal_{2,2}^r$ have a  support over two unit cells while $\Rcal^r_{2,2}$ and $\Lcal^r_{2,1}$ have a support over a single unit cell.

\begin{figure}[t]
	\centering
	\includegraphics[width=1.0\columnwidth]{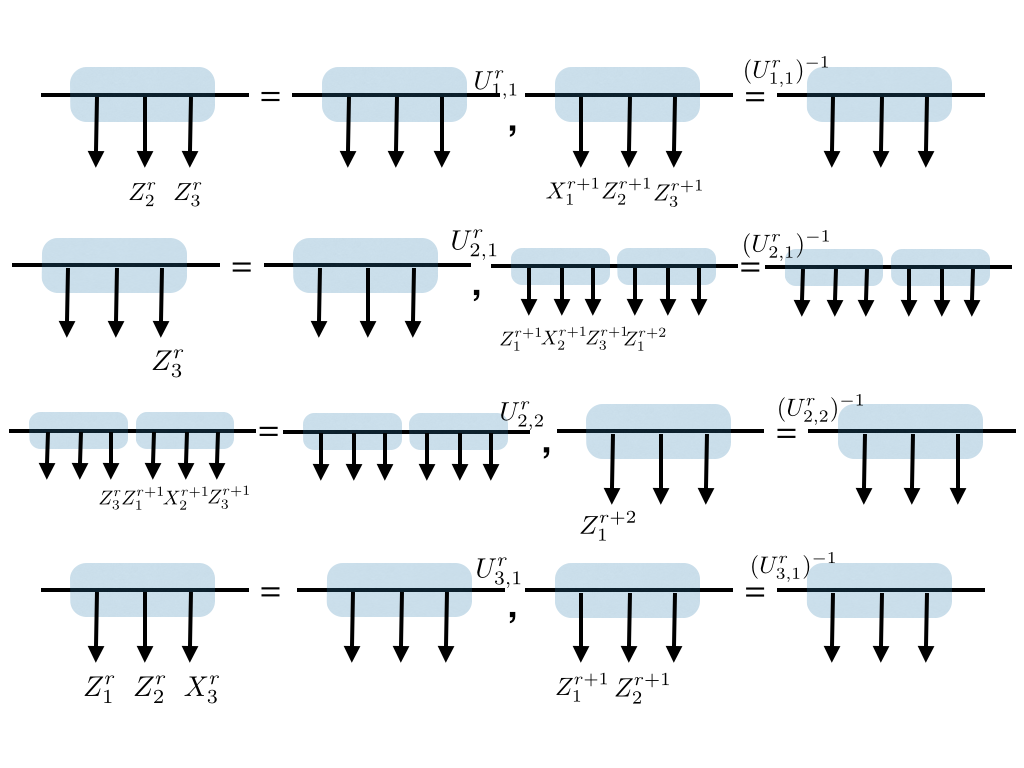}
	\caption{A graphical representation of Eqs.~\eqref{Eq.LRsingle1}, \eqref{Eq.LRsingle2} and \eqref{Eq.LRsingle3}. }
	\label{FigZZXZZEqsSingle}
\end{figure}

Let us consider the action of $\Lcal_{\alpha, \tau}^{r}$'s and $\Rcal_{\alpha, \tau}^{r}$'s on the $T$ matrices. First we focus on $\Ocal_1^r$. The product of two matrices $T^{g^{r}_1g^{r}_2g^{r}_3}\cdot T^{g^{r+1}_1g^{r+1}_2g^{r+1}_3}$ should be invariant under the combined action of $\Lcal_{1,1}^r \Rcal_{1,1}^r$, where $\Lcal_{1,1}^r $ acts only on $T^{g^{r}_1g^{r}_2g^{r}_3}$ while $\Rcal_{1,1}^r $ only on $T^{g^{r+1}_1g^{r+1}_2g^{r+1}_3}$. In App.~\ref{app.Deriving}, we prove in a general setting of stabilizer codes that the following condition is both necessary and sufficient of Eq.~\eqref{Eq.ConstraintTmatrices}: the action of $\Lcal^r_{1,1}$ on $T^{g^{r}_1g^{r}_2g^{r}_3}$ can be encoded by a transformation $U^r_{1,1}$ on the right virtual index of $T^{g^{r}_1g^{r}_2g^{r}_3}$, while the action of $\Rcal^r_{1,1}$ on $T^{g^{r+1}_1g^{r+1}_2g^{r+1}_3}$ can be encoded by the inverse of the same transformation $(U_{1,1}^{r})^{-1}$ on the left virtual index of $T^{g^{r+1}_1g^{r+1}_2g^{r+1}_3}$. 
Concretely, we have 
\begin{equation}\label{Eq.LRsingle1}
\begin{split}
\Lcal_{1,1}^r\circ T^{g^{r}_1g^{r}_2g^{r}_3}&=T^{g^{r}_1g^{r}_2g^{r}_3}\cdot U^r_{1,1}\\
\Rcal_{1,1}^r\circ T^{g^{r+1}_1g^{r+1}_2g^{r+1}_3}&=(U_{1,1}^{r})^{-1}\cdot T^{g^{r+1}_1g^{r+1}_2g^{r+1}_3},
\end{split}
\end{equation}
where $\circ$ represents the action on the physical indices (see Eqs.~\eqref{Eq.Or1TT} and \eqref{Eq.Or23TT}), and $\cdot$ represents the matrix multiplication (i.e., the contraction over the virtual indices). See Fig.~\ref{FigZZXZZEqsSingle}.  Similarly, for $\Lcal_{2,1}^r, \Rcal_{2,1}^r$ and $\Lcal_{2,2}^r,\Rcal_{2,2}^r$,
\begin{subequations}\label{Eq.LRsingle2}
\begin{align}
&\Lcal_{2,1}^r\circ T^{g^{r}_1g^{r}_2g^{r}_3}
=T^{g^{r}_1g^{r}_2g^{r}_3}\cdot U^r_{2,1}\label{Eq.LRsingle2a}\\
&\Rcal_{2,1}^r\circ (T^{g^{r+1}_1g^{r+1}_2g^{r+1}_3}\cdot T^{g^{r+2}_1g^{r+2}_2g^{r+2}_3})
\nonumber\\&~~~=(U_{2,1}^{r})^{-1}\cdot (T^{g^{r+1}_1g^{r+1}_2g^{r+1}_3}\cdot T^{g^{r+2}_1g^{r+2}_2g^{r+2}_3})\label{Eq.LRsingle2b}\\
&\Lcal_{2,2}^r\circ (T^{g^{r}_1g^{r}_2g^{r}_3}\cdot T^{g^{r+1}_1g^{r+1}_2g^{r+1}_3})
\nonumber\\&~~~=(T^{g^{r}_1g^{r}_2g^{r}_3}\cdot T^{g^{r+1}_1g^{r+1}_2g^{r+1}_3})\cdot U_{2,2}^r\label{Eq.LRsingle2c}\\
&\Rcal_{2,2}^r\circ T^{g^{r+2}_1g^{r+2}_2g^{r+2}_3}
\nonumber\\&~~~=(U^{r}_{2,2})^{-1}\cdot T^{g^{r+2}_1g^{r+2}_2g^{r+2}_3}. \label{Eq.LRsingle2d}
\end{align}
\end{subequations}
For $\Lcal_{3,1}^r, \Rcal_{3,1}^r$, we get
\begin{equation}\label{Eq.LRsingle3}
\begin{split}
\Lcal_{3,1}^r\circ T^{g^{r}_1g^{r}_2g^{r}_3}&=T^{g^{r}_1g^{r}_2g^{r}_3}\cdot U^r_{3,1}\\
\Rcal_{3,1}^r\circ T^{g^{r+1}_1g^{r+1}_2g^{r+1}_3}&=(U_{3,1}^{r})^{-1}\cdot T^{g^{r+1}_1g^{r+1}_2g^{r+1}_3}.
\end{split}
\end{equation}
Eqs.~\eqref{Eq.LRsingle1}, \eqref{Eq.LRsingle2} and \eqref{Eq.LRsingle3} are graphically represented in Fig.~\ref{FigZZXZZEqsSingle}.

Let us use the translational invariance and focus on the operators that act on the virtual indices between the $r$-th and $(r+1)$-th unit cells, i.e., $U_{1,1}^{r}, U_{2,1}^{r}, U_{2,2}^{r-1}$ and $U_{3,1}^r$. For Eqs.~\eqref{Eq.LRsingle1}, \eqref{Eq.LRsingle2} and \eqref{Eq.LRsingle3} being consistent, the virtual $U^{r'}_{\alpha', \tau'}$ and $(U^{r'}_{\alpha', \tau'})^{-1}$ operators on the right hand side (RHS) should satisfy the same commutation relations as the physical $\Lcal^{r'}_{\alpha', \tau'}$ and $\Rcal^{r'}_{\alpha', \tau'}$ operators on the left hand side (LHS) respectively. As a result, $\Lcal^{r'}_{\alpha', \tau'}$ and $\Rcal^{r'}_{\alpha', \tau'}$ share the same commutation relation.\footnote{This is because the commutation relations of $\Lcal$'s and those of the $\Rcal$'s are both related to  the commutation relations of $U$'s. } Here, 
\begin{equation}\label{Eq.set}
\begin{split}
\Lcal_{\alpha', \tau'}^{r'}&\in\{\Lcal_{1,1}^{r}, \Lcal_{2,1}^{r}, \Lcal_{2,2}^{r-1}, \Lcal_{3,1}^r\},\\
\Rcal_{\alpha', \tau'}^{r'}&\in \{\Rcal_{1,1}^{r}, \Rcal_{2,1}^{r}, \Rcal_{2,2}^{r-1}, \Rcal_{3,1}^r \}\\
U^{r'}_{\alpha', \tau'} &\in \{U_{1,1}^{r}, U_{2,1}^{r}, U_{2,2}^{r-1}, U_{3,1}^r\}
\end{split}
\end{equation}
This statement is proved in App.~\ref{app.Ucommutation} in a general setting of stabilizer codes. The commutation relations can be encoded using the compact notations:
\begin{equation}\label{eq.generalizedCliffordAlgebra}
\begin{split}
&\Lcal_{\alpha', \tau'}^{r'}\Lcal_{\alpha'', \tau''}^{r''} = (-1)^{\mathbf{t}_{(\alpha'\tau'),(\alpha''\tau'')}^{r',r''}} \Lcal_{\alpha'', \tau''}^{r''} \Lcal_{\alpha', \tau'}^{r'},	\\
&\Rcal_{\alpha', \tau'}^{r'}\Rcal_{\alpha'', \tau''}^{r''} = (-1)^{\mathbf{t}_{(\alpha'\tau'),(\alpha''\tau'')}^{r',r''}} \Rcal_{\alpha'', \tau''}^{r''} \Rcal_{\alpha', \tau'}^{r'},	\\
&U_{\alpha', \tau'}^{r'}U_{\alpha'', \tau''}^{r''} = (-1)^{\mathbf{t}_{(\alpha'\tau'),(\alpha''\tau'')}^{r',r''}} U_{\alpha'', \tau''}^{r''} U_{\alpha', \tau'}^{r'},	\\
\end{split}
\end{equation}
Since $\Rcal$ operators obey the same commutation relations as the $\Lcal$'s, we just focus on the $\Lcal$ operators. 
The coefficients $\mathbf{t}_{(\alpha'\tau'),(\alpha''\tau'')}^{r',r''}$ form an anti-symmetric $\mathbf{t}$ matrix, which under the basis $(\Lcal_{1,1}^{r}, \Lcal_{2,1}^{r}, \Lcal_{2,2}^{r-1}, \Lcal_{3,1}^r )^T$ is given by\footnote{We only focus on the commutation relations between the virtual $U$ operators which act on the same virtual bond, e.g. $\{U_{1,1}^{r}, U_{2,1}^{r}, U_{2,2}^{r-1}, U_{3,1}^r\}$. According to Eq.~\eqref{eq.generalizedCliffordAlgebra}, they have the same commutation relations as $\{\Lcal_{1,1}^{r}, \Lcal_{2,1}^{r}, \Lcal_{2,2}^{r-1}, \Lcal_{3,1}^r \}$.  These commutation relations constrain the representation of the virtual $U$ operators. The commutation relations between $\{\Lcal_{1,1}^{r}, \Lcal_{2,1}^{r}, \Lcal_{2,2}^{r-1}, \Lcal_{3,1}^r \}$ and other $\Lcal$ operators, e.g. $\Lcal^{r+1}_{1,1}$ do not constrain the representation of the virtual $U$ operators, thus we do not consider them. }
\begin{equation}\label{eq.exampletmatrix}
\mathbf{t}= \left( \begin{matrix}
0 & 0 & 1 & 1	\\
0 & 0 & 0 & 1	\\
-1 & 0 & 0 & 0	\\
-1 & -1 & 0 & 0	\\
\end{matrix}\right).
\end{equation}
We first determine the dimension of irreducible representation of the algebra Eq.~\eqref{eq.generalizedCliffordAlgebra} that $\{\Lcal_{1,1}^{r}, \Lcal_{2,1}^{r}, \Lcal_{2,2}^{r-1}, \Lcal_{3,1}^r \}$ and $\{U_{1,1}^{r}, U_{2,1}^{r}, U_{2,2}^{r-1}, U_{3,1}^r \}$ obey. It is convenient to introduce: 
\begin{equation}\label{Eq.CommutingPairs}
\begin{split}
\widetilde{\Lcal}_{1}=\Lcal_{3,1}^r,&~~\widetilde{U}_{1}=U_{3,1}^r,\\ 
\widetilde{\Lcal}_2=\Lcal_{2,1}^{r},&~~\widetilde{U}_2=U_{2,1}^{r} \\
\widetilde{\Lcal}_3=\Lcal_{2,2}^{r-1},&~~\widetilde{U}_3=U_{2,2}^{r-1} \\
\widetilde{\Lcal}_4=\Lcal_{1,1}^{r}\Lcal_{2,1}^{r},&~~\widetilde{U}_4=U_{1,1}^{r}U_{2,1}^{r}.
\end{split}
\end{equation}
The new operators $\{\widetilde{\Lcal}_{\alpha}, \widetilde{U}_{\alpha}, \alpha=1,2,3,4\}$ satisfy a simpler algebra, 
\begin{equation}
\begin{split}
\widetilde{\Lcal}_{\alpha}\widetilde{\Lcal}_{\beta}&=(-1)^{\widetilde{\mathbf{t}}_{\alpha\beta}}\widetilde{\Lcal}_{\beta}\widetilde{\Lcal}_{\alpha}\\
\widetilde{U}_{\alpha}\widetilde{U}_{\beta}&=(-1)^{\widetilde{\mathbf{t}}_{\alpha\beta}}\widetilde{U}_{\beta}\widetilde{U}_{\alpha},
\end{split}
\end{equation}
with 
\begin{equation}
\widetilde{\mathbf{t}}=
\begin{pmatrix}
0 & 1 & 0 & 0\\
-1 & 0 & 0 & 0\\
0 & 0 & 0 & 1\\
0 & 0 & -1 & 0
\end{pmatrix}.
\end{equation}
The algebra of $\{\widetilde{\Lcal}_{\alpha}\}$ (or $\{\widetilde{U}_\alpha\}$) is decoupled into two subalgebras, generated by $\{\widetilde{\Lcal}_1, \widetilde{\Lcal}_2\}$ (or $\{\widetilde{U}_1, \widetilde{U}_2\}$) and $\{\widetilde{\Lcal}_3, \widetilde{\Lcal}_4\}$ (or $\{\widetilde{U}_3, \widetilde{U}_4\}$) respectively. Each subalgebra has a two dimensional irreducible representation. Hence the total dimension of the irreducible representation of $\{\widetilde{\Lcal}_{\alpha}\}$ (or $\{\widetilde{U}_\alpha\}$) is $2\times 2=4$. Finally, noticing that the transformation between $\Lcal$'s and $\widetilde{\Lcal}$'s is invertible. The inverse transformation of  Eq.~\eqref{Eq.CommutingPairs} is 
\begin{eqnarray}
\begin{split}
\Lcal_{3,1}^r=\widetilde{\Lcal}_{1},&~~U_{3,1}^r=\widetilde{U}_{1},\\ 
\Lcal_{2,1}^{r}=\widetilde{\Lcal}_2,&~~U_{2,1}^{r}=\widetilde{U}_2 \\
\Lcal_{2,2}^{r-1}=\widetilde{\Lcal}_3,&~~U_{2,2}^{r-1}=\widetilde{U}_3 \\
\Lcal_{1,1}^{r}=\widetilde{\Lcal}_4(\widetilde{\Lcal}_2)^{-1},&~~U_{1,1}^{r}=\widetilde{U}_4 (\widetilde{U}_2)^{-1}
\end{split}
\end{eqnarray}
Since $\widetilde{\Lcal}$'s form a irreducible representation, the physical operators $\Lcal_{1,1}^{r+1}, \Lcal_{2,1}^{r+1}, \Lcal_{2,2}^{r}$ and $\Lcal_{3,1}^r$ (and thus by Eq.~\eqref{eq.generalizedCliffordAlgebra}, $U_{1,1}^{r+1}, U_{2,1}^{r+1}, U_{2,2}^{r}$ and $U_{3,1}^r$ as well) also provide a 4 dimensional irreducible representation. 

For the rest of the task, we need to first find a matrix representation of $U_{1,1}^{r}, U_{2,1}^{r}, U_{2,2}^{r-1}$ and $U_{3,1}^r$ satisfying the algebra Eq.~\eqref{eq.generalizedCliffordAlgebra}, and then solve for $T^{g^{r}_{1}g^{r}_2g^{r}_3}$ in Eqs.~\eqref{Eq.LRsingle1}, \eqref{Eq.LRsingle2} and \eqref{Eq.LRsingle3}. Since the irreducible representation of the algebra is 4-dimensional, the virtual $U$ operators should be $4\times 4$ matrices. 

The matrix representations for the $U$ operators are not unique.  If $U^{r'}_{\alpha', \tau'}$ satisfies the algebra Eq.~\eqref{eq.generalizedCliffordAlgebra} and $T^{g^{r'}_1 g^{r'}_2 g^{r'}_3}$ is the solution of Eqs.~\eqref{Eq.LRsingle1}, \eqref{Eq.LRsingle2} and \eqref{Eq.LRsingle3}, then $S\cdot U^{r'}_{\alpha', \tau'}\cdot S^{-1}$, with $S$ independent of $r',\alpha',\tau'$, also satisfies Eq.~\eqref{eq.generalizedCliffordAlgebra} and the corresponding solution for the $T$ matrices is given by $S\cdot T^{g^{r'}_1 g^{r'}_2 g^{r'}_3}\cdot S^{-1}$.
Hence without loss of generality, let us choose the virtual $U$ operators $\{U^{r'}_{\alpha', \tau'}\}$ to be:\footnote{We choose these $U$ operators in order to compare with the MPS matrices derived from the RBM states in Sec.~\ref{Sec.BMSofQSC}.}
\begin{equation}\label{eq.virtualoperators}
\begin{split}
U_{1,1}^{r}&=(X \otimes X)^r,	\\
U_{2,1}^{r}&=(I \otimes X)^r,	\\
U_{2,2}^{r-1}&=((-Y) \otimes I)^r.	\\
U_{3,1}^r&=(I \otimes (-Y))^r.	\\
\end{split}
\end{equation} 
where the superscript $r$ on the RHS indicates that the operators acts on the virtual bonds connecting the the $r$-th and $r+1$-th unit cell. 
We denote the corresponding MPS matrix elements as $T^{g^r_1g^r_2g^r_3}_{h_1h_2, h_3h_4}$ where $h_1, h_2\in\{0,1\}$ represent the left virtual indices and $h_3, h_4\in\{0,1\}$ represent the right virtual indices. In Eq.~\eqref{eq.virtualoperators}, the first Pauli matrices act on the virtual indices $h_1$ and $h_3$, while the second Pauli matrices act on the virtual indices $h_2$ and $h_4$. 

So far, we have only considered bipartition of the Hamiltonian terms $\Ocal^r_\alpha$. For the operators that have support over three or more unit cells, we can take the combinations of $\Lcal$ and $\Rcal$ operators so that $\Ocal^r_\alpha$ can be decomposed into a product of operators which only act within a single unit cell.  We enumerate the decompositions for all three types of operators:
\begin{eqnarray}\label{Eq.enu}
\begin{split}
\Ocal^r_{1,1}&=\Lcal^r_{1,1}\cdot \Rcal^r_{1,1}\\
\Ocal^r_2&=  \Lcal^r_{2,1}\cdot \left(\Rcal^r_{2,1}(\Rcal^r_{2,2})^{-1}\right)\cdot \Rcal^r_{2,2}\\
&= \Lcal^r_{2,1}\cdot \left((\Lcal^r_{2,1})^{-1}\Lcal^r_{2,2}\right)\cdot \Rcal^r_{2,2}\\
\Ocal^{r}_{3,1}&= \Lcal^r_{3,1} \cdot \Rcal^r_{3,1}
\end{split}
\end{eqnarray}
where all the terms on the RHS of Eq.~\eqref{Eq.enu}, in particular $\left(\Rcal^r_{2,1}(\Rcal^r_{2,2})^{-1}\right)$ and  $\left((\Lcal^r_{2,1})^{-1}\Lcal^r_{2,2}\right)$, are supported within one unit cell.   
In App.~\ref{app.Linearizing}, we show in a general setting of stabilizer codes, that Eq.~\eqref{Eq.LRsingle2} is equivalent to the following equations linear to the $T$ matrix in the $r$-th unit cell:
\begin{equation}\label{Eq.11}
\begin{split}
&\Lcal_{1,1}^r\circ T^{g^{r}_1g^{r}_2g^{r}_3}=T^{g^{r}_1g^{r}_2g^{r}_3}\cdot U^r_{1,1}\\
&\Rcal_{1,1}^{r-1}\circ T^{g^{r}_1g^{r}_2g^{r}_3}=(U_{1,1}^{r-1})^{-1}\cdot T^{g^{r}_1g^{r}_2g^{r}_3}\\
&\Lcal_{2,1}^r\circ T^{g^{r}_1g^{r}_2g^{r}_3}=T^{g^{r}_1g^{r}_2g^{r}_3}\cdot U^r_{2,1}\\&
((\Lcal^{r-1}_{2,1})^{-1}\Lcal_{2,2}^{r-1})\circ T^{g^{r}_1g^{r}_2g^{r}_3}= (U^{r-1}_{2,1})^{-1}\cdot T^{g^{r}_1g^{r}_2g^{r}_3}\cdot U_{2,2}^{r-1}\\&
(\Rcal_{2,1}^{r-1}(\Rcal^{r-1}_{2,2})^{-1})\circ T^{g^{r}_1g^{r}_2g^{r}_3}=(U^{r-1}_{2,1})^{-1}\cdot T^{g^{r}_1g^{r}_2g^{r}_3}\cdot U^{r-1}_{2,2}\\&
\Rcal_{2,2}^{r-2}\circ T^{g^{r}_1g^{r}_2g^{r}_3}=(U^{r-2}_{2,2})^{-1}\cdot T^{g^{r}_1g^{r}_2g^{r}_3}\\
&\Lcal_{3,1}^r\circ T^{g^{r}_1g^{r}_2g^{r}_3}=T^{g^{r}_1g^{r}_2g^{r}_3}\cdot U^r_{3,1}\\
&\Rcal_{3,1}^{r-1}\circ T^{g^{r}_1g^{r}_2g^{r}_3}=(U_{3,1}^{r-1})^{-1}\cdot T^{g^{r}_1g^{r}_2g^{r}_3},
\end{split}
\end{equation}
where by translational invariance
\begin{equation}
\begin{split}
&U_{1,1}^{r}=(X \otimes X)^r, ~~~U_{1,1}^{r-1}=(X \otimes X)^{r-1},	\\	
&U_{2,1}^{r}=(I \otimes X)^{r}, ~~~U_{2,1}^{r-1}=(I \otimes X)^{r-1},	\\
&U_{2,2}^{r-1}=((-Y) \otimes I)^{r}, ~~~U_{2,2}^{r-2}=((-Y) \otimes I)^{r-1},	\\
&U_{3,1}^r=(I \otimes (-Y))^{r}, ~~~U_{3,1}^{r-1}=(I \otimes (-Y))^{r-1}.	\\
\end{split}
\end{equation} 
More concretely, the equations in \eqref{Eq.11} are 
\begin{equation}\label{Eq.MPSequations}
\begin{split}
T^{g^r_1g^r_2g^r_3}_{h_1h_2, h_3h_4}(-1)^{g^r_2+g^r_3}&=T^{g^r_1g^r_2g^r_3}_{h_1h_2, (1-h_3)(1-h_4)}\\
T^{(1-g^{r}_1)g^{r}_2g^{r}_3}_{h_1h_2, h_3h_4}(-1)^{g^{r}_2+g^{r}_3}&=T^{g^{r}_1g^{r}_2g^{r}_3}_{(1-h_1)(1-h_2), h_3h_4}\\
T^{g^r_1g^r_2g^r_3}_{h_1h_2, h_3h_4}(-1)^{g^r_3}&=T^{g^r_1g^r_2g^r_3}_{h_1h_2, h_3(1-h_4)}\\
T^{g^{r}_1(1-g^{r}_2)g^{r}_3}_{h_1h_2, h_3h_4}(-1)^{g^{r}_1+g^{r}_3}&=-i T^{g^{r}_1g^{r}_2g^{r}_3}_{h_1(1-h_2), (1-h_3)h_4}(-1)^{1-h_3}\\
T^{g^{r}_1(1-g^{r}_2)g^{r}_3}_{h_1h_2, h_3h_4}(-1)^{g^{r}_1+g^{r}_3}&= -iT^{g^{r}_1g^{r}_2g^{r}_3}_{h_1(1-h_2), (1-h_3)h_4}(-1)^{1-h_3}\\
T^{g^{r}_1g^{r}_2g^{r}_3}_{h_1h_2, h_3h_4}(-1)^{g^{r}_1}&=-i T^{g^{r}_1g^{r}_2g^{r}_3}_{(1-h_1)h_2, h_3h_4}(-1)^{h_1}\\
T^{g^r_1g^r_2(1-g^r_3)}_{h_1h_2, h_3h_4}(-1)^{g^r_1+g^r_2}&=-i T^{g^r_1g^r_2g^r_3}_{h_1h_2, h_3(1-h_4)}(-1)^{1-h_4}\\
T^{g^{r}_1g^{r}_2g^{r}_3}_{h_1h_2, h_3h_4}(-1)^{g^{r}_1+g^{r}_2}&=-i T^{g^{r}_1g^{r}_2g^{r}_3}_{h_1(1-h_2), h_3h_4}(-1)^{h_2}.
\end{split}
\end{equation}
Thus we have derived a set of linear equations Eq.~\eqref{Eq.MPSequations} of the $T$-matrices  from the non-linear ones Eq.~\eqref{Eq.ConstraintTmatrices}. Solving Eq.~\eqref{Eq.MPSequations}, we obtain a solution up to a total scale factor:
\begin{equation}\label{Eq.MPSZZXZZ}
	\begin{split}
		T^{000}=
		\begin{pmatrix}
			1&1&1&1\\
			i&i&i&i\\
			i&i&i&i\\
			-1&-1&-1&-1
		\end{pmatrix},
		T^{001}=
		\begin{pmatrix}
			i&-i&i&-i\\
			-1&1&-1&1\\
			-1&1&-1&1\\
			-i&i&-i&i
		\end{pmatrix},\\
		T^{010}=
		\begin{pmatrix}
			-1&-1&1&1\\
			i&i&-i&-i\\
			-i&-i&i&i\\
			-1&-1&1&1
		\end{pmatrix},
		T^{011}=
		\begin{pmatrix}
			i&-i&-i&i\\
			1&-1&-1&1\\
			-1&1&1&-1\\
			i&-i&-i&i
		\end{pmatrix},\\
		T^{100}=
		\begin{pmatrix}
			-1&-1&-1&-1\\
			i&i&i&i\\
			i&i&i&i\\
			1&1&1&1
		\end{pmatrix},
		T^{101}=
		\begin{pmatrix}
			i&-i&i&-i\\
			1&-1&1&-1\\
			1&-1&1&-1\\
			-i&i&-i&i
		\end{pmatrix},\\
		T^{110}=
		\begin{pmatrix}
			1&1&-1&-1\\
			i&i&-i&-i\\
			-i&-i&i&i\\
			1&1&-1&-1
		\end{pmatrix},
		T^{111}=
		\begin{pmatrix}
			i&-i&-i&i\\
			-1&1&1&-1\\
			1&-1&-1&1\\
			i&-i&-i&i
		\end{pmatrix}.
	\end{split}
\end{equation}
These matrices are of rank 1 and all the tensor elements are nonzero. We emphasize that they match the two properties (a) and (b) in Theorem \ref{theoremII.rank1ofRBM}, and this match depends on the proper choice of the matrices for $U$ operators. Indeed, if there is a $U$ operator containing Pauli $Z$ matrix, for instance $U_{1,1}^{r}=X\otimes X, U_{2,1}^{r}=X\otimes I, U_{2,2}^{r-1}=I\otimes Z, U_{3,1}^r=Z\otimes I$, then the MPS matrix elements can have both zeros and nonzeros. The appearance of zero matrix elements makes it difficult to match the MPS to the RBM element-wise, because the matrix elements of RBM-MPS are all non-vanishing as shown in Theorem \ref{theoremII.rank1ofRBM}.

In fact, we do not have to solve the matrices and then find their ranks. We can immediately find the rank of the matrices from the Hamiltonian terms.  From Eq.~\eqref{Eq.MPSequations} used only for $T^{000}$, 
\begin{equation}\label{Eq.MPSequationsRank}
\begin{split}
T^{000}_{h_1h_2, h_3h_4}&=T^{000}_{h_1h_2, (1-h_3)(1-h_4)}\\
T^{000}_{h_1h_2, h_3h_4}&=T^{000}_{h_1h_2, h_3(1-h_4)}\\
T^{000}_{h_1h_2, h_3h_4}&=-i T^{000}_{(1-h_1)h_2, h_3h_4}(-1)^{h_1}\\
T^{000}_{h_1h_2, h_3h_4}&=-i T^{000}_{h_1(1-h_2), h_3h_4}(-1)^{h_2}.
\end{split}
\end{equation}
The physical indices are unchanged on both sides of Eq.~\eqref{Eq.MPSequationsRank} simply because the four equations are coming from acting with the physical operators on the LHS of Eq.~\eqref{Eq.MPSequationsRank},
\begin{equation}\label{Eq.fouroperators}
\begin{split}
\Lcal^{r}_{1,1}&=I^r_{1}\otimes Z^r_{2}\otimes Z^r_3, \\
\Lcal^{r}_{2,1}&=I^r_1\otimes I^r_2\otimes Z^r_3, \\
\Rcal^{r}_{2,2}&=Z^{r}_1\otimes I^{r}_2\otimes I^{r}_3, \\
\Rcal^{r}_{3,1}&=Z^{r}_1\otimes Z^{r}_2\otimes I^{r}_3, 
\end{split}
\end{equation}
which contain only Pauli $Z$ operators and identities. 
Hence, the left indices $h_1$ and $h_2$ of the matrix $T^{000}$ obey two independent constraints, and the right indices $h_3$ and $h_4$ obey two independent constraints. Therefore, the rank of the matrix $T^{000}$ can be at most 1, since each constraint for the left (\emph{or} right) indices eliminates half of the total rank. Acting with Eq.~\eqref{Eq.fouroperators} on the $T$-matrices with other physical indices, we also get two independent constraints on the left and right indices respectively. Hence the $T$ matrices with any physical indices are of rank 1. For other general models obeying the assumptions \eqref{Assumption1}, \eqref{Assumption2} and \eqref{Assumption3}, we can similarly find the constraints on the rank of the matrices by counting the independent $\Lcal$ or $\Rcal$ operators with only Pauli $Z$ matrices without solving explicitly the matrices by brute-force. We elaborate this idea in Sec.~\ref{Sec.RankMPS}.

We finally comment that for the $ZZXZZ$ model with one spin per unit cell, i.e., 
\begin{eqnarray}
H^{1-\mathrm{site}}_{ZZXZZ}= \sum_{i=0}^{r-1} Z^{r-2}Z^{r-1}X^r Z^{r+1}Z^{r+2}
\end{eqnarray}
Using the same calculation in this section, the ground state of $H^{1-\mathrm{site}}_{ZZXZZ}$ can be expressed as an MPS, 
\begin{eqnarray}\label{Eq.rank1MPS}
|\mathrm{GS}\rangle_{1-\mathrm{site}}= \sum_{\{g^r\}}\Tr\bigg(  \prod_{r=0}^{L-1}  T^{g^r} \bigg)|\{g^r\}\rangle
\end{eqnarray}
where the MPS matrices are
\begin{eqnarray}
T^{0}= 
\begin{pmatrix}
1 & 0 & 1 & 0\\
-i & 0 & -i & 0\\
0 & 1 &0 & 1\\
0 & -i & 0 & -i
\end{pmatrix}, ~~~
T^{1}=
\begin{pmatrix}
0 & -1 & 0 & 1\\
0 & -i & 0 & i\\
-1 & 0 & 1 & 0\\
-i & 0 & i & 0
\end{pmatrix}
\end{eqnarray}
Notice that both $T^{0}$ and $T^{1}$ are of rank $2$. By theorem \ref{theoremII.rank1ofRBM}, it is impossible to express the MPS  Eq.~\eqref{Eq.rank1MPS} as an RBM state.

\subsection{General Stabilizer Code Convention}
\label{Sec.GeneralSC}

A generic translational invariant stabilizer code is described by the Hamiltonian
\begin{equation}\label{Eq.GeneralStabilizerCodes}
H=-\sum_{r=0}^{L-1}\sum_{\alpha=1}^t \Ocal^r_{\alpha}
\end{equation}
where $t$ is the total number of types of interactions and $\alpha\in{1,\ldots, t}$ labels the type. Each unit cell contains $q$ spin-$\frac{1}{2}$'s. $\Ocal^r_{\alpha}$ is a product of Pauli $X$ and $Z$ operators such that (1) $\Ocal^r_\alpha$ is Hermitian and $(\Ocal^r_\alpha)^2=1$ for any $r, \alpha$; and (2) $\Ocal^r_\alpha$ and $\Ocal^{r'}_{\alpha'}$ commute for any $r, r', \alpha, \alpha'$, i.e., 
\begin{equation}
[\Ocal^r_{\alpha}, \Ocal^{r'}_{\alpha'}]=0, \quad \forall\; \alpha, \alpha^\prime, r, r^\prime.
\end{equation}
Each interaction term $\Ocal^r_\alpha$ is supported over the unit cells $r, r+1, ..., r+P_{\alpha}-1$, and  can be written as an ordered product of $P_\alpha$ number of local operators $o^r_{\alpha, \tau}$
\begin{equation}
\Ocal^r_{\alpha}=\prod_{\tau=1}^{P_\alpha} o^r_{\alpha, \tau},
\end{equation}
where $o^r_{\alpha, \tau}$ is a product of Pauli matrices only supported on the $(r+\tau-1)$-th unit cell. For convenience, we define $\Lcal^r_{\alpha, \tau}$ and $\Rcal^r_{\alpha, \tau}$ as follows
\begin{equation}\label{Eq.LR}
\begin{split}
\Lcal^r_{\alpha, \tau}&= \prod_{\mu=1}^{\tau} o^r_{\alpha, \mu},	\\
\Rcal^r_{\alpha, \tau}&= \prod_{\mu=\tau+1}^{P_\alpha} o^r_{\alpha, \mu},\quad \tau=1, \ldots, P_\alpha-1.
\end{split}
\end{equation}
By Assumption \ref{Assumption1}, we consider only cases where Eq.~\eqref{Eq.GeneralStabilizerCodes} has a unique ground state when using PBC. The ground state $\ket{\mathrm{GS}}$ is the common eigenstate of all $\Ocal^r_{\alpha}$'s with eigenvalue 1 for all $r$ and $\alpha$, 
\begin{equation}\label{Eq.GSproperty}
\Ocal^r_{\alpha} \ket{\mathrm{GS}} = \ket{\mathrm{GS}}, \quad \forall\; r, \alpha.
\end{equation}
Our goal is to express the ground state $|\mathrm{GS}\rangle$ as an MPS
\begin{equation}\label{Eq.MPSGS}
|\mathrm{GS}\rangle=\sum_{\{g^r_i\}}\Tr\bigg(\prod_{r=0}^{L-1}T^{g^r_1\ldots g^r_q}\bigg)|\{g^r_i\}\rangle.
\end{equation}
where $g^r_{\alpha}$, $(\alpha=1, ..., q)$, labels the value of the $\alpha$-th physical spin in the $r$-th unit cell. 

\subsection{General Algorithm to Construct MPS}
\label{Sec.GeneralMPS}

The calculation algorithm to construct an MPS representation is divided into in 4 steps. These 4 steps will follow those of Sec.~\ref{Sec.Example} for the $ZZXZZ$ model.

\begin{figure}[b]
	\centering
	\includegraphics[width=1\columnwidth]{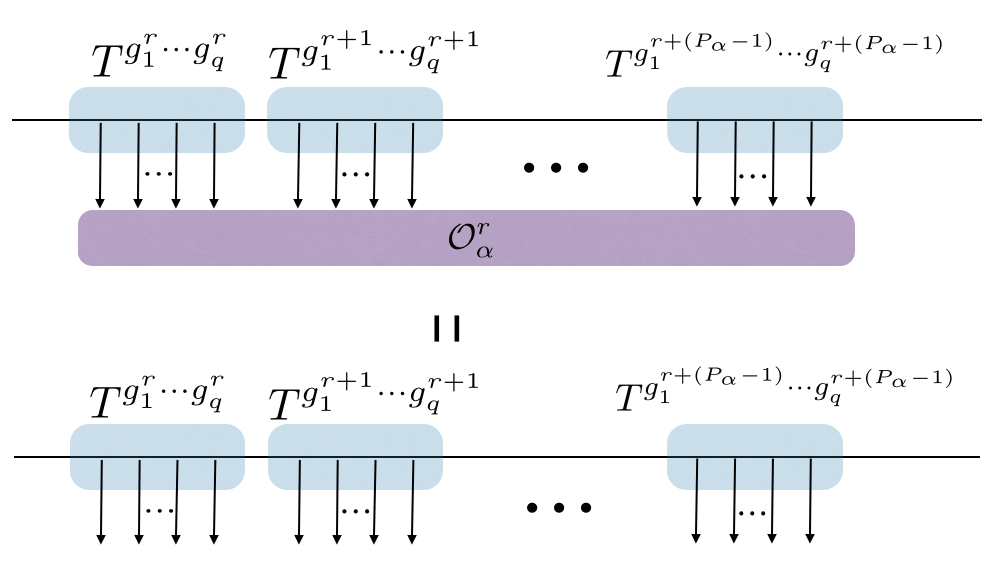}
	\caption{Graphical representation of Eq.~\eqref{Eq.MPSgeneral}. The shaded purple region represents the operator $\Ocal^r_\alpha$ acting on the physical indices.}
	\label{Fig_EqMPS}
\end{figure}

\paragraph*{1.} We start with the stabilizer condition Eq.~\eqref{Eq.GSproperty}. A sufficient condition for the MPS of Eq.~\eqref{Eq.MPSGS} to satisfy Eq.~\eqref{Eq.GSproperty} is
\begin{equation}\label{Eq.MPSgeneral}
\Ocal^r_\alpha\circ \Big(\prod_{r'=r}^{r+P_\alpha-1}T^{g^{r'}_{1}\ldots g^{r'}_q}\Big)=\Big(\prod_{r'=r}^{r+P_\alpha-1}T^{g^{r'}_{1}\ldots g^{r'}_q}\Big).
\end{equation}
Eq.~\eqref{Eq.MPSgeneral} is graphically represented as Fig.~\ref{Fig_EqMPS}.   In fact, Eq.~\eqref{Eq.MPSgeneral} is not only sufficient, but also necessary for Eq.~\eqref{Eq.GSproperty}, as derived in App.~\ref{app.Deriving0}.

\begin{figure*}
	\centering
	\includegraphics[width=0.9\textwidth]{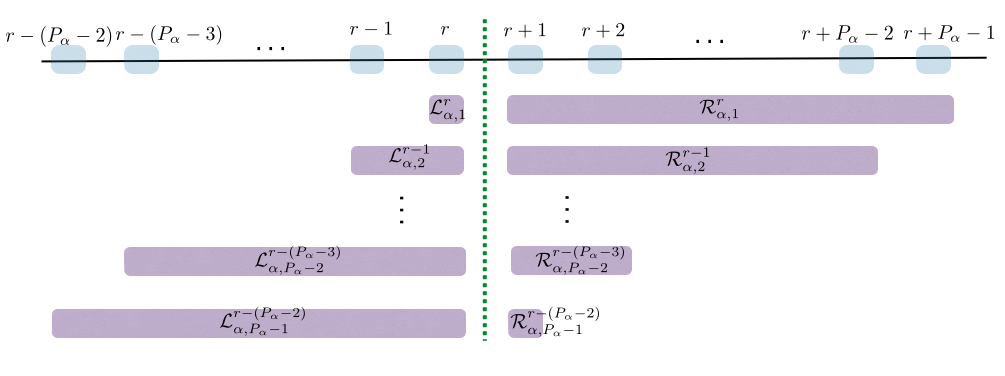}
	\caption{An illustration of the  operators $\Lcal^{r-(\tau-1)}_{\alpha, \tau}$ and $\Rcal^{r-(\tau-1)}_{\alpha, \tau}$ with fixed $r$ and $\alpha$, and all $1\tau\leq P_{\alpha}-1$.  The blue blocks represent unit cells. The purple blocks represent the operators $\Lcal^{r-(\tau-1)}_{\alpha, \tau}$, and the operators $\Rcal^{r-(\tau-1)}_{\alpha, \tau}$. }
	\label{Fig_Operators}
\end{figure*}

\paragraph*{2.} To find a solution of Eq.~\eqref{Eq.MPSgeneral}, we consider a  bipartition of  the Hamiltonian term $\Ocal^r_{\alpha}$ into the product of the left and right part, i.e., $\Lcal^r_{\alpha, \tau}$ and $\Rcal^r_{\alpha, \tau}$. The two parts act solely on two disjoint and contiguous sets of unit cells. For $P_{\alpha}>1$, since  $\Ocal^r_{\alpha}$ is supported on the unit cells between $r$ and $r+P_\alpha-1$, $\Lcal^r_{\alpha, \tau}$ is chosen to be supported from $r$ to $r+\tau-1$-th unit cell, and $\Rcal^r_{\alpha, \tau}$ is supported from $(r+\tau)$ to $(r+P_\alpha-1)$-th unit cell. The definitions of $\Lcal^r_{\alpha, \tau}$ and $\Rcal^r_{\alpha, \tau}$ are in given  Eq.~\eqref{Eq.LR}. 
Either $\Lcal^r_{\alpha, \tau}$ or $\Rcal^r_{\alpha, \tau}$ can act nontrivially on the MPS, although their product leaves the MPS invariant. This nontrivial action can be captured by a transformation on the virtual index exactly across the cut (between the $(r+\tau-1)$-th and the $(r+\tau)$-th unit cell). From Eq.~\eqref{Eq.MPSgeneral}, we find
\begin{eqnarray}\label{Eq.LRgeneral}
&&\Lcal^r_{\alpha, \tau}\circ \Big(\prod_{r'=r}^{r+\tau-1}T^{g^{r'}_{1}\ldots g^{r'}_q}\Big)= \Big(\prod_{r'=r}^{r+\tau-1}T^{g^{r'}_{1}\ldots g^{r'}_q}\Big)\cdot U^r_{\alpha, \tau}\\
&&\Rcal^r_{\alpha, \tau}\circ \Big(\prod_{r'=r+\tau}^{r+P_\alpha-1}T^{g^{r'}_{1}\ldots g^{r'}_q}\Big)=(U^r_{\alpha, \tau})^{-1}\cdot \Big(\prod_{r'=r+\tau}^{r+P_\alpha-1}T^{g^{r'}_{1}\ldots g^{r'}_q}\Big) .\nonumber
\end{eqnarray}
Eq.~\eqref{Eq.LRgeneral} is graphically represented as Fig.~\ref{Fig_EqMPSleftright}. We prove in App.~\ref{app.Deriving} that Eq.~\eqref{Eq.LRgeneral} is both necessary and sufficient for Eq.~\eqref{Eq.MPSgeneral}.

\begin{figure}
	\centering
	\includegraphics[width=1\columnwidth]{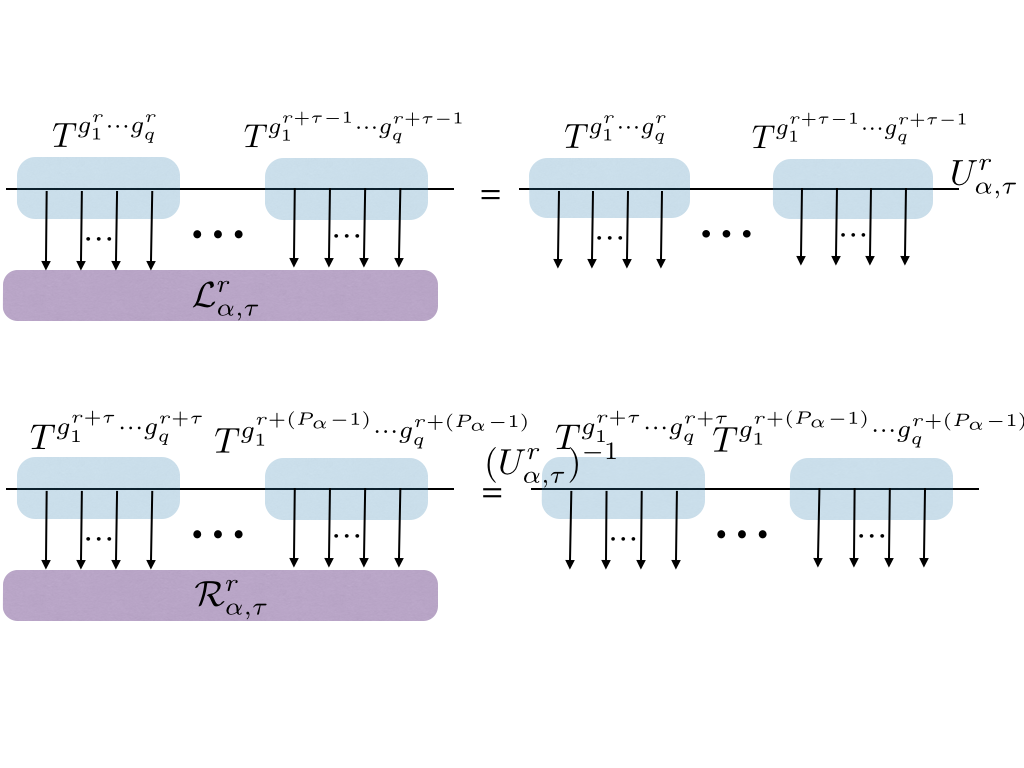}
	\caption{Graphical representation of Eq.~\eqref{Eq.LRgeneral}. The virtual operator $U^r_{i, \tau}$ and $(U^r_{i, \tau})^{-1}$ act on the right virtual index between the $r+\tau-1$ and $r$-th unit cell.}
	\label{Fig_EqMPSleftright}
\end{figure}

For convenience, for each choice of $(\alpha, \tau)$ we can shift $r$ to $r-\tau+1$ in Eq.~\eqref{Eq.LRgeneral} by translational invariance such that the $U^{r-\tau+1}_{\alpha, \tau}$ acts on the virtual bond between the $r$-th and $(r+1)$-th unit cell. See Fig.~\ref{Fig_Operators} for the operators that are obtained from shifting $\Ocal^r_{\alpha}$.  
Under the shift $r\to r-\tau+1$, Eq.~\eqref{Eq.LRgeneral} becomes
\begin{equation}\label{Eq.Lequations}
\begin{split}
&\Lcal^{r-\tau+1}_{\alpha, \tau}\circ \Big(\prod_{r'=r-\tau+1}^{r}T^{g^{r'}_{1}\ldots g^{r'}_q}\Big)= \Big(\prod_{r'=r-\tau+1}^{r}T^{g^{r'}_{1}\ldots g^{r'}_q}\Big)\cdot U^{r-\tau+1}_{\alpha, \tau}\\
&\Rcal^{r-\tau+1}_{\alpha, \tau}\circ \Big(\prod_{r'=r+1}^{r+P_\alpha-\tau}T^{g^{r'}_{1}\ldots g^{r'}_q}\Big)=(U^{r-\tau+1}_{\alpha, \tau})^{-1}\cdot \Big(\prod_{r'=r+1}^{r+P_\alpha-\tau}T^{g^{r'}_{1}\ldots g^{r'}_q}\Big),\\
& 1\leq \alpha \leq t, ~~~ 1\leq \tau\leq P_\alpha-1
\end{split}
\end{equation}

When the operator $\Ocal_{\alpha}^r$ is supported only over 1 unit cell, i.e., $P_\alpha=1$,  Eq.~\eqref{Eq.MPSgeneral} is already linear. Without loss of generality, we take $\Lcal^r_{\alpha, 1}=\Ocal^{r}_{\alpha}$,  $\Rcal^r_{\alpha, 1}=I^r$ and $U^r_{\alpha, 1}=I^{r}$ where $I^r$ is an identity operator acting on the $r$-th unit cell.

\paragraph*{3.} We further determine the minimal bond dimension of $T^{g_1^{r}\ldots g_q^{r}}$. 
We prove in App.~\ref{app.Ucommutation} that the commutation and anti-commutation relations of the virtual $U$ operators on RHS of Eq.~\eqref{Eq.Lequations} should match those of the physical $\Lcal$ and $\Rcal$ operators on the LHS,
	\begin{equation}\label{Eq.Algebra}
	\begin{split}
	\Lcal_{\alpha', \tau'}^{r-(\tau'-1)}&\Lcal_{\alpha'', \tau''}^{r-(\tau''-1)} 
	\\&= (-1)^{\mathbf{t}_{(\alpha'\tau'),(\alpha''\tau'')}^{r-(\tau'-1),r-(\tau''-1)}} \Lcal_{\alpha'', \tau''}^{r-(\tau''-1)} \Lcal_{\alpha', \tau'}^{r-(\tau'-1)},	\\
	\Rcal_{\alpha', \tau'}^{r-(\tau'-1)}&\Rcal_{\alpha'', \tau''}^{r-(\tau''-1)} 
	\\&= (-1)^{\mathbf{t}_{(\alpha'\tau'),(\alpha''\tau'')}^{r-(\tau'-1),r-(\tau''-1)}} \Rcal_{\alpha'', \tau''}^{r-(\tau''-1)} \Rcal_{\alpha', \tau'}^{r-(\tau'-1)},	\\
	U_{\alpha', \tau'}^{r-(\tau'-1)}&U_{\alpha'', \tau''}^{r-(\tau''-1)} 
	\\&= (-1)^{\mathbf{t}_{(\alpha'\tau'),(\alpha''\tau'')}^{r-(\tau'-1),r-(\tau''-1)}} U_{\alpha'', \tau''}^{r-(\tau''-1)} U_{\alpha', \tau'}^{r-(\tau'-1)},	\\
	1\leq \alpha', \alpha''&\leq t, ~~~1\leq \tau'\leq P_{\alpha'}-1, ~~~1\leq \tau''\leq P_{\alpha''}-1
	\end{split}
	\end{equation}
The parameter 
\begin{equation}\label{Eq.mod2}
\mathbf{t}_{(\alpha'\tau'),(\alpha''\tau'')}^{r-(\tau'-1),r-(\tau''-1)}=0, 1\mod 2
\end{equation} 
encodes whether $\textstyle U_{\alpha', \tau'}^{r-(\tau'-1)}$ and $\textstyle U_{\alpha'', \tau''}^{r-(\tau''-1)}$ commute or anti-commute. We ensemble the parameters $\textstyle \mathbf{t}_{(\alpha'\tau'),(\alpha''\tau'')}^{r-(\tau'-1),r-(\tau''-1)}$ into an anti-symmetric matrix $\mathbf{t}$. \footnote{In general, if the set of operators $\{U_{\lambda}\}$ satisfy $U_{\lambda'} U_{\lambda''}= e^{i \mathbf{t}_{\lambda', \lambda''}} U_{\lambda'}U_{\lambda''}$, where $\mathbf{t}_{\lambda', \lambda''}$ is a real number characterizing the commutation relations between $\{U_{\lambda}\}$, $\mathbf{t}_{\lambda', \lambda''}$ is antisymmetric: $\mathbf{t}_{\lambda', \lambda''}=-\mathbf{t}_{\lambda'', \lambda'}$. This is because from the above commutation relation, one move the phase to the left hand side as $e^{-i \mathbf{t}_{\lambda', \lambda''}}U_{\lambda'} U_{\lambda''}=  U_{\lambda'}U_{\lambda''}$. Combining with the definition $U_{\lambda''} U_{\lambda'}= e^{i \mathbf{t}_{\lambda'', \lambda'}} U_{\lambda''}U_{\lambda'}$, we derive that $e^{-i \mathbf{t}_{\lambda', \lambda''}}=e^{i \mathbf{t}_{\lambda'', \lambda'}}$. This yields that $\mathbf{t}_{\lambda'', \lambda'}=-\mathbf{t}_{\lambda', \lambda''}\mod 2\pi$. Applying to our case,  $\mathbf{t}_{(\alpha'\tau'),(\alpha''\tau'')}^{r-(\tau'-1),r-(\tau''-1)}=-\mathbf{t}_{(\alpha''\tau''),(\alpha'\tau')}^{r-(\tau''-1),r-(\tau'-1)}\mod 2$.  }

The algebra Eq.~\eqref{Eq.Algebra} is a generalization of the Clifford algebra, where the standard Clifford algebra is generated by mutually anti-commuting operators. In Ref.~\onlinecite{2010arXiv1005.4300J}, it was shown that any integer-valued antisymmetric matrix $\mathbf{t}$ can be block diagonalized by a unimodular integer matrix $V$, such that each nontrivial block is a $2\times 2$ anti-symmetric matrix with integer off-diagonal elements. Due to Eq.~\eqref{Eq.mod2}, only the modulo 2 values of the off-diagonal elements of the nontrivial $2\times 2$ blocks matter. The nontrivial blocks can therefore  be written as follows:\footnote{When there is an operator $\Lcal^{r-\tau+1}_{\alpha, \tau}$ commuting with all other $\Lcal^{r-\tau'+1}_{\alpha', \tau'}$ for any $(\tau', \alpha')$, the $\mathbf{t}$-matrix is not full rank. For instance, when there is an operator $\Ocal^r_{\alpha}$ which is only supported over one unit cell, i.e., $P_{\alpha}=1$, then by our convention in the previous paragraph, $\Lcal^r_{\alpha, 1}=\Ocal^r_{\alpha}, \Rcal^r_{\alpha, 1}=1$. Then $\Lcal^r_{\alpha, 1}$ commutes with $\Lcal^{r-\tau'+1}_{\alpha', \tau'}$ for any $(\tau', \alpha')$. Hence there are 0 blocks in the decomposition Eq.~\eqref{Eq.29}, i.e., $\mathbf{t}$ matrix is not full rank. }
\begin{equation}\label{Eq.29}
V \mathbf{t} V^{T} = 
\left(\begin{matrix}
0	&	1	\\
-1	&	0	\\
\end{matrix}\right)
\oplus
\left(\begin{matrix}
0	&	1	\\
-1	&	0	\\
\end{matrix}\right)
\cdots
\oplus
\left(\begin{matrix}
0	&	1	\\
-1	&	0	\\
\end{matrix}\right)
\oplus
0
\cdots.
\end{equation}
Here we explicitly keep the minus signs to make the antisymmetry manifest.   
In the new basis, the operators $\Lcal^r_{\alpha, \tau}$ become decoupled pairs of anti-commuting operators (such as Eq.~\eqref{Eq.CommutingPairs}); there are $\frac{\mathrm{rank}(\mathbf{t})}{2}$ such pairs. Since each pair provides a two dimensional irreducible representation, the dimension of the irreducible representation of the generalized Clifford algebra Eq.~\eqref{Eq.Algebra} is given by 
\begin{equation}\label{eq.dim=rankt}
D = 2^{\frac{\mathrm{rank}(\mathbf{t})}{2}}.
\end{equation}
Since the dimension of an irreducible representation of the algebra Eq.~\eqref{Eq.Algebra} is $D$, the matrices of the $U^r_{\alpha, \tau}$ operators, as well as the MPS matrix $T^{g^r_1\ldots g^r_q}$ under the irreducible representation should be $D\times D$ matrices.\footnote{One may consider $U^r_{\alpha, \tau}=0$ (for all $\alpha, \tau$ and  $r$) to be a solution of Eq.~\eqref{Eq.Algebra}. However, due to Eq.~\eqref{Eq.Lequations}, the $T^{g^r_1\ldots g^r_q}$ would be zero, hence the MPS is a null state. So we do not consider this solution.} Since the representation is irreducible, $D$ is also the minimal bond dimension. For the $ZZXZZ$ model with 3 spins per unit cell discussed in Sec.~\ref{Sec.Example}, the $\mathbf{t}$ matrix is given by Eq.~\eqref{eq.exampletmatrix}, which is of rank 4. By Eq.~\eqref{eq.dim=rankt}, the minimal bond dimension of the MPS is $D=2^{4/2}=4$, which matches the MPS explicitly derived in Eq.~\eqref{Eq.MPSZZXZZ}. 

\paragraph*{4.} We solve Eq.~\eqref{Eq.Lequations} for the MPS matrices $T$ with the minimal bond dimension $D$. Let us first determine the form of $U$. The matrix elements of $U$ can be obtained by finding the  representation of the algebra Eq.~\eqref{Eq.Algebra}. Here we focus only on irreducible representations such that the bond dimension is minimal.  Notice that there exist multiple choices of $U$ operators satisfying the same algebra Eq.~\eqref{Eq.Algebra}. However, since we only consider models with a single ground state, different solutions of $T$ from different choices of $U$ should correspond to the same ground state. Hence it is sufficient to work with one choice of $U$. As shown in App.~\ref{App. VirtualUOperator},  $U$ can always be constructed as a tensor product of $\frac{\mathrm{rank}(\mathbf{t})}{2}$ Pauli matrices. After specifying the virtual $U$ operators, we manipulate the equations in Eq.~\eqref{Eq.Lequations} such that all the physical operators on the LHS only act on the $r$-th unit cell, and all the equations are linear in $T^{g^r_1\cdots g^r_q}$. For instance, using the definition Eq.~\eqref{Eq.LR},  $\Lcal^{r-\tau+1}_{\alpha, \tau}=\prod_{\mu=1}^{\tau} o_{\alpha, \mu}^{r-\tau+1}$ and  $\Lcal^{r-\tau+1}_{\alpha, \tau-1}=\prod_{\mu=1}^{\tau-1} o_{\alpha, \mu}^{r-\tau+1}$, the combination $\left((\Lcal^{r-\tau+1}_{\alpha, \tau-1})^{-1}\Lcal^{r-\tau+1}_{\alpha, \tau}\right)=o_{\alpha, \tau}^{r-\tau+1}$ is only supported on the $r$-th  unit cell. In App.~\ref{app.Linearizing}, we show that Eq.~\eqref{Eq.Lequations} are equivalent to
\begin{equation}\label{Eq.LFinial}
\begin{split}
\Lcal^r_{\alpha, 1}\circ T^{g^{r}_{1}\ldots g^{r}_q}
&=T^{g^{r}_{1}\ldots g^{r}_q} \cdot U^r_{\alpha, 1}\\
\left((\Lcal^{r-\tau+1}_{\alpha, \tau-1})^{-1}\Lcal^{r-\tau+1}_{\alpha, \tau}\right) \circ T^{g^r_1\cdots g^r_q}
&=(U^{r-\tau+1}_{\alpha, \tau-1})^{-1}\cdot T^{g^r_1\cdots g^r_q}\cdot U^{r-\tau+1}_{\alpha, \tau}\\
\left(\Rcal^{r-\tau+1}_{\alpha, \tau-1}(\Rcal^{r-\tau+1}_{\alpha, \tau})^{-1}\right)\circ T^{g_1^{r} \ldots g_q^{r}} 
&=(U^{r-\tau+1}_{\alpha, \tau-1})^{-1} \cdot T^{g_1^{r} \ldots g_q^{r}} \cdot U^{r-\tau+1}_{\alpha, \tau}\\
\Rcal^{r-(P_\alpha-1)}_{\alpha, P_\alpha-1}\circ T^{g_1^{r} \ldots g_q^{r}} 
&= (U^{r-(P_\alpha-1)}_{\alpha, P_\alpha-1})^{-1}\cdot T^{g_1^{r} \ldots g_q^{r}},\\
 1\leq \alpha\leq t, ~~&~ 2\leq \tau\leq P_{\alpha}-1
\end{split}
\end{equation}
Since Eq.~\eqref{Eq.LFinial} is a set of linear equations in $T$, they can be numerically solved efficiently. For all the models we have explicitly checked (e.g. $Z^{q-1}X Z^{q-1}$ with $2\leq q\leq 6$), Eq.~\eqref{Eq.LFinial} has one non-zero solution up to an overall scaling. 

\section{An Inequality for Rank of MPS}
\label{Sec.RankMPS}

As discussed in Sec.~\ref{Sec.IntroBM}, a necessary condition for the existence of a finitely connected RBM of a stabilizer code ground state is Theorem \ref{theoremII.rank1ofRBM}. In this section, we propose an inequality which allows us to directly constrain the rank of the MPS without solving for the MPS matrices. 

Before we state and prove our theorem, it is convenient to introduce two notations. Denote a set of operators:
\begin{equation}\label{Eq.BigL}
\begin{split}
\mathbb{L}=\bigg\{\Lcal^r_{1, 1}, \Lcal^r_{2, 1}, \ldots , \Lcal^r_{t, 1}\bigg\}.
\end{split}
\end{equation}
In particular, $\mathbb{L}$ contains a special subset dubbed as $\mathbb{L}_{Z}$ such that the operators in $\mathbb{L}_Z$ are only the tensors products of Pauli $Z$ and the identity $I$ matrices. Denote $N^{\Lcal}_{Z}$ as the number of independent operators in $\mathbb{L}_{Z}$. Notice that due to translational invariance,  $N^{\Lcal}_{Z}$ is independent of $r$.

\begin{theorem}\label{MainTheorem}
For the matrices $T^{g_1^{r}\ldots g_q^{r}}$ satisfying Eq.~\eqref{Eq.LFinial} where the $U$ matrices are tensor product of Pauli matrices, the rank of $T^{g_1^{r}\ldots g_q^{r}}$ is upper bounded:
\begin{equation}\label{Eq.RankEstimate}
\mathrm{rank}(T^{g_1^{r}\ldots g_q^{r}})\le \frac{D}{2^{N^{\Lcal}_{Z}}}=2^{\frac{\mathrm{rank}(\mathbf{t})}{2}-N^{\Lcal}_{Z}},~~~ \forall\{g^r_i\}.
\end{equation} 
where $N^{\Lcal}_{Z}$ is the number of independent operators in $\mathbb{L}_{Z}$. 
\end{theorem}

\begin{proof}
To constrain $\mathrm{rank}(T^{g_1^{r}\ldots g_q^{r}})$, we only focus on a subset of Eq.~\eqref{Eq.LFinial} satisfying:
\begin{enumerate}[(1)]
	\item the physical operator on LHS of Eq.~\eqref{Eq.LFinial} only involves the operators in $\mathbb{L}_Z$;
	\item the virtual operator on RHS of Eq.~\eqref{Eq.LFinial} only acts on the right virtual index.
\end{enumerate}
Explicitly, this subset of equations are all included in the following equations:
\begin{equation}\label{Eq.66}
\begin{split}
\Lcal^r_{\alpha, 1}\circ T^{g_1^{r} \ldots g_q^{r}} &= T^{g_1^{r} \ldots g_q^{r}}\cdot U^r_{\alpha, 1}, \quad 
\forall \; \Lcal^r_{\alpha, 1} \in \mathbb{L}_Z. \\
\end{split}
\end{equation}
This subset is useful to constrain $\mathrm{rank}(T^{g_1^{r}\ldots g_q^{r}})$ because
\begin{enumerate}[(1)]
	\item both LHS and RHS of this subset of equations only involve the same matrix $T^{g_1^{r} \ldots g_q^{r}}$. Indeed, since $\Lcal^r_{\alpha, 1}$ belongs to $\mathbb{L}_Z $, the LHS is proportional to the matrix $T^{g_1^{r} \ldots g_q^{r}}$; 
	\item only the columns of $T^{g_1^{r} \ldots g_q^{r}}$ are constrained.
\end{enumerate} 
Using Theorem \ref{theoremF.sameLsameU} of App.~\ref{app.Linearizing},
the number of independent equations among Eq.~\eqref{Eq.66}(i.e., the number of independent constraints for the columns ) is given by the number of independent operators in $\mathbb{L}_Z$, i.e., $N^{\Lcal}_Z$. We know that $U$ operators form a generalized Clifford algebra, and as proven in App.~\ref{App. VirtualUOperator}, their matrices are tensor products of the Pauli matrices. More precisely, each virtual $U$ operator either swaps and/or multiplies by some factors ($\pm i$ or $\pm 1$) on half of the columns.  Hence, each independent  constraint eliminates half of the rank. Therefore, the rank of the MPS $T$ matrix is upper bounded:
\begin{equation}
\mathrm{rank}(T^{g_1^{r} \ldots g_q^{r}}) \le \frac{D}{2^{N^{\Lcal}_{Z}}}.
\end{equation}
This completes proving Theorem \ref{MainTheorem}. 
\end{proof}

In the 1D stabilizer codes we have studied, the upper bound in Eq.~\eqref{Eq.RankEstimate} always saturates.
	
\section{Restricted Boltzmann Machine State of a Stabilizer Code}
\label{Sec.BMSofQSC}

In this section, we discuss how to express the ground states of a class of stabilizer codes, which we dub as \textit{cocycle models},  as RBM states. They are a special class of Hamiltonians describing 1D symmetry protected topological phases.  We first use Theorem \ref{MainTheorem} to prove that the rank of the ground state MPS is 1. Then we use the $ZZXZZ$ model as an example to illustrate the construction of the RBM state with the RBM-MPS bond dimension 4. We further present a general and explicit algorithm to construct the RBM states for an arbitrary cocycle model, with the minimal RBM-MPS bond dimension. We finally conjecture that for any stabilizer code which satisfies Assumptions \ref{Assumption1}, \ref{Assumption2} and \ref{Assumption3} and also the necessary condition \ref{theoremII.rank1ofRBM}, it is possible to express its ground state as an RBM state with the minimal RBM-MPS bond dimension. 

\bigskip

\subsection{MPS Matrix Rank For Cocycle SPT Models}
\label{Sec.CocycleMPSrankOne}

In this section, we apply Theorem \ref{MainTheorem} to a particular family of stabilizer codes --- the cocycle Hamiltonians for symmetry protected topological phases --- and show that their MPS matrices are of rank 1. In App.~\ref{app.SPT}, we provide some backgrounds about the cocycle Hamiltonians, including the projective representations of the global symmetry $G$, cocycles $\omega_2$, cohomology group $\mathcal{H}^2(G,U(1))$ and 1D SPT phases. The cocycle $\omega_2 \in \mathcal{H}^2(G,U(1))$ classifies the 1D SPT phases with the discrete onsite symmetry $G$. In this paper, we restrict $G$ to be $(\mathbb{Z}_{2})^{q}$. The group elements are parametrized by $g=(g_1,g_2,\ldots,g_q)$ with $g_i\in \mathbb{Z}_2=\{0,1\}$, and the generic form of the cocycle is \cite{propitius1995topological,wang2015field}:
\begin{equation}\label{eq.cocycleexpression}
	\omega_2(g, g') = \exp\left(- i \pi \sum_{1\leq i<j\leq q} P_{ij} g_{j} g_{i}' \right), \quad
	g, g'\in G,
\end{equation}
where $P_{ij}$ can be either 0 or 1. The cocycles can also be used to construct representative SPT wave functions and representative parent Hamiltonians which are stabilizer codes. For simplicity, we dub the representative states and representative Hamiltonians as \textit{cocycle states} and \textit{cocycle Hamiltonians} respectively. See App.~\ref{app.SPT} for a brief overview.

The cocycle Hamiltonian for a $(\mathbb{Z}_2)^q$ SPT phase (with $q$ spin-$\frac{1}{2}$'s per unit cell) with a given generic cocycle $\omega_2$ Eq.~\eqref{eq.cocycleexpression} is 
\begin{equation}\label{eq.SPTHamiltonian}
	H_{(\Z_2)^q, \omega_2}=-\sum_{r=0}^{L-1}\sum_{\alpha=1}^{q} \Ocal^r_\alpha,
\end{equation}
with
\begin{widetext}
\begin{equation}\label{Eq.10}
	\Ocal^r_\alpha=
	\begin{cases}
		\displaystyle 
		\prod_{1<l\leq q}(Z^{r+1}_l Z^{r}_l)^{P_{1 l}} X^{r+1}_1 & \alpha=1\\
		\prod_{\alpha<l\leq q}
		(Z^{r+1}_l Z^{r}_l)^{P_{\alpha l}} X^{r+1}_\alpha \prod_{1\leq k< \alpha}(Z^{r+2}_k Z^{r+1}_k)^{P_{k\alpha}} & 1< \alpha <q\\
		\displaystyle X^r_q \prod_{1\leq k< q}(Z^{r+1}_k Z^{r}_k)^{P_{kq}} & \alpha=q.
	\end{cases}
\end{equation}
\end{widetext}
For $1< \alpha < q$, $\Ocal^r_\alpha$ are supported on $3$ unit cells; while for $\alpha=1, q$, $\Ocal^r_{1}$ and $\Ocal^r_{q}$ are supported on 2 unit cells. 
The Hamiltonian $H_{(\Z_2)^q, \omega_2}$ has the ground state (see App.~\ref{app.SPT} for details)
\begin{equation}\label{Eq.SPTState1}
|\mathrm{GS}\rangle_{(\Z_2)^q, \omega_2}= \sum_{\{g_i^r\}} \exp\bigg( i\pi \sum_{r=0}^{L-1}\sum_{1\leq i<j\leq q} P_{ij}(g^{r}_j-g^{r-1}_j)g^r_i \bigg) |\{g^r_i\}\rangle.
\end{equation}
When
\begin{equation}
	P = \left(\begin{matrix}
		0	&	1	&	1	\\
		0	&	0	&	1	\\
		0	&	0	&	0	\\
	\end{matrix}\right),
\end{equation}
the Hamiltonian Eq.~\eqref{eq.SPTHamiltonian} reduces to the Hamiltonian of the $ZZXZZ$ model, i.e.,  Eq.~\eqref{eq.ZZXZZ}.

\begin{theorem}\label{TheoremIV1}
	For the stabilizer codes of Eq.~\eqref{eq.SPTHamiltonian}, if $T^{g^r_1 \ldots g^r_q}$ is not null, then
	\begin{equation}
		\mathrm{rank}(T^{g^r_1 \ldots g^r_q}) = 1.
	\end{equation}
\end{theorem}

\begin{proof}
	
	To calculate $\mathrm{rank}(T^{g^r_1 \ldots g^r_q})$, we apply Theorem \ref{MainTheorem}, where the upper bound of the rank of $T^{g^r_1\cdots g^r_{q}}$ is given by $2^{\frac{\mathrm{rank}(\mathbf{t})}{2}-N^{\Lcal}_Z}$. We will first compute $\mathrm{rank}(\mathbf{t})$ and $N^{\Lcal}_Z$ respectively, and show that the upper bound is 1. We further show that the upper bound is saturated, which completes the proof of the theorem. 
		
	We first compute $\mathrm{rank}(\mathbf{t})$. To calculate the $\mathbf{t}$-matrix, we enumerate all possible $\Lcal^{r-\tau+1}_{\alpha, \tau}$ with all possible $(\alpha, \tau)$ and fixed $r$.  For $1< \alpha<q$, $\tau=1,2$; for $\alpha=1~ \mathrm{or}~ q$, $\tau=1$. Hence there are $2(q-1)$ $\Lcal$ operators:
	\begin{widetext}
		\begin{equation}\label{eq.Ocalleft}
			\begin{split}
				\Lcal^{r}_{1,1} &\equiv (Z^{r}_{2})^{P_{12}}\otimes (Z^{r}_{3})^{P_{13}}\otimes \cdots \otimes (Z^{r}_{q-1})^{P_{1(q-1)}}\otimes (Z^{r}_{q})^{P_{1q}}\\
				\vdots\\
				\Lcal^{r}_{q-2,1} &\equiv (Z^{r}_{q-1})^{P_{(q-2)(q-1)}}\otimes (Z^{r}_{q})^{P_{(q-2)q}}\\
				\Lcal^{r}_{q-1,1} &\equiv (Z^{r}_q)^{P_{(q-1)q}}\\
				\Lcal^{r-1}_{2,2} &\equiv (Z^{r-1}_{3})^{P_{23}}\otimes \cdots \otimes (Z^{r-1}_{q})^{P_{2q}}\otimes (Z^{r}_{1})^{P_{12}}\otimes X^r_2 \otimes (Z^{r}_{3})^{P_{23}}\otimes (Z^{r}_{4})^{P_{24}}\otimes \cdots \otimes (Z^{r}_{q})^{P_{2q}}\\
				\Lcal^{r-1}_{3,2} &\equiv (Z^{r-1}_{4})^{P_{34}}\otimes \cdots \otimes (Z^{r-1}_{q})^{P_{3q}}\otimes (Z^{r}_{1})^{P_{13}}\otimes (Z^{r}_{2})^{P_{23}} \otimes X^r_3 \otimes (Z^{r}_{4})^{P_{34}}\otimes (Z^{r}_{5})^{P_{35}}\otimes \cdots \otimes (Z^{r}_{q})^{P_{3q}}\\
				\vdots\\
				\Lcal^{r-1}_{q-1,2} &\equiv (Z^{r-1}_{q})^{P_{(q-1)q}} \otimes (Z^{r}_{1})^{P_{1(q-1)}}\otimes (Z^{r}_{2})^{P_{2(q-1)}} \otimes \cdots \otimes (Z^{r}_{q-2})^{P_{(q-2)(q-1)}}\otimes X^r_{q-1} \otimes (Z^{r}_{q})^{P_{(q-1)q}}\\
				\Lcal^{r}_{q,1} &\equiv (Z^{r}_{1})^{P_{1q}}\otimes (Z^{r}_{2})^{P_{2q}}\otimes \cdots \otimes (Z^{r}_{q-1})^{P_{(q-1)q}}\otimes X^r_q.
			\end{split}
		\end{equation}
	\end{widetext}
	We have suppressed the identity operators for simplicity. Among all the operators in Eq.~\eqref{eq.Ocalleft}, the first $q-1$ and the last one act  only on the $r$-th unit cell, while the remaining act both on the $r-1$-th and $r$-th unit cells. 
	It is straightforward to compute the commutation relation and determine the $\mathbf{t}$ matrix. In the basis where the operators are listed as in Eq.~\eqref{eq.Ocalleft}, i.e., $\{\Lcal^{r}_{1,1}, \cdots, \Lcal^{r}_{q-2,1}, \Lcal^{r}_{q-1,1} , \Lcal^{r-1}_{2,2}, \Lcal^{r-1}_{3,2}, \cdots, \Lcal^{r-1}_{q-1,2}, \Lcal^{r}_{q,1}\}$, the $\mathbf{t}$ matrix reads
	\begin{equation}\label{Eq.tSPT}
		\mathbf{t}=
		\begin{pmatrix}
			\mathbf{0}& \Lambda\\
			-\Lambda^{T} & \mathbf{0}
		\end{pmatrix},
	\end{equation}
	where $\mathbf{0}$ is a $(q-1)\times (q-1)$ dimensional  zero matrix, and $\Lambda$ is a $(q-1)\times (q-1)$ upper triangular matrix:
	\begin{equation}\label{Eq.LambdaMatrix}
		\Lambda=
		\begin{pmatrix}
			P_{12}&\cdots&P_{1(q-1)}& P_{1q}\\
			&\ddots&\vdots&\vdots\\
			& & P_{(q-2)(q-1)} & P_{(q-2)q}\\
			& & & P_{(q-1)q}
		\end{pmatrix}.
	\end{equation}
	Therefore, by Eq.~\eqref{Eq.tSPT}, we have:
	\begin{equation}\label{eq.SPTrankt}
		\mathrm{rank}(\mathbf{t})=2\mathrm{rank}(\Lambda).
	\end{equation} 
	Counting $\mathrm{rank}(\Lambda)$ is simply counting the number of independent rows in $\Lambda$. 
	
	We proceed to evaluate $\textstyle N_{Z}^{\Lcal}$. Recall that $\textstyle N_{Z}^{\Lcal}$ is defined to be the number of independent operators among $\mathbb{L}_Z$. In this case, we have: 
	\begin{equation}\label{Eq.12}
		\mathbb{L}_Z=\{\Lcal^r_{1, 1}, \Lcal^r_{2, 1}, \ldots , \Lcal^r_{q-1, 1}\}.
	\end{equation}
	A crucial observation is that the powers of the $Z$'s among the operators in Eq.~\eqref{Eq.12} are in  one-to-one correspondence  with the rows of the $\Lambda$ matrix in Eq.~\eqref{Eq.LambdaMatrix}. Hence, the number of independent operators among Eq.~\eqref{Eq.12} coincides with the number of independent rows of the $\Lambda$ matrix Eq.~\eqref{Eq.LambdaMatrix}, i.e., 
	\begin{equation}\label{eq.SPTNz}
		N_{Z}^{\Lcal}=\mathrm{rank}(\Lambda).
	\end{equation} 
	Using Theorem \ref{MainTheorem} and Eqs.~\eqref{eq.SPTrankt} and \eqref{eq.SPTNz}, we obtain 
	\begin{equation}\label{Eq.rankSPT}
		\mathrm{rank}(T^{g_1\ldots g_q})\le 
		2^{\frac{\mathrm{rank}(\mathbf{t})}{2}-N_{Z}^{\Lcal}}=
		2^{\frac{2\mathrm{rank}(\Lambda)}{2}-\mathrm{rank}(\Lambda)}=1.
	\end{equation}
	We have assumed that $T^{g_1\ldots g_q}$ is not null. $\mathrm{rank}(T^{g_1\ldots g_q})$ is thus assumed to be positive. Constrained by Eq.~\eqref{Eq.rankSPT}, we conclude that
	\begin{equation}
		\mathrm{rank}(T^{g_1\ldots g_q}) = 1.
	\end{equation}
\end{proof}

Since in the ground state Eq.~\eqref{Eq.SPTState1} for any spin configuration $\{g^r_i\}$ the coefficient of the basis $|\{g^r_i\}\rangle$ is a non-vanishing number, the MPS matrices are non-vanishing for any physical indices $g^r_1\ldots g^r_q$. This shows that the matrices $T^{g^r_1 \ldots g^r_q}$ are indeed not null. Hence the MPS matrix rank is 1 for the ground state MPS of an arbitrary cocycle Hamiltonian in Eq.~\eqref{eq.SPTHamiltonian} with the global symmetry $(\Z_2)^q$.

\subsection{An Example: $ZZXZZ$ Model Revisited}
\label{Sec.ExamplesRevisited}

In this section, we derive the RBM for the $ZZXZZ$ model with the RBM-MPS bond dimension $4$. 

We start with the ground state $|\mathrm{GS}\rangle$ of the $ZZXZZ$ model Eq.~\eqref{Eq.GSZZXZZ}. Concretely, by restricting Eq.~\eqref{Eq.SPTState1} to $q=3$,  and using $ P_{12}=P_{23}=P_{13}=1$, we obtain the ground state
\begin{equation}\label{Eq.GSZ2XZ2}
|\mathrm{GS}\rangle_{ZZXZZ}=\sum_{\{g^r_i\}}\exp\bigg(i\pi \sum_{r=0}^{L-1} \sum_{1\leq i< j\leq 3}(g_j^{r}-g^{r-1}_{j})g^r_{i}\bigg)|\{g^r_i\}\rangle.
\end{equation}
The coefficient of the configuration $|\{g^r_i\}\rangle$ is an exponent of a quadratic function of the physical spins. The idea to write Eq.~\eqref{Eq.GSZ2XZ2} in the form of an RBM state is to introduce hidden spins and to transform the quadratic terms in $g$ to linear terms. This is achieved by applying a series of identities proved in App.~\ref{app.UsefulIden}. The identities can be summarized as
\begin{equation}\label{Eq.Identities_n}
\begin{split}
&\exp\bigg(i\pi \mathrm{Sym}(g_1, \cdots, g_n)\bigg)\\
=&\frac{1}{\sqrt{2}}\sum_{h=0}^{1} \exp\bigg(i\frac{\pi}{2}(1-2h)\sum_{i=1}^n g_i-i\frac{\pi}{4}(1-2h)\bigg),
\end{split}
\end{equation}
where $g_i\in\{0,1\}$, and $\mathrm{Sym}(g_1, \cdots, g_n)$ is a symmetric summation of quadratic expressions in $g_i$, i.e.,
\begin{equation}
\mathrm{Sym}(g_1, \cdots, g_n)\equiv \sum_{1\leq i<j\leq n}g_jg_i.
\end{equation}
We introduce the following definitions to simplify the discussion below: 
\begin{enumerate}
	\item \textit{The on-site terms}: the quadratic terms involving only the visible spins from a single unit cell. For example: $g_j^{r}g^r_{i}$, $g_j^{r-1}g^{r-1}_{i}$, etc. 
	\item \textit{The inter-site terms}: the quadratic terms involving the visible spins from different unit cells. For example: $g_j^{r-1}g^r_{i}$, $g_j^{r}g^{r-1}_{i}$, etc. 
	\item \emph{The on-site symmetric expressions}: the symmetric expressions involving only visible spins from a single unit cell. For example: $\Sym(g^r_i, g^r_j, g^r_k)$, etc. 
	\item \emph{The inter-site symmetric expressions}: the symmetric expressions involving visible spins from different unit cells. For example: $\Sym(g^{r-1}_i, g^{r}_j, g^r_k)$, etc. 
\end{enumerate}

To convert Eq.~\eqref{Eq.GSZ2XZ2} into an RBM state, our strategy is as follows. We group all the quadratic terms in the exponent of Eq.~\eqref{Eq.GSZ2XZ2} into a sum of symmetric expressions, and apply the identity Eq.~\eqref{Eq.Identities_n} to each symmetric expression. 
For the inter-site symmetric expression, applying Eq.~\eqref{Eq.Identities_n} introduces a hidden spin of type-$h$; for the on-site symmetric expression, applying Eq.~\eqref{Eq.Identities_n} introduces a hidden spin of type-$\widetilde{h}$. As discussed in Sec.~\ref{Sec.IntroBM}, each hidden spin of type-$h$ doubles the bond dimension once we write the RBM state as an MPS (i.e., RBM-MPS), while the hidden spin of type-$\widetilde{h}$ does not contribute to the bond dimension. Hence, to obtain the RBM state whose RBM-MPS bond dimension is as small as possible,  we are aiming to group the quadratic expressions in Eq.~\eqref{Eq.GSZ2XZ2} to as few inter-site symmetric expressions as possible, together with some additional on-site symmetric expressions.


We first discuss the inter-site terms in Eq.~\eqref{Eq.GSZ2XZ2}, i.e., $\sum_{1\leq i<j\leq 3}g^{r-1}_{j}g^r_{i}$, because on-site terms do not contribute to the inter-site symmetric expressions.  
There are different ways to decompose the inter-site terms in the exponent of Eq.~\eqref{Eq.GSZ2XZ2} as a summation of symmetric expressions. Superficially, there are 3 inter-site terms, $\sum_{1\leq i<j\leq 3}g^{r-1}_{j}g^r_{i}=g^{r-1}_{2}g^r_{1}+g^{r-1}_{3}g^r_{1}+g^{r-1}_{3}g^r_{2}$, and it seems that one has to introduce 3 hidden variables by applying Eq.~\eqref{Eq.Identities_n} to the three terms separately. However, it is possible to organize the three inter-site terms into the sum of two inter-site symmetric expressions and one on-site symmetric expression. Concretely, 
\begin{equation}\label{eq.regroupexample}
\begin{split}
&\sum_{1\leq i<j\leq 3}g^{r-1}_{j}g^r_{i}= \\&\mathrm{Sym}(g^{r-1}_2, g^r_1)+\mathrm{Sym}(g^{r-1}_3, g^r_1, g^r_2)-\mathrm{Sym}(g^r_1,g^r_2).
\end{split}
\end{equation}
Under the decomposition Eq.~\eqref{eq.regroupexample} and applying Eq.~\eqref{Eq.Identities_n}, we need to introduce 2 hidden spins of type-$h$, which we denote as $h^r_1$ and $h^r_2$. From the discussion in the last paragraph, the bond dimension of the RBM-MPS is $2^2=4$, which precisely matches the \emph{minimal} bond dimension of the $ZZXZZ$ model derived in Sec.~\ref{Sec.Example}. This shows that there is no way to decompose the quadratic expression $\sum_{1\leq i<j\leq 3}g^{r-1}_{j}g^r_{i}$ in Eq.~\eqref{Eq.GSZ2XZ2} as a sum of at most one inter-site symmetric expression, together with some additional on-site symmetric expressions. Different decompositions of $\sum_{1\leq i<j\leq 3}g^{r-1}_{j}g^r_{i}$ should include at least two inter-site symmetric expressions. We will provide a general recipe of grouping the inter-site terms in Sec.~\ref{Sec.RBMCocycle} for all the 1D cocycle models and show that the grouping is optimal.

We further consider the on-site terms $\sum_{1\leq i<j\leq 3}g^{r}_{j}g^r_{i}$. We use the same decomposition as Eq.~\eqref{eq.regroupexample} by replacing $g^{r-1}_j$ with $g^r_j$, and obtain
\begin{equation}
\begin{split}
\sum_{1\leq i<j\leq 3}g^{r}_{j}g^r_{i}= \mathrm{Sym}(g^{r}_2, g^r_1)+\mathrm{Sym}(g^{r}_3, g^r_1, g^r_2)-\mathrm{Sym}(g^r_1,g^r_2).
\end{split}
\end{equation}
Applying Eq.~\eqref{Eq.Identities_n} for all the symmetric expressions, the ground state $|\mathrm{GS}\rangle_{ZZXZZ}$ can be rewritten as an RBM state
\begin{widetext}
\begin{equation}\label{eq.ZZXZZRBMderivation}
\begin{split}
|\mathrm{GS}\rangle_{ZZXZZ}
=&\sum_{\{g^r_i\}}\exp\bigg(i\pi \sum_{r=0}^{L-1} -\mathrm{Sym}(g^{r-1}_2, g^r_1)-\mathrm{Sym}(g^{r-1}_3, g^r_1, g^r_2)
+ \mathrm{Sym}(g^{r}_2, g^r_1)+\mathrm{Sym}(g^{r}_3, g^r_1, g^r_2)\bigg)|\{g^r_i\}\rangle\\
=& \sum_{\{g^r_i\}}\sum_{\substack{\{h^r_1,h^{r}_2\}\{\widetilde{h}^r_1, \widetilde{h}^r_2\}}} \prod_{r=0}^{L-1}\exp\bigg( 
-i\frac{\pi}{2} (1-2h^r_1)(g_2^{r-1}+g_1^r)+i\frac{\pi}{4}(1-2h_1^r)-i\frac{\pi}{2} (1-2h^r_2)(g_3^{r-1}+g_1^r+g_2^r)	\\
&+i\frac{\pi}{4}(1-2h_2^r)+i\frac{\pi}{2} (1-2\widetilde{h}^r_1)(g_1^{r}+g_2^r)-i\frac{\pi}{4}(1-2\widetilde{h}_1^r)+i\frac{\pi}{2} (1-2\widetilde{h}^r_2)(g^r_1+g_2^{r}+g_3^r)-i\frac{\pi}{4}(1-2\widetilde{h}_2^r)\bigg)|\{g^r_i\}\rangle
\end{split}
\end{equation}
We have suppressed the overall normalization constant. From the discussion in Sec.~\ref{Sec.IntroBM}, the RBM state Eq.~\eqref{eq.ZZXZZRBMderivation} can further be written as an MPS with the MPS matrix:
\begin{equation}\label{eq.ZZXZZtensor}
\begin{split}
&T^{g^{r}_1g^r_2g^r_3}_{h^r_1h^r_2, h^{r+1}_1h^{r+1}_2}=\\
&\sum_{\substack{\{\widetilde{h}^r_1, \widetilde{h}^r_2\}}}\exp\bigg( -i\frac{\pi}{2} (1-2h^r_1)g_1^r-i\frac{\pi}{2} (1-2h^{r+1}_1)g_2^r+i\frac{\pi}{4}(1-2h_1^r)
-i\frac{\pi}{2} (1-2h^r_2)(g_1^r+g_2^r)-i\frac{\pi}{2} (1-2h^{r+1}_2)g_3^r\\
&+i\frac{\pi}{4}(1-2h_2^r)+i\frac{\pi}{2} (1-2\widetilde{h}^r_1)(g_1^{r}+g_2^r)-i\frac{\pi}{4}(1-2\widetilde{h}_1^r)
+i\frac{\pi}{2} (1-2\widetilde{h}^r_2)(g^r_1+g_2^{r}+g_3^r)-i\frac{\pi}{4}(1-2\widetilde{h}_2^r)\bigg).
\end{split}
\end{equation}
\end{widetext}
The bond dimension of the RBM-MPS Eq.~\eqref{eq.ZZXZZtensor} is indeed 4, which matches the bond dimension derived from the RBM state Eq.~\eqref{eq.ZZXZZRBMderivation}. Since we have shown in Sec.~\ref{Sec.Example} that the \emph{minimal} bond dimension of the $ZZXZZ$ MPS is 4, there can not be an RBM state with the number of hidden spin of type-$h$ per unit cell less than 2.   This implies that our RBM state is the most \emph{optimal}, in the sense that the number of hidden spins of type-$h$ is minimal.   



\begin{figure}[b]
	\centering
	\includegraphics[width=\columnwidth]{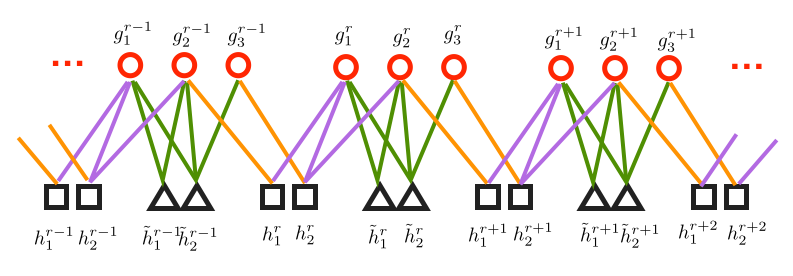}
	\caption{Graphical representation of the RBM state of the $ZZXZZ$ model. }
	\label{FigRBMZZXZZ}
\end{figure}

Fig.~\ref{FigRBMZZXZZ} is a graphical representation of the RBM state Eq.~\eqref{eq.ZZXZZRBMderivation}. In fact, the RBM-MPS matrices  Eq.~\eqref{eq.ZZXZZtensor} are the same as the MPS matrices Eq.~\eqref{Eq.MPSZZXZZ} in derived in Sec.~\ref{Sec.Example}.  As we will see in the next subsection, for more general models $Z^{q-1} X Z^{q-1}$, each unit cell contains $q$ visible spins. Our construction yields the RBM-MPS bond dimension $2^{q-1}$, and we need to introduce $2(q-1)$ hidden spins on average for each unit cell.  Among them,  $(q-1)$ are of the type-$h$ while the remaining $(q-1)$ are of the type-$\widetilde{h}$.

\subsection{RBM States of Cocycle Hamiltonians}
\label{Sec.RBMCocycle}

In Sec.~\ref{Sec.CocycleMPSrankOne}, we have shown that the MPS matrices of the $(\Z_2)^q$ cocycle Hamiltonians (with $q$ spin-$\frac{1}{2}$'s per unit cell) are all of rank 1. Then it is natural to ask if the ground state of the cocycle Hamiltonians can always be expressed as an RBM state, whose RBM-MPS bond dimension being $D$ defined in Eq.~\eqref{eq.dim=rankt}. In this subsection, we describe a procedure to obtain the RBM states with minimal number of hidden spins. In particular, we generalize and apply the procedures of Sec.~\ref{Sec.ExamplesRevisited}, and we present explicit RBM states for $Z^{q-1}XZ^{q-1}$ cocycle Hamiltonians with arbitrary $q$. See App.~\ref{app.J} for more examples. 

The cocycle Hamiltonian in Eq.~\eqref{eq.SPTHamiltonian} has the ground state $|\mathrm{GS}\rangle_{(\Z_2)^q, \omega_2}$ in Eq.~\eqref{Eq.SPTState1}. To convert it to an RBM state, we follow the same procedures in Sec.~\ref{Sec.ExamplesRevisited}. The core idea  is that we need to group the inter-site terms $\sum_{1\leq i<j\leq q} P_{ij}g^{r-1}_j g^r_i$ as a sum of the $\mathrm{rank}(\Lambda)$ inter-site symmetric expressions together with some on-site terms. Since each inter-site symmetric expression contributes a hidden spin of type-$h$ which doubles the bond dimension of the RBM-MPS, the bond dimension of the RBM-MPS is thus $2^{\mathrm{rank}(\Lambda)}\equiv  2^{\frac{\mathrm{rank}(\mathbf{t})}{2}}$. This is precisely the \emph{minimal} bond dimension derived in Sec.~\eqref{Sec.GeneralMPS}, which in turn implies that the decomposition of the inter-site terms is optimal, i.e, the number of type-$h$ hidden spins is minimal in our construction. 

\begin{lemma}\label{lemma.GaussianElimination}
For an inter-site quadratic term,
\begin{equation}
\left( g^{r-1} \right)^{T} \cdot \Gamma \cdot g^r = \sum_{i,j=1}^{q} \Gamma_{ij}g^{r-1}_i g^r_j, \quad \Gamma_{ij} \in \{0,1\},
\end{equation} 
there exists a unimodular transformation $G$ such that $\Gamma$ transforms to
\begin{equation}\label{Eq.24}
\Gamma \to
\widehat{\Gamma}=(G)^T\cdot \Gamma \cdot G=
\begin{pmatrix}
\gamma	\\
\mathbf{0}
\end{pmatrix}
\mod{2},
\end{equation}
where the integer matrix $\gamma$ of size $\mathrm{rank}(\Gamma)\times q$ has full row rank:
\begin{equation}
\mathrm{rank}(\gamma) = \mathrm{rank}(\Gamma).
\end{equation}
The vectors $g^{r-1}$ and $g^{r}$ transform as
\begin{equation}\label{Eq.90}
g^{r-1}_{i}\to \widehat{g}^{r-1}_i=\sum_{j=1}^{q}G^{-1}_{ij}g^{r-1}_j, \quad
g^{r}_{i}\to \widehat{g}^r_i=\sum_{j=1}^{q}G^{-1}_{ij}g^r_j.
\end{equation}
and 
\begin{eqnarray}
\left( g^{r-1} \right)^{T} \cdot \Gamma \cdot g^r = \left( \widehat{g}^{r-1} \right)^{T} \cdot \widehat{\Gamma} \cdot \widehat{g}^r 
\end{eqnarray}
\end{lemma}

\begin{proof}
Our proof is based on the Gaussian elimination algorithm. For simplicity, we first introduce the matrix notations: $I$ represents the identity $q\times q$ matrix, and $E(i, j)$ represents a $q\times q$ matrix whose elements are
\begin{equation}
\left( E(i,j) \right)_{m,n} = \delta_{m,i}\delta_{n,j}, \quad\forall\; m,n=1,2,\ldots,q.
\end{equation}
In other words, the only nonzero value of $E(i,j)$ is $1$ located at the $i$-th row and $j$-th column. Moreover, we use the following two types of matrix row transformations: 
\begin{equation}\label{Eq.G1G2}
\begin{split}
G_1(i, j)&=I+ E(i, j)+E(j, i)-E(i, i)-E(j, j)\\
G_2(i, j)&=I+E(j, i), \quad i\neq j.
\end{split}
\end{equation}
It is obvious that both $G_1$ and $G_2$ are unimodular, i.e.,
\begin{equation}
|\det(G_1(i,j))|=1, \quad
|\det(G_2(i,j))|=1.
\end{equation} 
The products of $G_1$'s and $G_2$'s are also unimodular.

The first transformation $G_1(i, j)$ interchanges the $i$-th row and the $j$-th row of $\Gamma$, and the second transformation $G_2(i, j)$ adds the $i$-th row to the $j$-th row.\footnote{Notice that the matrix determinant $\det(G_2(i,i))= \det(I+E(i,i))=0$. Since we only consider uni-modular transformation, we do not allow $i=j$ in $G_2(i,j)$.  } There exists a sequence of $G_1(i,j)$ and $G_2(i,j)$ such that:
\begin{equation}\label{Eq.ReducedLambda}
\prod_{m} G_{k_m} (i_m,j_m) \cdot \Gamma= 
\begin{pmatrix}
\gamma'	\\
\mathbf{0}
\end{pmatrix} \quad \mod{2},
\end{equation}
where the matrix $\gamma^\prime$ of size $\mathrm{rank}(\Gamma)\times q$ has full row rank, and its elements are either 0 and 1. Denote:
\begin{equation}
G = \left( \prod_{m} G_{k_m} (i_m,j_m) \right)^T.
\end{equation}
Using Eq.~\eqref{Eq.ReducedLambda}, we have:
\begin{equation}
\widehat{\Gamma} = 
G^T \cdot \Gamma \cdot G 
= 
\begin{pmatrix}
\gamma'	\\
\mathbf{0}
\end{pmatrix} \cdot G 
=
\begin{pmatrix}
\gamma	\\
\mathbf{0}
\end{pmatrix}
\mod{2},
\end{equation}
where
\begin{equation}
\gamma = \gamma^\prime \cdot G,
\end{equation}
and $\gamma$ of size $\mathrm{rank}(\Gamma) \times q$ has full row rank. 
\end{proof}

\begin{lemma}\label{lemma.numberOfSymmetricExpressions}
The inter-site term in the ground state $|\mathrm{GS}\rangle_{(\Z_2)^q, \omega_2}$ Eq.~\eqref{Eq.SPTState1} $\sum_{1\leq i<j\leq q} P_{ij}g^{r-1}_j g^r_i$ can be grouped into $\mathrm{rank}(\Lambda)$ number of inter-site symmetric expressions and $\mathrm{rank}(\Lambda)$ on-site symmetric expressions,  where $\Lambda$ is defined in Eq.~\eqref{Eq.LambdaMatrix}.
\end{lemma}

\begin{proof}
We first define the $\Gamma$ matrix:
\begin{equation}\label{Eq.Gamma}
\Gamma \equiv
\begin{pmatrix}
0 & 0 & \cdots & 0 & 0\\
P_{12} & 0 & \cdots & 0 & 0\\
P_{13} & P_{23} &\cdots &0 & 0\\
\vdots & \vdots & \ddots & \vdots & \vdots\\
P_{1q} & P_{2q} & \cdots & P_{(q-1)q} & 0
\end{pmatrix}=
\begin{pmatrix}
0 & 0 & \cdots & 0 & 0\\
&&& & 0\\
&&& & 0\\
&&\Lambda^T& & \vdots\\
&&& & 0
\end{pmatrix}.
\end{equation}
The matrix $\Gamma$ is a $q\times q$ matrix, whose each element is defined modulo 2. There are 0’s in the first row and last column because $g^{r-1}_1$ and $g^{r}_q$ do not appear in the sum $\sum_{1\leq i<j\leq q}P_{ij}g^{r-1}_j  g^{r}_i$. The bottom-left $(q-1)\times (q-1)$ block of $\Gamma$ is $\Lambda^T$ where $\Lambda$ is defined in Eq.~\eqref{Eq.LambdaMatrix}. In particular,
\begin{equation}
\mathrm{rank}(\Gamma)=\mathrm{rank}(\Lambda).
\end{equation} 
Using this notation, we have:
\begin{equation}\label{Eq.45}
\sum_{1\leq i<j\leq q} P_{ij}g^{r-1}_j g^r_i = (g^{r-1})^T\cdot \Gamma \cdot g^r.
\end{equation} 
Using Lemma \ref{lemma.GaussianElimination}, Eq.~\eqref{Eq.45} can be simplified:
\begin{equation}\label{Eq.26}
\begin{split}
\sum_{1\leq i<j\leq q} P_{ij}g^{r-1}_j g^r_i
=\sum_{i=1}^{\mathrm{rank}(\Lambda)}\widehat{g}^{r-1}_{i}\sum_{j=1}^{q}\widehat{\Gamma}_{ij}\widehat{g}^r_{j}.
\end{split}
\end{equation}
It can be decomposed by the symmetric expressions:
\begin{equation}
\begin{split}
\sum_{1\leq i<j\leq q} P_{ij}g^{r-1}_j g^r_i
&= \sum_{i=1}^{\mathrm{rank}(\Lambda)} \Sym(\widehat{g}^{r-1}_i, \widehat{\Gamma}_{i1}\widehat{g}^r_{1}, \ldots , \widehat{\Gamma}_{iq}\widehat{g}^r_{q})\\
&~~~-\sum_{i=1}^{\mathrm{rank}(\Lambda)} \Sym( \widehat{\Gamma}_{i1}\widehat{g}^r_{1}, \ldots , \widehat{\Gamma}_{iq}\widehat{g}^r_{q}).
\end{split}
\end{equation}
The first $\mathrm{rank}(\Lambda)$ terms are inter-site symmetric expressions, and the remaining $\mathrm{rank}(\Lambda)$ terms are the on-site terms. This completes the proof. 
\end{proof} 

\begin{theorem}
There exists an RBM for the state Eq.~\eqref{Eq.SPTState1} whose RBM-MPS has the minimal bond dimension $2^{\mathrm{rank}(\Lambda)}$ where $\Lambda$ is defined in Eq.~\eqref{Eq.LambdaMatrix}. 
\end{theorem}

\begin{proof}
Using Lemma \ref{lemma.GaussianElimination} and \ref{lemma.numberOfSymmetricExpressions}, we obtain
\begin{equation}\label{Eq.CocycleGrouping}
\begin{split}
&\exp\bigg(-i \pi \sum_{1\leq i<j\leq q} P_{ij}g^{r-1}_j g^r_i \bigg)
=\exp\bigg(-i\pi (\widehat{g}^{r-1})^{T} \cdot \widehat{\Gamma }\cdot \widehat{g}^r\bigg)	\\
&=\exp\bigg(-i\pi \sum_{i=1}^{\mathrm{rank}(\Lambda)} \Sym(\widehat{g}^{r-1}_i, \widehat{\Gamma}_{i1}\widehat{g}^r_{1}, \ldots , \widehat{\Gamma}_{iq}\widehat{g}^r_{q}) \\
&~~~+ i\pi \sum_{i=1}^{\mathrm{rank}(\Lambda)} \Sym( \widehat{\Gamma}_{i1}\widehat{g}^r_{1}, \ldots , \widehat{\Gamma}_{iq}\widehat{g}^r_{q}) \bigg).
\end{split}
\end{equation}
Applying Eq.~\eqref{Eq.Identities_n} to the inter-site symmetric expressions leads to:
\begin{widetext}
\begin{equation}\label{Eq.CocycleGrouping1}
\begin{split}
\exp\bigg(-i \pi \sum_{1\leq i<j\leq q} P_{ij}g^{r-1}_j g^r_i \bigg)
=&\prod_{i=1}^{\mathrm{rank}(\Lambda)}\Bigg[\frac{1}{\sqrt{2}}\sum_{h^r_i=0}^{1} \exp\bigg(-i\frac{\pi}{2}(1-2h^r_i)(\widehat{g}^{r-1}_i+\sum_{j=1}^{q}\widehat{\Gamma}_{ij}\widehat{g}^{r}_j )+i\frac{\pi}{4}(1-2h^r_i)\bigg) 	\\
&\times \exp\bigg(-i\pi \mathrm{Sym}(\widehat{\Gamma}_{i1}\widehat{g}^r_1, \cdots, \widehat{\Gamma}_{iq}\widehat{g}^r_{q})\bigg) \Bigg].
\end{split}
\end{equation}
\end{widetext}
Notice that further introducing the hidden spins by linearizing the on-site terms on RHS of Eq.~\eqref{Eq.CocycleGrouping} does not increase the bond dimension of the RBM-MPS. Hence we have shown that the RBM-MPS derived via the above algorithm has $\mathrm{rank}(\Lambda)$ hidden spins of type $h$,  which corresponds to the RBM-MPS bond dimension $D=2^{\mathrm{rank}(\Gamma)}=2^{\mathrm{rank}(\Lambda)}$.  This matches the bond dimension Eq.~\eqref{eq.dim=rankt} associated with the irreducible representation in Sec.~\ref{Sec.MPSSC}. 
\end{proof}

We use the rest of this section to express the state Eq.~\eqref{Eq.SPTState1} as an RBM explicitly.
\begin{eqnarray}\label{Eq.46}
\begin{split}
&\exp\bigg(i \pi \sum_{1\leq i<j\leq q} P_{ij}(g^r_j-g^{r-1}_j )g^r_i \bigg)\\
&=\exp\bigg(-i\pi \sum_{i=1}^{\mathrm{rank}(\Lambda)} \Sym(\widehat{g}^{r-1}_i, \widehat{\Gamma}_{i1}\widehat{g}^r_{1}, \ldots , \widehat{\Gamma}_{iq}\widehat{g}^r_{q}) \\
&~~~+i\pi \sum_{i=1}^{\mathrm{rank}(\Lambda)} \Sym(\widehat{g}^{r}_i, \widehat{\Gamma}_{i1}\widehat{g}^r_{1}, \ldots , \widehat{\Gamma}_{iq}\widehat{g}^r_{q})  \bigg).
\end{split}
\end{eqnarray}
Applying Eq.~\eqref{Eq.Identities_n} to Eq.~\eqref{Eq.46}, we can write the ground state $|\mathrm{GS}\rangle_{(\Z_2)^q, \omega_2}$ as an RBM state
\begin{equation}\label{Eq.32}
\begin{split}
&|\mathrm{GS}\rangle_{(\Z_2)^q, \omega_2} = \\&\sum_{\{g_i^r\}, \{h^r_i\}, \{\widetilde{h}^r_{i}\}}  \prod_{i=1}^{\mathrm{rank}(\Lambda)}\exp\bigg(-i\frac{\pi}{2}(1-2h^r_i)(\widehat{g}^{r-1}_i+\sum_{j=1}^{q}\widehat{\Gamma}_{ij}\widehat{g}^{r}_j )\\&+i\frac{\pi}{4}(1-2h^r_i)+i\frac{\pi}{2}(1-2\widetilde{h}^r_i)(\widehat{g}^{r}_i+\sum_{j=1}^{q}\widehat{\Gamma}_{ij}\widehat{g}^{r}_j )\\&-i\frac{\pi}{4}(1-2\widetilde{h}^r_i)\bigg)|\{g^r_i\}\rangle.
\end{split}
\end{equation}
We find that in the particular construction Eq.~\eqref{Eq.32}, the number of inter-site hidden spin is the same as the number of on-site hidden spin, for an arbitrary cocycle Hamiltonian.  The relation between $\{\widehat{g}^r_i\}$ and $\{g^r_i\}$ depends on the cocycle parameters $P_{ij}$, as per Eq.~\eqref{Eq.90}.

\subsection{RBM Construction for $Z^{q-1}X Z^{q-1}$ Model}

To exemplify our RBM construction, we apply  the above algorithm to the stabilizer code $Z^{q-1}X Z^{q-1}$ for an arbitrary cocycle. Another example is discussed in App.~\ref{app.J}. The $Z^{q-1}X Z^{q-1}$ model corresponds to the cocycle Hamiltonian with $P_{ij}=1$ for any $1\leq i<j\leq q$. 

The Hamiltonian of the $Z^{q-1} X Z^{q-1}$ model is 
\begin{equation}
\begin{split}
&H_{Z^{q-1}XZ^{q-1}}\\&=-\sum_{r=0}^{L-1} \Bigg(\prod_{i=1}^{q-1}Z^r_i X_{q}^r \prod_{i=1}^{q-1} Z^{r+1}_i+\prod_{i=2}^{q}Z^r_{i} X^{r+1}_1 \prod_{i=2}^{q}Z^{r+1}_i\\&~~~+ \sum_{s=3}^{q}\bigg( \prod_{i=s}^{q} Z^r_{i} \prod_{j=1}^{s-2} Z^{r+1}_j X_{s-1}^{r+1} \prod_{k=s}^{q} Z^{r+1}_{k} \prod_{l=1}^{s-2} Z^{r+2}_l \bigg)~ \Bigg).
\end{split}
\end{equation}
Its ground state is 
\begin{equation}\label{Eq.GSZnXZn}
\begin{split}
&|\mathrm{GS}\rangle_{Z^{q-1}XZ^{q-1}}=\\&\sum_{\{g^r_i\}}\prod_{r=0}^{L-1}\exp\bigg(i\pi \sum_{1 \leq j<i\leq q}(g_i^{r}-g^{r-1}_{i})g^r_{j}\bigg)|\{g^r_i\}\rangle.
\end{split}
\end{equation}
The the $q\times q$ $\Gamma$ matrix and the $(q-1)\times (q-1)$ $\Lambda$ matrix are
\begin{equation}\label{Eq.LambdaMatrixZZXZZ}
\Gamma=
\begin{pmatrix}
0 & 0 & \cdots && 0\\
1 & & && 0\\
1 & 1 & & &0\\
\vdots & \vdots & \ddots & & \vdots \\
1 & 1 & \cdots &1& 0
\end{pmatrix}, 
~~
\Lambda= 
\begin{pmatrix}
1 & 1 & \cdots& 1 \\
& 1 & \cdots& 1\\
 &  & \ddots & \vdots\\
 &  &  & 1
\end{pmatrix}.
\end{equation}
To transform the $\Gamma$ matrix to the form as in Eq.~\eqref{Eq.24}, we switch the rows using
\begin{eqnarray}
G^{T}= G_1(q-1, q)\cdots G_1(1,2).
\end{eqnarray}
The visible spins transform as 
\begin{eqnarray}
\begin{pmatrix}
g^r_{1}\\
g^r_{2}\\
\vdots\\
g^r_{q-1}\\
g^r_{q}
\end{pmatrix}
\to 
\begin{pmatrix}
\widehat{g}^r_{1}\\
\widehat{g}^r_{2}\\
\vdots\\
\widehat{g}^r_{q-1}\\
\widehat{g}^r_{q}
\end{pmatrix}
= 
G^{-1}\cdot 
\begin{pmatrix}
g^r_{1}\\
g^r_{2}\\
\vdots\\
g^r_{q-1}\\
g^r_{q}
\end{pmatrix}
=
\begin{pmatrix}
g^r_{2}\\
g^r_{3}\\
\vdots\\
g^r_{q}\\
g^r_{1}
\end{pmatrix}.\\
~~~\nonumber
\end{eqnarray}
The $\Gamma$ matrix transforms as 
\begin{eqnarray}
\Gamma\to \widehat{\Gamma}= G^T\cdot \Gamma\cdot G= 
\begin{pmatrix}
1 & & && 0\\
1 & 1 & & &0\\
\vdots & \vdots & \ddots & & \vdots \\
1 & 1 & \cdots &1& 0\\
0 & 0 & \cdots & 0& 0
\end{pmatrix}.
\end{eqnarray}
All the $q-1$ rows in the top $(q-1)\times q$ block of $\widehat{\Gamma}$ are independent,
\begin{eqnarray}
\mathrm{rank}(\widehat{\Gamma})=\mathrm{rank}(\Gamma)=\mathrm{rank}(\Lambda)=q-1.
\end{eqnarray}
As a result, the exponents in Eq.~\eqref{Eq.GSZnXZn} can be written as 
\begin{equation}
\begin{split}
\sum_{1 \leq j<i\leq q}(g_i^{r}-g^{r-1}_{i})g^r_{j}=&-\sum_{i=1}^{q-1} \Sym(g_{i+1}^{r-1}, g_{i}^r, g_{i-1}^r, \ldots , g_1^r)\\
&+\sum_{i=1}^{q-1}\mathrm{Sym}(g_{i+1}^{r}, g_{i}^r, g_{i-1}^r, \ldots , g_1^r).
\end{split}
\end{equation}
On RHS of the equality, the first $q-1$ symmetric functions are inter-site terms. Using Eq.~\eqref{Eq.Identities_n} we introduce $q-1$ hidden spins of type-$h$ contributing to $2^{\mathrm{rank}(\Gamma)}=2^{q-1}$ bond dimension of the RBM-MPS. The remaining   $q-1$ symmetric functions only contain on-site quadratic terms. Using Eq.~\eqref{Eq.Identities_n}, we introduce $q-1$ hidden spins of type-$\widetilde{h}$. Combining these two operations, we have:
\begin{widetext}
\begin{equation}
\begin{split}
|\mathrm{GS}\rangle_{Z^{q-1}XZ^{q-1}}=
&\sum_{\{g_i^r\}} \sum_{\substack{\{h_1^r\}\ldots \{h^r_{q-1}\}}} \prod_{r=0}^{L-1} \exp\bigg(-i\frac{\pi}{2}\sum_{i=1}^{q-1} (1-2h_{i}^r)(g_{i+1}^{r-1}+\sum_{j=1}^{i}g_j^r)+i\frac{\pi}{4}\sum_{i=1}^{q-1} (1-2h_{i}^r)\bigg)	\\
&\times\sum_{\{\widetilde{h}_{i}^r\} }\exp\bigg(i\frac{\pi}{2}\sum_{i=1}^{q-1}(1-2\widetilde{h}_{i}^r)\sum_{j=1}^{i+1}g^r_{j}-i\frac{\pi}{4}\sum_{i=1}^{q-1}(1-2\widetilde{h}^r_{i})\bigg)|\{g^r_i\}\rangle.
\end{split}
\end{equation}
This RBM can be casted into an rank-1 MPS, and the matrix elements of the RBM-MPS are:
\begin{equation}\label{Eq.RBMMPSZnXZn}
\begin{split}
T^{g^r_1, \ldots , g^r_{q}}_{h^{r}_1\ldots h^r_{q-1}, h^{r+1}_1\ldots h^{r+1}_{q-1}}
=&\exp\bigg(-i\frac{\pi}{2}\sum_{i=1}^{q-1} (1-2h_{i}^r)(\sum_{j=1}^{i}g_j^r)- i\frac{\pi}{2}\sum_{i=1}^{q-1} (1-2h_{i}^{r+1})g_{i+1}^{r}+i\frac{\pi}{4}\sum_{i=1}^{q-1} (1-2h_{i}^r)\bigg)	\\
&\times  \sum_{\{\widetilde{h}_{i}^r\} }\exp\bigg(i\frac{\pi}{2}\sum_{i=1}^{q-1}(1-2\widetilde{h}_{i}^r)\sum_{j=1}^{i+1}g^r_{j}-i\frac{\pi}{4}\sum_{i=1}^{q-1}(1-2\widetilde{h}^r_{i})\bigg).
\end{split}
\end{equation}
\end{widetext}
We discuss two particular cases. When $q=2$, the model corresponds to  the $ZXZ$ model. A graphical representation of the $ZXZ$ model is shown in Fig.~\ref{FigRBMZXZ}. We notice that the corresponding RBM-MPS has bond dimension $2$. In the RBM derived in Ref.~\onlinecite{PhysRevB.96.195145}, the corresponding bond dimension is $4$, which is not minimal. When $q=3$ which corresponds to the $ZZXZZ$ model, we find that the RBM-MPS matrices in Eq.~\eqref{Eq.RBMMPSZnXZn} precisely agrees with the MPS matrices in Eq.~\eqref{Eq.MPSZZXZZ}.

\begin{figure}[t]
	\centering
	\includegraphics[width=1\columnwidth]{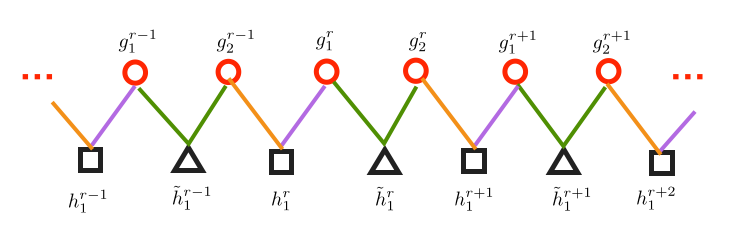}
	\caption{Graphical representation of the RBM state of the $ZXZ$ model. The red circles represent visible spins, the black rectangles represent the hidden spins connecting visible spin belonging to different unit cells, and the black triangles represent the hidden spins connecting visible spins within the same unit cell. }
	\label{FigRBMZXZ}
\end{figure}

In summary, we have shown that for cocycle Hamiltonians, the ground state can be expressed as an RBM state with the minimal RBM-MPS bond dimension. 
We further conjecture, that for an arbitrary translational invariant stabilizer code with non-degenerate ground state with PBC, if its ground state MPS matrix is of rank 1, then it is possible to express its ground state as an RBM state with the minimal RBM-MPS bond dimension matching Eq.~\eqref{eq.dim=rankt}. We leave the proof of this conjecture for future work. 

\section{Conclusion}
\label{Sec.Conclusion}

We conclude this paper by summarizing the main results for a 1D translational invariant stabilizer code with a non-degenerate ground state on PBC. 
\begin{enumerate}
	\item We have shown that a translational invariant and finitely connected RBM state can be converted to an MPS, which we dub as an RBM-MPS. We show that the non-vanishing matrix of the RBM-MPS is always of rank 1. 
	\item We provide an algorithm to determine the MPS for the stabilizer codes satisfying our assumptions given in Sec.~\ref{Sec.MPSSC}. 
	The MPS bond dimension $2^{\frac{\mathrm{rank}(\mathbf{t})}{2}}$ is proved to be the irreducible representation of the generalized Clifford algebra Eq.~\eqref{Eq.Algebra}. The $\mathbf{t}$ matrix elements can be read off from the Hamiltonian. 
	\item An upper bound for the rank of the MPS matrices is proved. For all the stabilizer codes we have explicitly considered, the upper bound is saturated. In particular, we have proved that the MPS matrices of the cocycle Hamiltonians are all of rank 1. 
	\item For the cocycle Hamiltonians, we present an explicit construction of the RBM state with the minimal RBM-MPS bond dimension. We exemplify our construction using a family of cocycle Hamiltonians explicitly. For a stabilizer code satisfying the assumptions \ref{Assumption1}, \ref{Assumption2} and \ref{Assumption3}, we conjecture that the ground state can be expressed as an RBM state and that its RBM-MPS bond dimension is the minimal one, as long as its MPS matrices are of rank 1.
\end{enumerate}

We have restricted our study to 1D stabilizer codes with a non-degenerate ground state for PBC, i.e., without symmetry breaking. Including cases with symmetry breaking could be envisioned rather simply. Consider a classical Ising model as the simplest case of such a 1D stabilizer code, which has two maximally polarized ground states.  A (trivial) MPS of bond dimension 1 can be built for each maximally polarized state. For each of these two MPS, we can derive a (trivial) RBM. We conjecture that such a construction could be extended to the more involved cases.

For future investigation, 
it would be interesting to extend the discussion of the present paper to higher dimensions. Some examples of the RBM representation of SPT states and topologically ordered states in  two  and three dimensions \cite{2018arXiv181002352L, Chen2018Equivalence, jia2018efficient} have been studied. A general understanding of RBM states for higher dimensional stabilizer codes is still missing. We hope to get some insights for higher dimensional systems by considering stabilizer codes. Following the ideas developed in this paper, we hope to address the two following challenges: 1) how to systematically derive the PEPS representations of stabilizer codes, where the Hamiltonian terms can be the mixed products of both Pauli $X$ and $Z$ operators\cite{PhysRevB.97.125102}; 2) whether a given PEPS of a stabilizer code can be represented by an efficient RBM state with the minimal RBM-PEPS bond dimension.

\bigskip

\textit{Note Added:} During the preparation of the present manuscript appeared a paper Ref.~\onlinecite{2018arXiv180908631Z} partially overlapping with our results. The authors have proposed an efficient numerical algorithm to construct efficient RBM state for an quantum stabilizer code. In particular, their RBM state of the $ZXZ$ model is identical to ours. 
	 
\section*{Acknowledgment}

Y.Z and H.H thanks the support from Physics Department of Princeton University. N.R. is supported by the grant No. ANR-17-CE30-0013-01. B.A.B. is supported by the Department of Energy Grant No. de-sc0016239, the National Science Foundation EAGER Grant No. noaawd1004957, Simons Investigator Grants No. ONRN00014-14-1-0330, and No. NSF-MRSEC DMR- 1420541, the Packard Foundation, the Schmidt Fund for Innovative Research. 

\appendix

\clearpage

\section{Conventions for MPS and Canonical MPS}
\label{app.MPS}

\subsection{Conventions for MPS and Transfer Matrix}

Since each unit cell contains $q$ spins-$\frac{1}{2}$'s, it is natural to start with the translational invariant MPS in Eq.~\eqref{Eq.MPSGS}, i.e., 
\begin{equation}
	|\mathrm{GS}\rangle=\sum_{\{g^r_i\}}\Tr\bigg(\prod_{r=0}^{L-1}T^{g^r_1\ldots g^r_q}\bigg)|\{g^r_i\}\rangle.
\end{equation} 
For convenience, we introduce the notation of the physical operators acting on the MPS tensors. Denoting $X^{r}_{i}$ and $Z^{r}_{i}$ as the Pauli $X$ and $Z$ operators acting on $i$-th orbital ($i=1, \ldots ,q$) in the $r$-th unit cell, their action on the MPS matrices are defined as:
\begin{equation}
	X^{r}_{i}\circ T^{g^{r'}_1\ldots g^{r'}_i\ldots g^{r'}_q}=
	\begin{cases}
		T^{g^{r'}_1\ldots (1-g^{r'}_i)\ldots g^{r'}_q} & \mathrm{if}~r'=r\\
		T^{g^{r'}_1\ldots g^{r'}_i\ldots g^{r'}_q} & \mathrm{if}~r'\neq r,
	\end{cases}
\end{equation}
and 
\begin{equation}
	Z^{r}_{i}\circ T^{g^{r'}_1\ldots g^{r'}_i\ldots g^{r'}_q}= (-1)^{\delta_{rr'}g^{r'}_i} T^{g^{r'}_1\ldots g^{r'}_i\ldots g^{r'}_q}.
\end{equation}
For other, more complex operators, the notation $\circ$ can be naturally generalized. 

To make the equations more compact, let $\mathbf{h}_i\in \{1,..., D\}$ be the virtual indices of the MPS matrices, where $D$ is the bond dimension. Notice that the bold font $\mathbf{h}_i$ is different from the $\Z_2$ valued virtual indices $h$'s in the main text. For instance, the MPS matrix elements of Eq.~\eqref{Eq.MPSequations} become $T^{g^r_1g^r_2g^r_3}_{h_1h_2, h_3h_4}\equiv (T^{g^r_1g^r_2g^r_3})_{\mathbf{h}_1, \mathbf{h}_2}$, so we identify $\mathbf{h}_1$ and $ \mathbf{h}_2$ as the composite of $\Z_2$ valued $h$ indices, i.e., $h_1h_2$ and $h_3h_4$ respectively. 
Given the MPS matrix elements $(T^{g_1^{r} \ldots g_q^{r}})_{\mathbf{h}_1, \mathbf{h}_2}$, where $\mathbf{h}_1, \mathbf{h}_2\in \{1,..., D\}$ are the left and right virtual indices, we can construct the MPS transfer matrix $\mathbb{T}_{\mathbf{h}_1\mathbf{h}_3, \mathbf{h}_2 \mathbf{h}_4}$ by contracting over the physical indices,
\begin{equation}\label{eq.transfermatrixdefinition}
	{\mathbb{T}}_{\mathbf{h}_1\mathbf{h}_3,\mathbf{h}_2\mathbf{h}_4}=\sum_{g^r_1\ldots g^r_q} ({T}^{g^r_1\ldots g^r_q})_{\mathbf{h}_1,\mathbf{h}_2} ({T}^{g^r_1\ldots g^r_q})^*_{\mathbf{h}_3,\mathbf{h}_4} .
\end{equation}
Here, $\mathbf{h}_1\mathbf{h}_3$ is regarded as a composite left virtual index of the transfer matrix, of dimension $D^2$. The same applies to $\mathbf{h}_2\mathbf{h}_4$. The transfer matrix ${\mathbb{T}}$ is a $D^2\times D^2$ matrix. 

\subsection{Review of Canonical MPS}
\label{app.CanonicalForm}

We now review the definition and the properties of canonical MPS, and apply the canonical MPS to stabilizer codes. The MPS matrix $\breve{T}^{g^r_1\ldots g^r_q}$ is called ``canonical" if its transfer matrix satisfies:
\begin{equation}\label{Eq.28}
	\begin{split}
		\sum_{\mathbf{h}_2} \left( \breve{\mathbb{T}} \right)_{\mathbf{h}_1\mathbf{h}_3,\mathbf{h}_2\mathbf{h}_2} &= \delta_{\mathbf{h}_1\mathbf{h}_3},\\
		\sum_{\mathbf{h}_1,\mathbf{h}_3} \Lambda_{\mathbf{h}_1\mathbf{h}_3} \left( \breve{\mathbb{T}} \right)_{\mathbf{h}_1\mathbf{h}_3,\mathbf{h}_2\mathbf{h}_4} &= \Lambda_{\mathbf{h}_2\mathbf{h}_4},
	\end{split}
\end{equation}
where $\breve{\mathbb{T}}$ is the transfer matrix of $\breve{T}^{g^r_1\ldots g^r_q}$, and $\Lambda$ is a \emph{full-rank} diagonal matrix whose diagonal elements are the entanglement spectrum of a single cut. 
In Ref.~\onlinecite{2006quant.ph..8197P}, it was shown that a generic MPS matrix $T^{g^r_1\ldots g^r_q}$ on an open chain can be mapped to the canonical form $\breve{T}^{g^r_1\ldots g^r_q}$ via a similarity transformation
\begin{equation}\label{eq.SimilarityTransformation}
	T^{g^r_1\ldots g^r_q}= S\cdot \breve{T}^{g^r_1\ldots g^r_q}\cdot S^{-1},
\end{equation}
where $S$ is an invertible matrix. We use $\breve{\ }$ to denote the canonical form of the MPS matrix and the MPS transfer matrix throughout the appendix.

In Ref.~\onlinecite{PhysRevB.94.075151}, it was proved that when there is a non-degenerate ground state on any compact space, the entanglement spectrum of a stabilizer code ground state is flat. The reduced density matrices are, in fact, projectors. Their original proof was formulated in 2 spatial dimensions, but it can be directly generalized to arbitrary dimensions. See Ref.~\onlinecite{2017arXiv171001744M} for the application to 3 spatial dimensions. Here we apply their conclusion to the case of 1 spatial dimension. Hence, the entanglement spectrum of a 1D stabilizer code with PBC is flat.

The reduced density matrix on a local and contractible region of a gapped state should not depend on the boundary condition far away from the local region. Thus the entanglement spectrum does not depend on the boundary condition either. Thus for the 1D stabilizer code with OBC, the entanglement spectrum is flat. Hence $\Lambda$ in Eq.~\eqref{Eq.28} is also flat for one of the ground states with OBC. Since $\Lambda$ is full-rank, there are no zero diagonal elements in $\Lambda$ and $\Lambda$ is proportional to an identity matrix. Hence the canonical MPS of a stabilizer code satisfies the following conditions
\begin{equation}\label{Eq.Canonical}
	\begin{split}
		\sum_{\mathbf{h}_2} \left( \breve{\mathbb{T}} \right)_{\mathbf{h}_1\mathbf{h}_3,\mathbf{h}_2\mathbf{h}_2} &= \delta_{\mathbf{h}_1\mathbf{h}_3},\\
		\sum_{\mathbf{h}_1} \left( \breve{\mathbb{T}} \right)_{\mathbf{h}_1\mathbf{h}_1,\mathbf{h}_2\mathbf{h}_4} &= \delta_{\mathbf{h}_2\mathbf{h}_4}.
	\end{split}
\end{equation}
The two conditions in Eq.~\eqref{Eq.Canonical} are graphically represented in Fig.~\ref{FigCnonical}. 

Hence we can use Eq.~\eqref{Eq.Canonical} to solve for the MPS with OBC. By Assumption~\ref{Assumption3}, the MPS matrices for the OBC shall also be the MPS matrices for the PBC. 

\begin{figure}
	\centering
	\includegraphics[width=1\columnwidth]{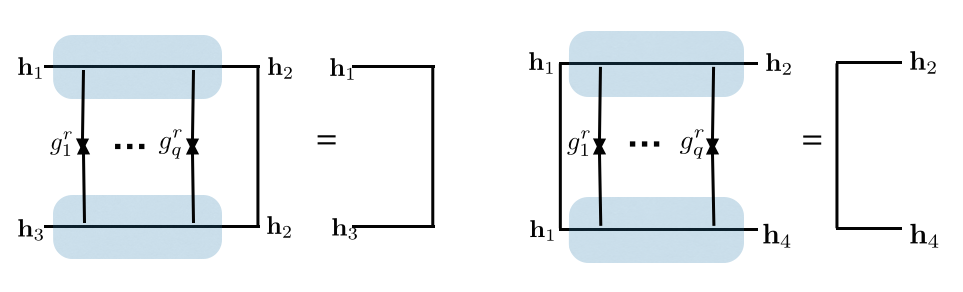}
	\caption{Graphical representation of Eq.~\eqref{Eq.Canonical}. }
	\label{FigCnonical}
\end{figure}

\section{Correlation Functions and Transfer Matrix Eigenvalues}
\label{app.CorrelationAndTransferMatrix}

In this appendix, we derive the eigenvalue structure of the transfer matrix of a general translational invariant stabilizer code. As we will prove, there is only one nonzero eigenvalue of the MPS transfer matrix, obtained by Jordan decomposition. Moreover, a finite power of the MPS transfer matrix can be decomposed as a tensor product of two vectors. The lemmas and theorems will be used in App.~\ref{app.Deriving0}.

\begin{lemma}\label{LemmaC1}
	Suppose an operator $\mathbf{O}$ anti-commutes with some of the Hamiltonian terms in Eq.~\eqref{Eq.GeneralStabilizerCodes}, i.e., $H=-\sum_{r=0}^{L-1}\sum_{\alpha=1}^t \Ocal^r_{\alpha}$, its expectation value of the ground state of Eq.~\eqref{Eq.GeneralStabilizerCodes} satisfies
	\begin{equation}
		\bra{\mathrm{GS}} \mathbf{O} \ket{\mathrm{GS}} =0.
	\end{equation}
\end{lemma}

\begin{proof}
	Without loss of generality, suppose $\mathbf{O}$ anti-commutes with $\Ocal^0_{1}$ in Eq.~\eqref{Eq.GeneralStabilizerCodes}. Since the ground state $\ket{\mathrm{GS}}$ satisfies the stabilizer condition Eq.~\eqref{Eq.GSproperty}, we have
	\begin{equation}
		\begin{split}
			\bra{\mathrm{GS}} \mathbf{O} \ket{\mathrm{GS}} 
			=& \bra{\mathrm{GS}} \mathbf{O} \Ocal^0_{1} \ket{\mathrm{GS}} 	\\
			=& - \bra{\mathrm{GS}} \Ocal^0_{1} \mathbf{O} \ket{\mathrm{GS}} \\
			=& - \bra{\mathrm{GS}} \mathbf{O} \ket{\mathrm{GS}}.
		\end{split}
	\end{equation}
	Hence
	\begin{equation}
		\bra{\mathrm{GS}} \mathbf{O} \ket{\mathrm{GS}} = 0.
	\end{equation}
\end{proof}

Consider two operators $\sigma_i$, $i=1,2$. We denote $p_1$ (resp. $p_2$) the support of $\sigma_1$ (resp. $\sigma_2$) on the unit cells $r_1\leq r\leq r_1+p_1-1$ (resp. $r_2\leq r\leq r_2+p_2-1$). We define the distance $d(\sigma_1, \sigma_2)$ of the two operators as the number of unit cells between the two operators plus one, i.e., 
\begin{equation}
	d(\sigma_1, \sigma_2)= 
	\begin{cases}
		r_2-r_1-p_1+1&, ~~r_2\geq r_1+p_1\\
		r_1-r_2-p_2+1&, ~~r_1\geq r_2+p_2\\
		0 &, ~~ r_1+p_1> r_2> r_1-p_2.
	\end{cases}
\end{equation}
In particular, when two operators overlap even only on  one site, their distance is zero. When the distance of two operators $\sigma_1$ and $\sigma_2$ are larger than $P$, where $P$ is the range of another operator $\Ocal$, then $\Ocal$ can not overlap simultaneously with $\sigma_1$ and $\sigma_2$. 

\begin{lemma}\label{LemmaC2}
	Suppose $\sigma_1$ and $\sigma_2$ are products of Pauli matrices supported on different regions of distance larger than the maximal interaction range, i.e.:
	\begin{equation}
		d(\sigma_1,\sigma_2) > \max\{P_{1}, \ldots , P_{t}\},
	\end{equation}
	where $P_\alpha$ is the support of $\alpha$-th type of the Hamiltonian term $\Ocal_\alpha^r$. Then, their expectation values satisfy
	\begin{equation}\label{Eq.16}
		\begin{split}
			\bra{\mathrm{GS}} \sigma_1 \sigma_2 \ket{\mathrm{GS}} 
			= \bra{\mathrm{GS}}\sigma_1 \ket{\mathrm{GS}} \bra{\mathrm{GS}} \sigma_2 \ket{\mathrm{GS}}.
		\end{split}
	\end{equation}
\end{lemma}

\begin{proof}
	$\sigma_1$ and $\sigma_2$ either commute or anti-commute with the Hamiltonian terms, because $\sigma_1$, $\sigma_2$ and stabilizer operators are all products of Pauli matrices. We prove this lemma case by case:
	\begin{enumerate}
		\item $\sigma_1$ and $\sigma_2$ both commute with all stabilizer operators. 
		\begin{equation}\label{Eq.13}
			[H, \sigma_i]=0, ~~ i=1, 2.
		\end{equation}
		Hence for any excited eigenstate $|E, k\rangle$ of the Hamiltonian $H$, i.e., $H|E, k\rangle=E|E, k\rangle$ ($E$ is the energy and $k$ labels the degeneracy within the energy eigenspace), $\sigma_i |E, k\rangle$ is also an excited eigenstate of $H$. One can see this from Eq.~\eqref{Eq.13}: $[H, \sigma_i]|E, k\rangle=0$ for $i=1, 2$, which implies $\sigma_i |E, k\rangle$ is an energy eigenstate of $H$ with energy $E$. So
		\begin{equation}\label{Eq.15}
			\langle \mathrm{GS}| \sigma_i |E, k\rangle=0, ~~ i=1, 2.
		\end{equation}
		Then
		\begin{equation}\label{Eq.18}
			\begin{split}
				&\bra{\mathrm{GS}}\sigma_1 \ket{\mathrm{GS}} \bra{\mathrm{GS}} \sigma_2 \ket{\mathrm{GS}} \\&= \bra{\mathrm{GS}}\sigma_1 \bigg(1-\sum_{E, k}\ket{E, k} \bra{E, k}\bigg) \sigma_2 \ket{\mathrm{GS}} \\
				&= \bra{\mathrm{GS}} \sigma_1 \sigma_2 \ket{\mathrm{GS}} -\sum_{E, k}\bra{\mathrm{GS}}\sigma_1 \ket{E, k} \bra{E, k}\sigma_2 \ket{\mathrm{GS}} \\
				&=\bra{\mathrm{GS}} \sigma_1 \sigma_2 \ket{\mathrm{GS}},
			\end{split}
		\end{equation}
		where in the first equality, we have used Assumption \ref{Assumption1}, and in the last equality, we have used Eq.~\eqref{Eq.15}. Hence Eq.~\eqref{Eq.16} holds true in this case. 
		
		\item $\sigma_1$ commutes with all stabilizer operators while $\sigma_2$ anti-commutes with some of the stabilizer operators. Hence, $\sigma_2$ and $\sigma_1\sigma_2$ both satisfy Lemma \ref{LemmaC1}. Their expectation values are both 0:
		\begin{equation}
			\bra{\mathrm{GS}} \sigma_2 \ket{\mathrm{GS}} = 0, \quad
			\bra{\mathrm{GS}} \sigma_1\sigma_2 \ket{\mathrm{GS}} = 0.
		\end{equation}
		Therefore, Eq.~\eqref{Eq.16} holds true in this case.
		
		\item $\sigma_1$ anti-commutes with some of the stabilizer operators while $\sigma_2$ commutes with all stabilizer operators. This is the same situation as the last one. Both sides of Eq.~\eqref{Eq.16} vanish.
		
		\item $\sigma_1$ and $\sigma_2$ both anti-commute with some of stabilizer operators. Using Lemma \ref{LemmaC1}, their expectation values both vanish. There does not exist a stabilizer operator which overlaps simultaneously with $\sigma_1$ and $\sigma_2$, because $\sigma_1$ and $\sigma_2$ are separated with a distance larger than the maximal interaction range $\max\{P_{1}, \ldots , P_{t}\}$. Hence, $\sigma_1 \sigma_2$ still anti-commutes with some of the stabilizer operators. So both sides of Eq.~\eqref{Eq.16} vanish.
	\end{enumerate}
	This completes the proof.
\end{proof}

\begin{theorem}\label{TheoremB3}
	Suppose two arbitrary operators $\mathbf{O}$ and $\widetilde{\mathbf{O}}$ are supported on different regions separated by a distance larger than $\max\{P_{1}, \ldots , P_{t}\}$. Then we have
	\begin{equation}
		\bra{\mathrm{GS}} \mathbf{O} \widetilde{\mathbf{O}} \ket{\mathrm{GS}} - \bra{\mathrm{GS}} \mathbf{O} \ket{\mathrm{GS}} \bra{\mathrm{GS}} \widetilde{\mathbf{O}} \ket{\mathrm{GS}} = 0.
	\end{equation}
\end{theorem}

\begin{proof}
	First we can expand the two operators as the summations of the products of Pauli matrices:
	\begin{equation}
		\begin{split}
			\mathbf{O} = \sum_{i} \phi_i \sigma_i	\\
			\widetilde{\mathbf{O}} = \sum_{j} \theta_j \widetilde{\sigma_j},	\\
		\end{split}
	\end{equation}
	where the terms $\sigma_i$ and $\widetilde{\sigma_j}$ are products of Pauli matrices supported in two separated regions, and $\phi_i$ and $\theta_j$ are complex coefficients. Recall our assumption that $\mathbf{O}$ and $\widetilde{\mathbf{O}}$ are supported on different regions separated by a distance larger than the maximal interaction range $\max\{P_{1},\ldots, P_{t}\}$. Then, $\sigma_i$ and $\widetilde{\sigma_j}$ are also supported on regions with a distance larger than $\max\{P_{1},\ldots, P_{t}\}$. Hence, $\sigma_i$ and $\widetilde{\sigma_j}$ satisfy Lemma \ref{LemmaC2}. Therefore
	\begin{equation}
		\begin{split}
			\bra{\mathrm{GS}} \mathbf{O} \widetilde{\mathbf{O}} \ket{\mathrm{GS}}	
			=& \sum_{i,j} \phi_i \theta_j \bra{\mathrm{GS}} \sigma_i \widetilde{\sigma_j} \ket{\mathrm{GS}}	\\
			=& \sum_{i,j} \phi_i \theta_j \bra{\mathrm{GS}} \sigma_i \ket{\mathrm{GS}} \bra{\mathrm{GS}} \widetilde{\sigma_j} \ket{\mathrm{GS}}	\\
			=& \bra{\mathrm{GS}} \mathbf{O} \ket{\mathrm{GS}} \bra{\mathrm{GS}} \widetilde{\mathbf{O}} \ket{\mathrm{GS}}.
		\end{split}
	\end{equation}
	This completes the proof.
\end{proof}

\begin{theorem}\label{TheoremC3}
	Let $T^{g^r_1 \ldots g^r_q}_{\mathbf{h}_1, \mathbf{h}_2}$ be the MPS matrix element of a translational invariant stabilizer code where $\mathbf{h}_1$ and $\mathbf{h}_2$ are the virtual indices, and $\mathbb{T}_{\mathbf{h}_1\mathbf{h}_3, \mathbf{h}_2\mathbf{h}_4}$ be the MPS transfer matrix of $T$ defined in Eq.~\eqref{eq.transfermatrixdefinition}. Then $\mathbb{T}$ has only 1 nonzero eigenvalue. 
\end{theorem}

\begin{proof}
	For convenience, we introduce the notation:
	\begin{equation}
		\begin{split}
			&\mathbb{T}[\mathbf{O}^r]_{\mathbf{h}_1\mathbf{h}_3, \mathbf{h}_2\mathbf{h}_4}	\\
			=&\sum_{g^r_1 \ldots g^r_q}( \mathbf{O}^r\circ T^{g^r_1\ldots g^r_q})_{\mathbf{h}_1, \mathbf{h}_2} (T^{g^r_1\ldots g^r_q})^{\star}_{\mathbf{h}_3, \mathbf{h}_4}.
		\end{split}
	\end{equation}
	Moreover, the transfer matrix $\mathbb{T}$ can always be decomposed into Jordan blocks: 
	\begin{equation}
		\mathbb{T}= U (\mathcal{P}_{\lambda_0}+ \mathcal{P}_{\lambda_1}+ \mathcal{P}_{\lambda_2}+\cdots ) U^{-1},
	\end{equation}
	where $|\lambda_0|>|\lambda_1|>|\lambda_2|>\cdots$ are the eigenvalues of $\mathbb{T}$, and $\mathcal{P}_{\lambda_i}$ is the corresponding Jordan block. By a proper scaling of $\mathbb{T}$, we let $\lambda_0=1$. Using this normalization, $\mathcal{P}_{\lambda_0}\equiv \mathcal{P}_{1}$ is non-degenerate due to the gap and non-degeneracy of the ground state.\cite{2006quant.ph..8197P, zeng2015quantum} Without loss of generality, let us consider the special basis of the virtual indices such that $U$ is an identity matrix, i.e.,
	\begin{equation}
		\mathbb{T}= \mathcal{P}_{\lambda_0}+ \mathcal{P}_{\lambda_1}+ \mathcal{P}_{\lambda_2}+\cdots.
	\end{equation}
	Suppose we have two operators $\mathbf{O}^r$ and $\widetilde{\mathbf{O}}^{r+l}$ with a sufficiently large (but finite) $l$ such that they satisfy Theorem \ref{TheoremB3}. The expectation value of $\mathbf{O}^r$ and $\widetilde{\mathbf{O}}^{r+l}$ can be written in terms of transfer matrices:
	\begin{equation}\label{Eq.19}
		\begin{split}
			\bra{\mathrm{GS}} \mathbf{O}^r \widetilde{\mathbf{O}}^{r+l} \ket{\mathrm{GS}}	
			=&\frac{\Tr\left[\mathbb{T}^{r}(\mathbb{T}[\mathbf{O}^r])\mathbb{T}^{l-1}(\mathbb{T}[\widetilde{\mathbf{O}}^{r+l}])\mathbb{T}^{L-1-r-l}\right]}{\Tr\left(\mathbb{T}^L\right)}	\\
			=&\frac{\Tr\left[\mathbb{T}^{L-l-1}(\mathbb{T}[\mathbf{O}^r])\mathbb{T}^{l-1}(\mathbb{T}[\widetilde{\mathbf{O}}^{r+l}])\right]}{\Tr\left(\mathbb{T}^L\right)}.
		\end{split}
	\end{equation}
	By Assumption \ref{Assumption2} in the beginning of Sec.~\ref{Sec.MPSSC}, the MPS matrices is independent of the system size when $L$ is sufficient large. For simplicity, let us take the limit:
	\begin{equation}
		\lim\limits_{L \rightarrow \infty} \mathbb{T}^{L-l-1}=\lim\limits_{L \rightarrow \infty} \mathbb{T}^{L} = \mathcal{P}_{1}.
	\end{equation}
	Eq.~\eqref{Eq.19} then simplifies to
	\begin{equation}\label{eq.simplifyCorrealtionFunctions}
		\begin{split}
			\bra{\mathrm{GS}} \mathbf{O}^r \widetilde{\mathbf{O}}^{r+l} \ket{\mathrm{GS}}	
			=&\frac{\Tr\left[\mathcal{P}_1(\mathbb{T}[\mathbf{O}^r])\mathbb{T}^{l-1}(\mathbb{T}[\widetilde{\mathbf{O}}^{r+l}])\right]}{\Tr\left(\mathcal{P}_1\right)}	\\
			=&\Tr\left[\mathcal{P}_1(\mathbb{T}[\mathbf{O}^r])\mathbb{T}^{l-1}(\mathbb{T}[\widetilde{\mathbf{O}}^{r+l}])\right].
		\end{split}
	\end{equation}
	Using the Jordan blocks decomposition of $\mathbb{T}$ ($\lambda_0=1$), we have
	\begin{equation}
		\mathbb{T}^{l-1} = \mathcal{P}_1 + \sum_{|\lambda| < 1} \mathcal{P}^{l-1}_{\lambda}.
	\end{equation}
	Substituting to Eq.~\eqref{eq.simplifyCorrealtionFunctions}, we have
	\begin{equation}
		\begin{split}
			&\bra{\mathrm{GS}} \mathbf{O}^r \widetilde{\mathbf{O}}^{r+l} \ket{\mathrm{GS}}	
			\\=&\Tr\left( \mathcal{P}_1 (\mathbb{T}[\mathbf{O}^r]) (\mathcal{P}_1 + \sum_{|\lambda| < 1} \mathcal{P}^{l-1}_{\lambda}) (\mathbb{T}[\widetilde{\mathbf{O}}^{r+l}]) \right)	\\
			=&\Tr\left[\mathcal{P}_1(\mathbb{T}[\mathbf{O}^r])\mathcal{P}_1(\mathbb{T}[\widetilde{\mathbf{O}}^{r+l}])\right] \\
			&+\sum_{|\lambda| < 1} \Tr\left[ \mathcal{P}_1 (\mathbb{T}[\mathbf{O}^r]) \mathcal{P}^{l-1}_{\lambda} (\mathbb{T}[\widetilde{\mathbf{O}}^{r+l}]) \right].
		\end{split}
	\end{equation}
	Since $\mathcal{P}_1$ is 1 dimensional (unique gapped ground state),
	\begin{equation}
		\begin{split}
			&\Tr\left[\mathcal{P}_1(\mathbb{T}[\mathbf{O}^r])\mathcal{P}_1(\mathbb{T}[\widetilde{\mathbf{O}}^{r+l}])\right] \\
			=&\Tr\left(\mathcal{P}_1\mathbb{T}[\mathbf{O}^r]\right) \Tr\left(\mathcal{P}_1\mathbb{T}[\widetilde{\mathbf{O}}^{r+l}]\right)	\\
			=&\bra{\mathrm{GS}} \mathbf{O}^r \ket{\mathrm{GS}}\bra{\mathrm{GS}} \widetilde{\mathbf{O}}^{r+l} \ket{\mathrm{GS}}.
		\end{split}
	\end{equation}
	Hence
	\begin{equation}
		\begin{split}
			&\bra{\mathrm{GS}} \mathbf{O}^r \widetilde{\mathbf{O}}^{r+l} \ket{\mathrm{GS}}	\\
			=& \bra{\mathrm{GS}} \mathbf{O}^r \ket{\mathrm{GS}}\bra{\mathrm{GS}} \widetilde{\mathbf{O}}^{r+l} \ket{\mathrm{GS}}	\\
			&+ \sum_{|\lambda| < 1} \Tr\left[ \mathcal{P}_1 (\mathbb{T}[\mathbf{O}^r]) \mathcal{P}^{l-1}_{\lambda} (\mathbb{T}[\widetilde{\mathbf{O}}^{r+l}]) \right].
		\end{split}
	\end{equation}
	Theorem \ref{TheoremB3} implies that:
	\begin{equation}
		\sum_{\lambda \neq 1} \Tr\left[ \mathcal{P}_1(\mathbb{T}[\mathbf{O}^r]) \mathcal{P}^{l-1}_{\lambda} (\mathbb{T}[\widetilde{\mathbf{O}}^{r+l}]) \right] = 0
	\end{equation}
	for \emph{any} operators $\mathbf{O}^r$ and $\widetilde{\mathbf{O}}^{r+l}$ with a sufficiently large but finite $l$. Then the only possibility is that for all $\lambda \neq 1$, $\lambda = 0$. In other words, the only nonzero eigenvalue of $\mathbb{T}$ is 1. This completes the proof.
\end{proof}

We numerically checked the $Z^{q-1}X Z^{q-1}$ models with $2\leq q \leq 6$ and found that the transfer matrix indeed has only 1 nonzero eigenvalue. 

\begin{lemma}\label{lemma.ZeroJordanBlock}
	For a Jordan block $\mathcal{P}_0$ of size $m \times m$ with zero diagonal elements, then
	\begin{equation}
		\left( \mathcal{P}_0 \right)^n=0,
	\end{equation}
	where the integer $n\ge m$. 
\end{lemma}

\begin{proof}
	In terms of matrix elements, $\mathcal{P}_0$ is:
	\begin{equation}
		\mathcal{P}_{0} = \left( \begin{matrix}
			0	&	1	&	0	&	\ldots	&	0 	\\
			0	&	0	&	1	&	\ldots 	&	0	\\
			\ldots &&\ldots &&	\\
			0	&	0	&	0	&	\ldots 	&	1	\\
			0	&	0	&	0	&	\ldots 	&	0	\\
		\end{matrix}\right)
	\end{equation}
	Denote $e_{i}$ as the vector of size $m$ whose $i$-th entry is 1 and 0 otherwise. Then we can show that:
	\begin{equation}
		\begin{split}
			\mathcal{P}_0 \cdot e_1 &= 0,	\\
			\mathcal{P}_0 \cdot e_i &= e_{i-1}, \quad\forall i=2,2,\ldots,m.	\\
		\end{split}
	\end{equation}
	Hence, for any vector $e_i$ ($i=1,2,\ldots,m$), we can prove that:
	\begin{equation}
		\left(\mathcal{P}_0\right)^m \cdot e_i = \left(\mathcal{P}_0\right)^{m-1} \cdot e_{i-1}  \ldots = \left(\mathcal{P}_0\right)^{m-i+1} \cdot e_{1} = 0
	\end{equation}
	Therefore, we conclude that:
	\begin{equation}
		\left(\mathcal{P}_0\right)^m = 0.
	\end{equation}
	For any integer $n\ge m$, we also have:
	\begin{equation}
		\left(\mathcal{P}_0\right)^{n} = 0.
	\end{equation}
\end{proof}

\begin{theorem}\label{theoremC.transfermatrixdecomposition}
	Suppose the transfer matrix $\mathbb{T}$ of size $D^2 \times D^2$ satisfies Theorem \ref{TheoremC3}. In other words, its nonzero eigenvalues contain a unique 1. Then
	\begin{equation}\label{eq.powerOfTransferMatrix}
		\left( \mathbb{T} \right)^{D^2} = u v
	\end{equation}
	for a column vector $u$ of size $D^2$ and a row vector $v$ of size $D^2$ such that
	\begin{equation}\label{eq.uvnorm1}
		v \cdot u = 1,
	\end{equation}
	where $\cdot$ represents the vector multiplication. In terms of matrix elements, Eq.~\eqref{eq.powerOfTransferMatrix} is
	\begin{equation}
		\left(\left( \mathbb{T} \right)^{D^2}\right)_{\mathbf{h}_1\mathbf{h}_3,\mathbf{h}_2\mathbf{h}_4} = u_{\mathbf{h}_1\mathbf{h}_3} v_{\mathbf{h}_2\mathbf{h}_4},
	\end{equation}
	and Eq.~\eqref{eq.uvnorm1} is
	\begin{equation}\label{eq.uvnorm1elements}
		\sum_{\mathbf{h}_1,\mathbf{h}_2=1}^{D} u_{\mathbf{h}_1\mathbf{h}_2} v_{\mathbf{h}_1\mathbf{h}_2} = 1.
	\end{equation}
\end{theorem}

\begin{proof}
	Using the fact that $\mathbb{T}$ satisfies Theorem \ref{TheoremC3}, its Jordan decomposition is:
	\begin{equation}
		\mathbb{T} = U (\mathcal{P}_1 + \mathcal{P}_0) U^{-1},
	\end{equation}
	where $\mathcal{P}_1$ is the projector into the 1 dimensional Jordan block for eigenvalue 1 and $\mathcal{P}_{0}$ is the projector into the Jordan block for eigenvalue 0. Therefore, 
	\begin{equation}
		\mathbb{T}^{D^2} = U (\mathcal{P}_1^{D^2} + \mathcal{P}_0^{D^2}) U^{-1} = U \mathcal{P}_1 U^{-1},
	\end{equation}
	where we have used Lemma \ref{lemma.ZeroJordanBlock} and the fact that the size of $\mathcal{P}_0$ is smaller than $D^2 \times D^2$:
	\begin{equation}
		\mathcal{P}_0^{D^2}=0.
	\end{equation} 
	Since the Jordan block with eigenvalue $1$ is 1 dimensional, there is only one nontrivial matrix element which locates at the diagonal of $\mathcal{P}_1$. Without loss of generality, we assume that the only nonzero element of $\mathcal{P}_1$ locates at 1-th row and 1-th column. Hence, we can write this equation in terms of matrix elements 
	\begin{equation}
		\begin{split}
			\mathbb{T}^{D^2}_{\mathbf{h}_1\mathbf{h}_3,\mathbf{h}_2\mathbf{h}_4} 
			=& U_{\mathbf{h}_1\mathbf{h}_3,1} \left( U^{-1} \right)_{1,\mathbf{h}_2\mathbf{h}_4}	\\
			\equiv& u_{\mathbf{h}_1\mathbf{h}_3} v_{\mathbf{h}_2\mathbf{h}_4},
		\end{split}
	\end{equation}
	where we define
	\begin{equation}
		\begin{split}
			u_{\mathbf{h}_1\mathbf{h}_3} &\equiv U_{\mathbf{h}_1\mathbf{h}_3,1}	\\
			v_{\mathbf{h}_2\mathbf{h}_4} &\equiv \left( U^{-1} \right)_{1,\mathbf{h}_2\mathbf{h}_4}.	\\
		\end{split}
	\end{equation}
	From these definitions
	\begin{equation}
		v\cdot u= (U^{-1}\cdot U)_{1,1}=1.
	\end{equation}
	This completes the proof. 
\end{proof}

Now we explore the properties for the canonical MPS with the tensor $\breve{T}$ and Eq.~\eqref{Eq.Canonical}.

\begin{lemma}\label{lemmaD.canonicalMPS}
	For a stabilizer code, the transfer matrix of the ground state canonical MPS satisfies Eq.~\eqref{Eq.Canonical}. We prove that:
	\begin{equation}\label{Eq.B41}
		\sum_{\mathbf{h}_1} \left( \breve{\mathbb{T}}^{n} \right)_{\mathbf{h}_1\mathbf{h}_1,\mathbf{h}_2\mathbf{h}_4} = \delta_{\mathbf{h}_2\mathbf{h}_4},	\quad
		\sum_{\mathbf{h}_2} \left( \breve{\mathbb{T}}^{n} \right)_{\mathbf{h}_1\mathbf{h}_3,\mathbf{h}_2\mathbf{h}_2} = \delta_{\mathbf{h}_1\mathbf{h}_3},
	\end{equation}
	for any integer $n>0$.
\end{lemma}

\begin{proof}
	Using the definition of the canonical MPS in Eq.~\eqref{Eq.Canonical}, we first show that
	\begin{equation}
		\begin{split}
			\sum_{\mathbf{h}_1} \left( \breve{\mathbb{T}}^{n} \right)_{\mathbf{h}_1\mathbf{h}_1,\mathbf{h}_2\mathbf{h}_4}
			=& \sum_{\mathbf{h}_1,\mathbf{h}_5,\mathbf{h}_6} \breve{\mathbb{T}}_{\mathbf{h}_1\mathbf{h}_1,\mathbf{h}_5\mathbf{h}_6} \left( \breve{\mathbb{T}}^{n-1} \right)_{\mathbf{h}_5\mathbf{h}_6,\mathbf{h}_2\mathbf{h}_4}	\\
			=& \sum_{\mathbf{h}_5,\mathbf{h}_6} \delta_{\mathbf{h}_5\mathbf{h}_6} \left( \breve{\mathbb{T}}^{n-1} \right)_{\mathbf{h}_5\mathbf{h}_6,\mathbf{h}_2\mathbf{h}_4}	\\
			=& \sum_{\mathbf{h}_1} \left( \breve{\mathbb{T}}^{n-1} \right)_{\mathbf{h}_1\mathbf{h}_1,\mathbf{h}_2\mathbf{h}_4}.
		\end{split}
	\end{equation}
	Then we repeatedly apply this equation until there is only 1 $\breve{\mathbb{T}}$ matrix.
	\begin{equation}
		\begin{split}
			\sum_{\mathbf{h}_1} \left( \breve{\mathbb{T}}^{n} \right)_{\mathbf{h}_1\mathbf{h}_1,\mathbf{h}_2\mathbf{h}_4}
			=& \sum_{\mathbf{h}_1} \left( \breve{\mathbb{T}}^{n-1} \right)_{\mathbf{h}_1\mathbf{h}_1,\mathbf{h}_2\mathbf{h}_4}	\\
			=& \sum_{\mathbf{h}_1} \left( \breve{\mathbb{T}}^{n-2} \right)_{\mathbf{h}_1\mathbf{h}_1,\mathbf{h}_2\mathbf{h}_4}	\\
			&\vdots	\\
			=& \sum_{\mathbf{h}_1} \left( \breve{\mathbb{T}} \right)_{\mathbf{h}_1\mathbf{h}_1,\mathbf{h}_2\mathbf{h}_4}	\\
			=&\delta_{\mathbf{h}_2\mathbf{h}_4}.
		\end{split}
	\end{equation}
	Similarly, we can prove the other equation. This completes the proof.
\end{proof}

\begin{lemma}\label{lemmaD.uvproportionaltoidentity}
	For a stabilizer code, the transfer matrix of its ground state canonical MPS satisfies Theorem \ref{theoremC.transfermatrixdecomposition}. We prove that the elements of $u$ and $v$ are 
	\begin{equation}\label{eq.uvproportionaltoidentity}
		u_{\mathbf{h}_1\mathbf{h}_2} = \frac{\delta_{\mathbf{h}_1\mathbf{h}_2}}{\Tr(v)},\quad
		v_{\mathbf{h}_1\mathbf{h}_2} = \frac{\delta_{\mathbf{h}_1\mathbf{h}_2}}{\Tr(u)},
	\end{equation}
	where 
	\begin{equation}
		\Tr(u) = \sum_{\mathbf{h}} u_{\mathbf{h}\mathbf{h}}, \quad
		\Tr(v) = \sum_{\mathbf{h}} v_{\mathbf{h}\mathbf{h}}.
	\end{equation}
	In other words,
	\begin{equation}\label{eq.caonicaltransfermatrixD2}
		\left( \breve{\mathbb{T}}^{D^2} \right)_{\mathbf{h}_1\mathbf{h}_3,\mathbf{h}_2\mathbf{h}_4} 
		= \frac{1}{\Tr(u)\Tr(v)} \delta_{\mathbf{h}_1\mathbf{h}_3}\delta_{\mathbf{h}_2\mathbf{h}_4} 
		= \frac{1}{D} \delta_{\mathbf{h}_1\mathbf{h}_3}\delta_{\mathbf{h}_2\mathbf{h}_4}.
	\end{equation}
\end{lemma}

\begin{proof}
	Using Theorem \ref{theoremC.transfermatrixdecomposition} for a canonical MPS, we have:
	\begin{equation}
		\left( \breve{\mathbb{T}}^{D^2} \right)_{\mathbf{h}_1\mathbf{h}_3,\mathbf{h}_2\mathbf{h}_4} = u_{\mathbf{h}_1\mathbf{h}_3} v_{\mathbf{h}_2\mathbf{h}_4}.
	\end{equation}
	Applying Lemma \ref{lemmaD.canonicalMPS} with $n=D^2$, we obtain:
	\begin{equation}
		\delta_{\mathbf{h}_2\mathbf{h}_4} = \sum_{\mathbf{h}_1}\left( \breve{\mathbb{T}}^{D^2} \right)_{\mathbf{h}_1\mathbf{h}_1,\mathbf{h}_2\mathbf{h}_4} = \Tr(u) v_{\mathbf{h}_2\mathbf{h}_4}.
	\end{equation}
	Hence, the second equation of Eq.~\eqref{eq.uvproportionaltoidentity} is proved. Similarly, we can prove the first one. Using Eqs.~\eqref{eq.uvnorm1} and \eqref{eq.uvproportionaltoidentity}, we find that 
	\begin{equation}
		v\cdot u= \frac{D}{\Tr(u)\Tr(v)}= 1.
	\end{equation} 
	This yields
	\begin{equation}
		\Tr(u)\Tr(v) = D.
	\end{equation}
	Hence, Eq.~\eqref{eq.caonicaltransfermatrixD2} is proved.
\end{proof}

Note that Lemma \ref{lemmaD.uvproportionaltoidentity} is not true for a general MPS transfer matrix. Indeed, using the similarity transformation Eq.~\eqref{eq.SimilarityTransformation}, a general MPS transfer matrix is related to a canonical one:
\begin{equation}
	\mathbb{T}_{\mathbf{h}_1\mathbf{h}_2,\mathbf{h}_3\mathbf{h}_4} = \sum_{\mathbf{h}_{5,6,7,8}} S_{\mathbf{h}_1,\mathbf{h}_5}S^\star_{\mathbf{h}_2,\mathbf{h}_6} \breve{\mathbb{T}}_{\mathbf{h}_5\mathbf{h}_6,\mathbf{h}_7\mathbf{h}_8} S^{-1}_{\mathbf{h}_7,\mathbf{h}_3}S^{-1 \star}_{\mathbf{h}_8,\mathbf{h}_4}
\end{equation}
where $S$ is the similarity transformation. Applying Lemma \ref{lemmaD.uvproportionaltoidentity}, we get:
\begin{equation}
	\begin{split}
		&\left( \mathbb{T}^{D^2} \right)_{\mathbf{h}_1\mathbf{h}_2,\mathbf{h}_3\mathbf{h}_4} 	\\
		=& \sum_{\mathbf{h}_{5,6,7,8}} S_{\mathbf{h}_1,\mathbf{h}_5}S^\star_{\mathbf{h}_2,\mathbf{h}_6} 
		\left( \breve{\mathbb{T}}^{D^2} \right)_{\mathbf{h}_5\mathbf{h}_6,\mathbf{h}_7\mathbf{h}_8} 
		S^{-1}_{\mathbf{h}_7,\mathbf{h}_3}S^{-1 \star}_{\mathbf{h}_8,\mathbf{h}_4}	\\
		=& \sum_{\mathbf{h}_{5,6,7,8}} S_{\mathbf{h}_1,\mathbf{h}_5}S^\star_{\mathbf{h}_2,\mathbf{h}_6} 
		\frac{1}{D} \delta_{\mathbf{h}_5\mathbf{h}_6}\delta_{\mathbf{h}_7\mathbf{h}_8}
		S^{-1}_{\mathbf{h}_7,\mathbf{h}_3}S^{-1 \star}_{\mathbf{h}_8,\mathbf{h}_4}	\\
		=& \frac{1}{D}\sum_{\mathbf{h}_{5,7}} S_{\mathbf{h}_1,\mathbf{h}_5}S^\star_{\mathbf{h}_2,\mathbf{h}_5} 
		S^{-1}_{\mathbf{h}_7,\mathbf{h}_3}S^{-1 \star}_{\mathbf{h}_7,\mathbf{h}_4}	\\
	\end{split}
\end{equation}
The similarity transformation $S$ is required to be invertible, but does not have to be unitary. Hence, we conclude that Lemma \ref{lemmaD.uvproportionaltoidentity} is not true for a general MPS transfer matrix.

\begin{lemma}\label{lemmaD.transfermatrixpower}
	For a stabilizer code, the transfer matrix of the ground state canonical MPS $\breve{\mathbb{T}}$ satisfies:
	\begin{equation}
		\left( \breve{\mathbb{T}}^{n} \right)_{\mathbf{h}_1\mathbf{h}_3,\mathbf{h}_2\mathbf{h}_4} =
		\frac{1}{D} \delta_{\mathbf{h}_1\mathbf{h}_3}\delta_{\mathbf{h}_2\mathbf{h}_4}, \quad \forall\; n > D^2 \in \mathbb{N}.
	\end{equation}
\end{lemma}

\begin{proof}
	Using Lemma \ref{lemmaD.uvproportionaltoidentity}, we have
	\begin{equation}
		\begin{split}
			\left( \breve{\mathbb{T}}^{n} \right)_{\mathbf{h}_1\mathbf{h}_3,\mathbf{h}_2\mathbf{h}_4} 
			=& \left( \breve{\mathbb{T}}^{D^2} \breve{\mathbb{T}}^{n-D^2} \right)_{\mathbf{h}_1\mathbf{h}_3,\mathbf{h}_2\mathbf{h}_4}	\\
			=& \sum_{\mathbf{h}_5,\mathbf{h}_6} 
			\frac{1}{D} \delta_{\mathbf{h}_1\mathbf{h}_3} \delta_{\mathbf{h}_5\mathbf{h}_6} 
			\left( \breve{\mathbb{T}}^{n-D^2} \right)_{\mathbf{h}_5\mathbf{h}_6,\mathbf{h}_2\mathbf{h}_4}	\\
			=& \frac{1}{D} \delta_{\mathbf{h}_1\mathbf{h}_3} \sum_{\mathbf{h}_5} \left( \breve{\mathbb{T}}^{n-D^2} \right)_{\mathbf{h}_5\mathbf{h}_5,\mathbf{h}_2\mathbf{h}_4}.
		\end{split}
	\end{equation}
	Using Lemma \ref{lemmaD.canonicalMPS}, we obtain
	\begin{equation}\label{Eq.B55}
		\left( \breve{\mathbb{T}}^{n} \right)_{\mathbf{h}_1\mathbf{h}_3,\mathbf{h}_2\mathbf{h}_4} 
		= \frac{1}{D} \delta_{\mathbf{h}_1\mathbf{h}_3} \delta_{\mathbf{h}_2\mathbf{h}_4}.
	\end{equation}
	This completes the proof.
\end{proof}

We further remark that the Lemma \ref{lemmaD.transfermatrixpower} holds only when $n>D^2$, which is more restricted than the condition, i.e., $n>0$, for the  Lemma \ref{lemmaD.canonicalMPS} holds true. However, when we contract over the two virtual indices $\mathbf{h}_1$ and $\mathbf{h}_3$ (or $\mathbf{h}_2$ and $\mathbf{h}_4$) in Eq.~\eqref{Eq.B55}, we get Eq.~\eqref{Eq.B41}.

\section{Stabilizer Operator Acts on MPS Locally}
\label{app.Deriving0}

In this appendix, we prove that Eq.~\eqref{Eq.ConstraintTmatrices} (and its generic case Eq.~\eqref{Eq.MPSgeneral}) is a sufficient and necessary condition satisfied by any MPS description of the 1D stabilizer codes fulfilling the 3 assumptions of Sec.~\ref{Sec.MPSSC}. 

\begin{theorem}\label{TheoremD1}
	Eq.~\eqref{Eq.MPSgeneral} is a necessary and sufficient condition for Eq.~\eqref{Eq.GSproperty} when the system size $L\ge D^2+ \max\{P_1, \ldots , P_t\}$ where $P_1,P_2,\ldots,P_t$ is defined in Lemma \ref{LemmaC2}.
\end{theorem}

\begin{proof}
	By substituting Eq.~\eqref{Eq.MPSgeneral} into the left hand side of Eq.~\eqref{Eq.GSproperty}, it is trivial to show that Eq.~\eqref{Eq.MPSgeneral} is a sufficient condition for Eq.~\eqref{Eq.GSproperty}. Hence, our focus in the rest of the proof is to show that Eq.~\eqref{Eq.MPSgeneral} is also a necessary condition for Eq.~\eqref{Eq.GSproperty}. It suffices to prove this statement for a particular operator $\Ocal^0_{1}$. The proof generalizes to other operators. 
	
	The strategy of this proof is to first establish this statement for the canonical MPS $\breve{T}$ and then for a general MPS $T$. Typically, we will encounter many long equations where there are $T$-matrices with their physical indices uncontracted on both sides. Using the properties of the canonical MPS, i.e., Eq.~\eqref{Eq.Canonical}, we are able to shorten the equations by contracting out those $T$-matrices. We will use this trick many times below.   
	Similar to Sec.~\ref{Sec.Example}, Eq.~\eqref{Eq.GSproperty} for $\Ocal^0_1$ implies: (Notice that $\Ocal^0_1$ is supported from $r=0$ to $r=P_1-1$)
	\begin{equation}\label{Eq.4}
		\begin{split}
			&\Tr \left( 
			\Ocal^0_1\circ 
			\left(\prod_{r=0}^{P_1-1} \breve{T}^{g^r_1 \ldots g^r_q}\right) \cdot 
			\left(\prod_{r=P_1}^{L-1} \breve{T}^{g^r_1 \ldots g^r_q}\right)\right)	\\
			=&
			\Tr\left(\prod_{r=0}^{L-1} \breve{T}^{g^{r}_1\ldots g^{r}_q}\right).
		\end{split}
	\end{equation}
	Multiplying both sides with $\left(\prod_{r=P_1}^{L-1} \breve{T}^{g^r_1 \ldots g^r_q}\right)^\star_{\mathbf{h}_1\mathbf{h}_2}$ and summing over their physical indices, we obtain:
	\begin{widetext}
		\begin{equation}\label{Eq.C2}
			\begin{split}
				&\sum_{g^{P_1}_1 \ldots g^{P_1}_q \ldots g^{L-1}_1 \ldots g^{L-1}_q} 
				\Tr \left( 
				\Ocal^0_1\circ 
				\left(\prod_{r=0}^{P_1-1} \breve{T}^{g^r_1 \ldots g^r_q}\right) \cdot 
				\left(\prod_{r=P_1}^{L-1} \breve{T}^{g^r_1 \ldots g^r_q}\right)\right)	
				\left(\prod_{r=P_1}^{L-1} \breve{T}^{g^r_1 \ldots g^r_q}\right)^\star_{\mathbf{h}_1\mathbf{h}_2}	\\
				=&
				\sum_{g^{P_1}_1 \ldots g^{P_1}_q \ldots g^{L-1}_1 \ldots g^{L-1}_q} 
				\Tr\left(\prod_{r=0}^{L-1} \breve{T}^{g^{r}_1\ldots g^{r}_q}\right)
				\left(\prod_{r=P_1}^{L-1} \breve{T}^{g^r_1 \ldots g^r_q}\right)^\star_{\mathbf{h}_1\mathbf{h}_2}.
			\end{split}
		\end{equation}
	\end{widetext}
	Summing over the physical indices gives rise to transfer matrices. We rewrite this equation with explicit virtual indices as follows
	\begin{equation}
		\begin{split}
			&\sum_{\mathbf{h}_3\mathbf{h}_4}
			\Ocal^0_1\circ \left(\prod_{r=0}^{P_1-1} \breve{T}^{g^r_1 \ldots g^r_q}\right)_{\mathbf{h}_3\mathbf{h}_4} \left( \breve{\mathbb{T}}^{L-P_1} \right)_{\mathbf{h}_4\mathbf{h}_1,\mathbf{h}_3\mathbf{h}_2}	\\
			=&
			\sum_{\mathbf{h}_3\mathbf{h}_4}
			\left(\prod_{r=0}^{P_1-1} \breve{T}^{g^r_1 \ldots g^r_q}\right)_{\mathbf{h}_3\mathbf{h}_4} \left( \breve{\mathbb{T}}^{L-P_1} \right)_{\mathbf{h}_4\mathbf{h}_1,\mathbf{h}_3\mathbf{h}_2}.
		\end{split}
	\end{equation}
	Using Lemma \ref{lemmaD.transfermatrixpower} and considering $L\ge D^2+ \max\{P_1, \ldots , P_t\}$ as stated, we simplify
	\begin{equation}
		\begin{split}
			&\sum_{\mathbf{h}_3\mathbf{h}_4}
			\Ocal^0_1\circ \left(\prod_{r=0}^{P_1-1} \breve{T}^{g^r_1 \ldots g^r_q}\right)_{\mathbf{h}_3\mathbf{h}_4} \delta_{\mathbf{h}_4\mathbf{h}_1} \delta_{\mathbf{h}_3\mathbf{h}_2}	\\
			=&
			\sum_{\mathbf{h}_3\mathbf{h}_4}
			\left(\prod_{r=0}^{P_1-1} \breve{T}^{g^r_1 \ldots g^r_q}\right)_{\mathbf{h}_3\mathbf{h}_4} \delta_{\mathbf{h}_4\mathbf{h}_1} \delta_{\mathbf{h}_3\mathbf{h}_2}.
		\end{split}
	\end{equation}
	Equivalently,
	\begin{equation}\label{Eq.20}
		\Ocal^0_1\circ \left(\prod_{r=0}^{P_1-1} \breve{T}^{g^r_1 \ldots g^r_q}\right)_{\mathbf{h}_2\mathbf{h}_1}
		=\left(\prod_{r=0}^{P_1-1} \breve{T}^{g^r_1 \ldots g^r_q}\right)_{\mathbf{h}_2\mathbf{h}_1}.
	\end{equation}
	\begin{figure}
		\centering
		\includegraphics[width=1\columnwidth]{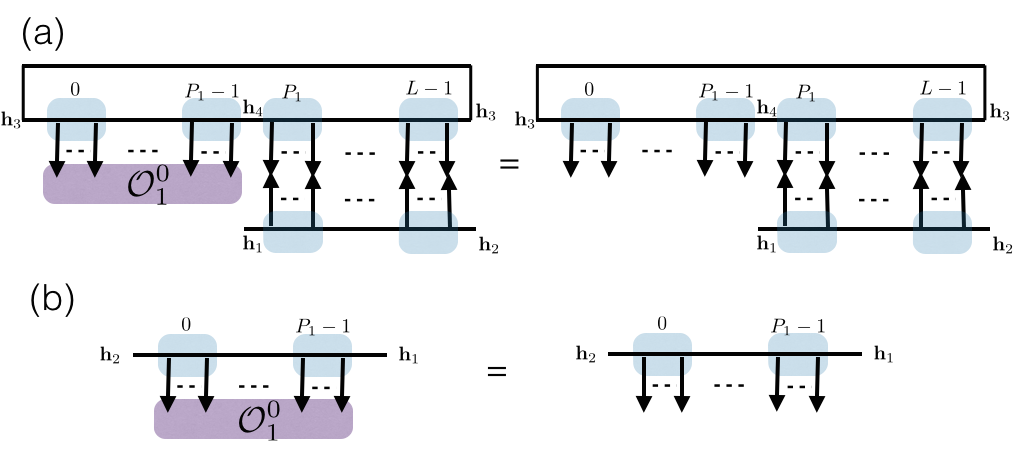}
		\caption{Graphical representation of (a) Eq.~\eqref{Eq.C2} and (b) Eq.~\eqref{Eq.20}. }
		\label{Fig_AppC}
	\end{figure}
	See Fig.~\ref{Fig_AppC} for the graphical representation of Eqs.~\eqref{Eq.C2} and \eqref{Eq.20}.  
	Notice that a general MPS tensor $T$ differs from $\breve{T}$ by a similarity transformation in Eq.~\eqref{eq.SimilarityTransformation}, then after doing a similarity transformation on both sides of Eq.~\eqref{Eq.20}, we find that an analogue equation for non-canonical MPS also holds, 
	\begin{equation}
		\Ocal^0_1\circ \left(\prod_{r=0}^{P_1-1} T^{g^r_1 \ldots g^r_q}\right)_{\mathbf{h}_2\mathbf{h}_1}
		=\left(\prod_{r=0}^{P_1-1} T^{g^r_1 \ldots g^r_q}\right)_{\mathbf{h}_2\mathbf{h}_1}.
	\end{equation}
	This completes the proof.
\end{proof}

Applying the theorem \ref{TheoremD1} to the $ZZXZZ$ model, we find that Eq.~\eqref{Eq.ConstraintTmatrices} is a necessary and sufficient condition for Eq.~\eqref{Eq.StabilizerEqsZZXZZ} when the system size is large enough, i.e., $L\geq 16+3=19$.

\section{The Action of $\Lcal$ and $\Rcal$ Operators on the MPS Matrices}
\label{app.Deriving}

\begin{theorem}\label{TheoremE1}
	Eq.~\eqref{Eq.LRgeneral} is a necessary and sufficient condition of Eq.~\eqref{Eq.MPSgeneral}.
\end{theorem}

\begin{proof}
	
	It is trivial to show that Eq.~\eqref{Eq.LRgeneral} is a sufficient condition of Eq.~\eqref{Eq.MPSgeneral}. Our focus in this proof is to show that it is also a necessary condition. Without loss of generality, we only need to prove this for a particular pair of $\Lcal$ and $\Rcal$ operators, $\Lcal^r_{1, 1}$ and $\Rcal^r_{1, 1}$.
	
	The strategy of this proof is to first establish this statement for the canonical MPS $\breve{T}$ and then for a general MPS $T$.
	The matrix element $\breve{T}^{g^{r}_{1}\ldots g^{r}_{q}}_{\mathbf{h}_1, \mathbf{h}_2}$ of a canonical MPS satisfies Eq.~\eqref{Eq.Canonical}. We start with Eq.~\eqref{Eq.MPSgeneral}, and restore the virtual indices as follows,
	\begin{equation}\label{Eq.B2}
		\begin{split}
			&\sum_{\mathbf{h}_2}
			\left(\Lcal^r_{1,1}\circ \breve{T}^{g^{r}_{1}\ldots g^{r}_{q}} \right)_{\mathbf{h}_1, \mathbf{h}_2} 
			\left( \Rcal^r_{1,1}
			\circ 
			\left( \prod_{r'=r+1}^{r+P_1-1} \breve{T}^{g^{r'}_{1} \ldots g^{r'}_q} \right)_{\mathbf{h}_2, \mathbf{h}_3} \right)\\
			&=\sum_{\mathbf{h}_2} \breve{T}^{g^{r}_{1}\ldots g^{r}_{q}} \left(\prod_{r'=r+1}^{r+P_1-1} \breve{T}^{g^{r'}_{1} \ldots g^{r'}_q} \right)_{\mathbf{h}_2, \mathbf{h}_3}.
		\end{split}
	\end{equation}
	Multiplying $\textstyle \left(\prod_{r'=r+1}^{r+P_1-1} \breve{T}^{g^{r'}_{1} \ldots g^{r'}_q} \right)^\star_{\mathbf{h}_4, \mathbf{h}_3}$ on both sides of the Eq.~\eqref{Eq.B2}, and summing over both the physical indices $g^{r'}_{1},\ldots,g^{r'}_{q}$ with $r+1\leq r'\leq r+P_1-1$ and the virtual index $\mathbf{h}_3$, we find that 
	\begin{widetext}
		\begin{equation}\label{Eq.B3}
			\begin{split}
				&\sum_{\mathbf{h}_2, \mathbf{h}_3, g^{r'}_{1},\ldots,g^{r'}_{q}|_{r+1\leq r'\leq r+P_1-1}}
				\left(\Lcal^r_{1,1}\circ \breve{T}^{g^{r}_{1}\ldots g^{r}_{q}} \right)_{\mathbf{h}_1, \mathbf{h}_2} 
				\left(\Rcal^r_{1,1}\circ \left(\prod_{r'=r+1}^{r+P_1-1} \breve{T}^{g^{r'}_{1} \ldots g^{r'}_q} \right) \right)_{\mathbf{h}_2, \mathbf{h}_3}
				\left(\prod_{r'=r+1}^{r+P_1-1} \breve{T}^{g^{r'}_{1} \ldots g^{r'}_q} \right)^\star_{\mathbf{h}_4, \mathbf{h}_3}	\\
				=&\sum_{\mathbf{h}_2, \mathbf{h}_3, g^{r'}_{1},\ldots,g^{r'}_{q}|_{r+1\leq r'\leq r+P_1-1}} 
				\breve{T}^{g^{r}_{1}\ldots g^{r}_{q}}_{\mathbf{h}_1, \mathbf{h}_2}
				\left( \prod_{r'=r+1}^{r+P_1-1} \breve{T}^{g^{r'}_{1} \ldots g^{r'}_q} \right)_{\mathbf{h}_2, \mathbf{h}_3}
				\left( \prod_{r'=r+1}^{r+P_1-1} \breve{T}^{g^{r'}_{1} \ldots g^{r'}_q} \right)^\star_{\mathbf{h}_4, \mathbf{h}_3}	\\
				=&\sum_{\mathbf{h}_2, \mathbf{h}_3} 
				\breve{T}^{g^{r}_{1}\ldots g^{r}_{q}}_{\mathbf{h}_1, \mathbf{h}_2}
				\left( \breve{\mathbb{T}}^{P_1-1} \right)_{\mathbf{h}_2\mathbf{h}_4, \mathbf{h}_3\mathbf{h}_3}	\\
				=&\sum_{\mathbf{h}_2} 
				\breve{T}^{g^{r}_{1}\ldots g^{r}_{q}}_{\mathbf{h}_1, \mathbf{h}_2}
				\delta_{\mathbf{h}_4, \mathbf{h}_2}	\\
				=&\breve{T}^{g^{r}_{1}\ldots g^{r}_{q}}_{\mathbf{h}_1, \mathbf{h}_4},
			\end{split}
		\end{equation}
		where in the third equality, we use Lemma \ref{lemmaD.canonicalMPS}. Let us define
		\begin{equation}\label{eq.notationUforL}
			\begin{split}
				(\breve{U}^r_{1,1})_{\mathbf{h}_2, \mathbf{h}_4}\equiv 
				\sum_{\mathbf{h}_3, g^{r'}_{1},\ldots ,g^{r'}_{q}|_{r+1\leq r'\leq r+P_1-1}}\left(\Rcal^r_{1,1}\circ \left(\prod_{r'=r+1}^{r+P_1-1} \breve{T}^{g^{r'}_{1} \ldots g^{r'}_q} \right) \right)_{\mathbf{h}_2, \mathbf{h}_3}
				\left(\prod_{r'=r+1}^{r+P_1-1} \breve{T}^{g^{r'}_{1} \ldots g^{r'}_q} \right)^\star_{\mathbf{h}_4, \mathbf{h}_3}.
			\end{split}
		\end{equation}
			\begin{figure}[t]
				\centering
				\includegraphics[width=0.7\columnwidth]{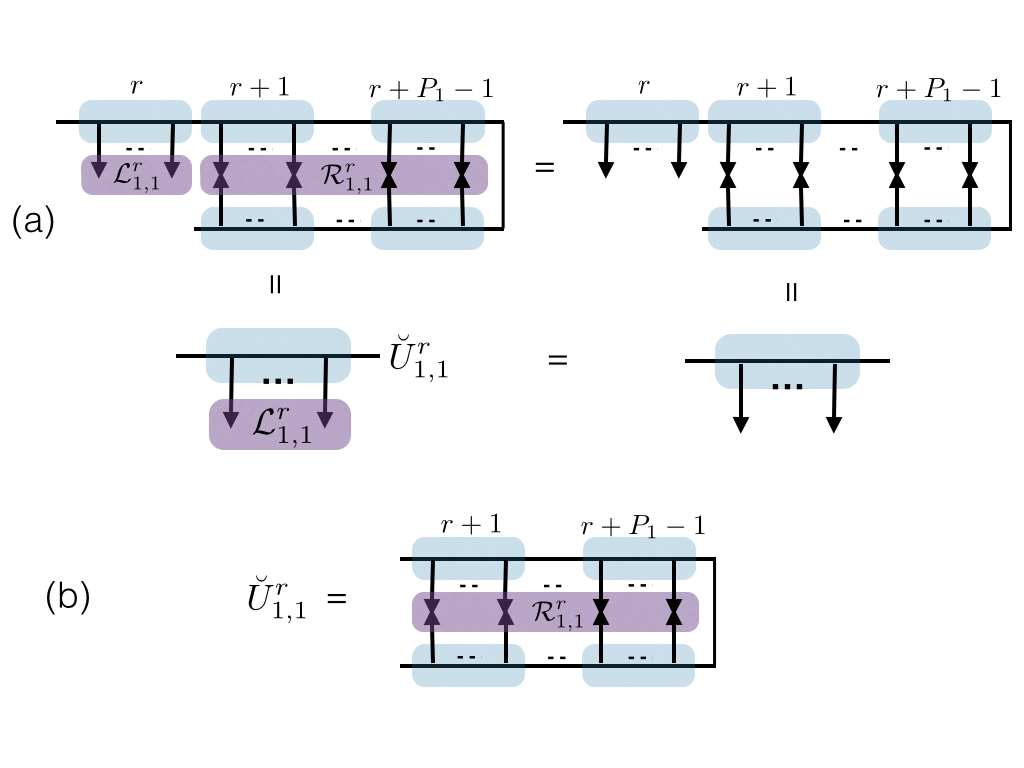}
				\caption{Graphical representation of (a) Eq.~\eqref{Eq.B3} and (b) the virtual operator $U^r_{1,1}$. }
				\label{Fig_Deriving}
			\end{figure}
	\end{widetext}
	LHS of Eq.~\eqref{Eq.B3} becomes
	\begin{equation}\label{eq.LHSofB3}
		\sum_{\mathbf{h}_2}(\Lcal^r_{1,1}\circ \breve{T}^{g^{r}_{1}\ldots g^{r}_{q}})_{\mathbf{h}_1, \mathbf{h}_2} (\breve{U}^r_{1,1})_{\mathbf{h}_2, \mathbf{h}_4}.
	\end{equation}
	Eq.~\eqref{Eq.B3} and the definition of $U^r_{1,1}$ are graphically represented in $(a)$ and $(b)$ of Fig.~\ref{Fig_Deriving} respectively. Combining Eqs.~\eqref{Eq.B3}, \eqref{eq.notationUforL} and \eqref{eq.LHSofB3}, we find
	\begin{equation}\label{Eq.D5}
		\sum_{\mathbf{h}_2}(\Lcal^r_{1,1}\circ \breve{T}^{g^{r}_{1}\ldots g^{r}_{q}})_{\mathbf{h}_1, \mathbf{h}_2} (\breve{U}^r_{1,1})_{\mathbf{h}_2, \mathbf{h}_4}=\breve{T}^{g^{r}_{1}\ldots g^{r}_{q}}_{\mathbf{h}_1, \mathbf{h}_4}.
	\end{equation}
	Applying $\Lcal^r_{1,1}$ on both sides, since $(\Lcal^r_{1,1})^2$ is an identity operator\footnote{This is because the Hamiltonian terms are Hermitian, and should be product of Hermitian operators, i.e. Pauli operators $X, Y$ and $Z$. }, we obtain
	\begin{equation}\label{Eq.B6}
		\begin{split}
			\sum_{\mathbf{h}_2}( \breve{T}^{g^{r}_{1}\ldots g^{r}_{q}})_{\mathbf{h}_1, \mathbf{h}_2} (\breve{U}^r_{1,1})_{\mathbf{h}_2, \mathbf{h}_4}=(\Lcal^r_{1,1}\circ \breve{T}^{g^{r}_{1}\ldots g^{r}_{q}})_{\mathbf{h}_1, \mathbf{h}_4}.
		\end{split}
	\end{equation}
	This is one of the first set of equations in Eq.~\eqref{Eq.LRgeneral} when the tensors are canonical. Substituting the RHS of Eq.~\eqref{Eq.B6} into the LHS of Eq.~\eqref{Eq.D5}, we find 
	\begin{eqnarray}
		\sum_{\mathbf{h}_2}(\breve{T}^{g^{r}_{1}\ldots g^{r}_{q}})_{\mathbf{h}_1, \mathbf{h}_2} [(\breve{U}^r_{1,1})^2]_{\mathbf{h}_2, \mathbf{h}_3}=(\breve{T}^{g^{r}_{1}\ldots g^{r}_{q}})_{\mathbf{h}_1, \mathbf{h}_3}.
	\end{eqnarray}
	Using the property of the canonical form Eq.~\eqref{Eq.Canonical}, we obtain that $(\breve{U}^r_{1,1})^2=I$ is an identity operator, hence
	\begin{eqnarray}
		\breve{U}^r_{1,1}=(\breve{U}^r_{1,1})^{-1}.
	\end{eqnarray}
	In particular, the $U$ matrices are invertible. 
	Since the product $\Lcal^{r}_{1,1}\Rcal^r_{1,1}$ leaves $\breve{T}^{g^{r}_{1}\ldots g^{r}_{q}}\cdot \prod_{r'=r+1}^{r+P_1-1}\breve{T}^{g^{r'}_{1}\ldots g^{r'}_{q}}$ invariant, $\Rcal^r_{1,1}$ has to transform $\prod_{r'=r+1}^{r+P_1-1}\breve{T}^{g^{r'}_{1}\ldots g^{r'}_{q}}$ as
	\begin{equation}\label{Eq.B7}
		\begin{split}
			& \left( \Rcal^r_{1,1}\circ \left( \prod_{r'=r+1}^{r+P_1-1} \breve{T}^{g^{r'}_{1} \ldots g^{r'}_q} \right) \right)_{\mathbf{h}_1, \mathbf{h}_4}	\\
			=& \sum_{\mathbf{h}_2}
			(\breve{U}^r_{1,1})^{-1}_{\mathbf{h}_1, \mathbf{h}_2}\left( \prod_{r'=r+1}^{r+P_1-1} \breve{T}^{g^{r'}_{1} \ldots g^{r'}_q} \right)_{\mathbf{h}_2, \mathbf{h}_4},
		\end{split}
	\end{equation}
	which is one of the second set of equations in Eq.~\eqref{Eq.LRgeneral} when the tensors are canonical. In Eq.~\eqref{Eq.B7}, we use $(\breve{U}^r_{1,1})^{-1}$ explicitly to manifest the fact that $(\Lcal^r_{1,1}\Rcal^r_{1,1})$ leaves the MPS invariant. Similarly, we can prove for other pairs of $\Lcal$ and $\Rcal$ operators. Therefore, we have completed the proof for the canonical MPS $\breve{T}$.

	For a generic MPS $T^{g^r_1\ldots g^r_q}$, it is related to its canonical form via a similarity transformation, Eq.~\eqref{eq.SimilarityTransformation}. The equations that $T$ obeys can be inferred from those $\breve{T}$ obeys in Eqs.~\eqref{Eq.B6} and \eqref{Eq.B7}:
	\begin{widetext}
		\begin{equation}
			\begin{split}
				&\sum_{\mathbf{h}_2}( T^{g^{r}_{1}\ldots g^{r}_{q}})_{\mathbf{h}_1, \mathbf{h}_2} (U^r_{1,1})_{\mathbf{h}_2, \mathbf{h}_4}=(\Lcal^r_{1,1}\circ T^{g^{r}_{1}\ldots g^{r}_{q}})_{\mathbf{h}_1, \mathbf{h}_4}\\&
				\sum_{\mathbf{h}_2}({U}^r_{1,1})^{-1}_{\mathbf{h}_1, \mathbf{h}_2}\left( \prod_{r'=r+1}^{r+P_1-1} {T}^{g^{r'}_{1} \ldots g^{r'}_q} \right)_{\mathbf{h}_2, \mathbf{h}_4}=\left(\Rcal^r_{1,1}\circ \left( \prod_{r'=r+1}^{r+P_1-1} {T}^{g^{r'}_{1} \ldots g^{r'}_q} \right)\right)_{\mathbf{h}_1, \mathbf{h}_4},
			\end{split}
		\end{equation}
	\end{widetext}
	where
	\begin{equation}
		U^r_{1,1}=S\cdot \breve{U}^r_{1,1}\cdot S^{-1}.
	\end{equation}
	where $S$ is the similarity transformation defined in Eq.~\eqref{eq.SimilarityTransformation}.
	Similarly for other pairs of $\Lcal$ and $\Rcal$ operators. Therefore, Eq.~\eqref{Eq.LRgeneral} also holds. This completes the proof. 
\end{proof}

Applying Theorem \ref{TheoremE1} to the $ZZXZZ$ model, we find that Eqs.~\eqref{Eq.LRsingle1}, \eqref{Eq.LRsingle2} and \eqref{Eq.LRsingle3} are necessary and sufficient conditions for Eq.~\eqref{Eq.ConstraintTmatrices}. 

\section{Commutation Relations of $U$ Operators}
\label{app.Ucommutation}

\begin{theorem}
	(Eq.~\eqref{Eq.LRgeneral}) $U_{\alpha,\tau}^{r}$ operators have the same commutation/anti-commutation relation as the $\Lcal_{\alpha,\tau}^{r}$ operators or $\Rcal_{\alpha,\tau}^{r}$ operators.
\end{theorem}

\begin{proof}
	For convenience, we first denote:
	\begin{equation}
		\Lcal^r_{\alpha,\tau} \Lcal^r_{\alpha^\prime,\tau^\prime} = (-1)^{\mathbf{t}^{r}_{\alpha\tau,\alpha^\prime\tau^\prime}} \Lcal^r_{\alpha^\prime,\tau^\prime} \Lcal^r_{\alpha,\tau}.
	\end{equation}
	where $\mathbf{t}^{r}_{\alpha\tau,\alpha^\prime\tau^\prime}$ is an integer.
	Consider these two operators acting on the tensors of the canonical MPS $\breve{T}$:
	\begin{equation}
		\begin{split}
			&\Lcal^r_{\alpha, \tau} \Lcal^r_{\alpha^\prime, \tau^\prime} \circ \left(\prod_{r'=r}^{r+\tau-1} \breve{T}^{g^{r'}_{1}\ldots g^{r'}_q} \right)	\\
			=&
			(-1)^{\mathbf{t}^{r}_{\alpha\tau,\alpha^\prime\tau^\prime}} 
			\Lcal^r_{\alpha^\prime, \tau^\prime} \Lcal^r_{\alpha, \tau} \circ \left(\prod_{r'=r}^{r+\tau-1} \breve{T}^{g^{r'}_{1}\ldots g^{r'}_q} \right).
		\end{split}
	\end{equation}
	Apply Eq.~\eqref{Eq.LRgeneral} to both sides of the equation twice when the tensor in Eq.~\eqref{Eq.LRgeneral} is the canonical one $\breve{T}$:
	\begin{equation}
		\begin{split}
			& \left(\prod_{r'=r}^{r+\tau-1} \breve{T}^{g^{r'}_{1}\ldots g^{r'}_q} \right) \breve{U}_{\alpha^\prime,\tau^\prime}^{r} \breve{U}_{\alpha,\tau}^{r} \\
			=& (-1)^{\mathbf{t}^{r}_{\alpha\tau,\alpha^\prime\tau^\prime}} 
			\left(\prod_{r'=r}^{r+\tau-1} \breve{T}^{g^{r'}_{1}\ldots g^{r'}_q} \right) \breve{U}_{\alpha,\tau}^{r} \breve{U}_{\alpha^\prime,\tau^\prime}^{r}.
		\end{split}
	\end{equation}
	Multiply both sides with $\left(\prod_{r'=r}^{r+\tau-1} \breve{T}^{g^{r'}_{1}\ldots g^{r'}_q} \right)^\dagger$ and sum over the physical indices:
	\begin{widetext}
		\begin{equation}
			\begin{split}
				&\sum_{g^{r}_{1}\ldots g^{r}_q \ldots g^{r+\tau-1}_{1}\ldots g^{r+\tau-1}_q} 
				\left(\prod_{r'=r}^{r+\tau-1} \breve{T}^{g^{r'}_{1}\ldots g^{r'}_q} \right)^\dagger \left(\prod_{r'=r}^{r+\tau-1} 
				\breve{T}^{g^{r'}_{1}\ldots g^{r'}_q} \right) 
				\breve{U}_{\alpha^\prime,\tau^\prime}^{r} \breve{U}_{\alpha,\tau}^{r}	\\
				=& (-1)^{\mathbf{t}^{r}_{\alpha\tau,\alpha^\prime\tau^\prime}} 
				\sum_{g^{r}_{1}\ldots g^{r}_q \ldots g^{r+\tau-1}_{1}\ldots g^{r+\tau-1}_q} 
				\left(\prod_{r'=r}^{r+\tau-1} \breve{T}^{g^{r'}_{1}\ldots g^{r'}_q} \right)^\dagger 
				\left(\prod_{r'=r}^{r+\tau-1} \breve{T}^{g^{r'}_{1}\ldots g^{r'}_q} \right) 
				\breve{U}_{\alpha,\tau}^{r} \breve{U}_{\alpha^\prime,\tau^\prime}^{r}.
			\end{split}
		\end{equation}
	\end{widetext}
	Using the canonical conditions Eq.~\eqref{Eq.Canonical}, we can find that:
	\begin{equation}
		\breve{U}_{\alpha^\prime,\tau^\prime}^{r} \breve{U}_{\alpha,\tau}^{r} = 
		(-1)^{\mathbf{t}^{r}_{\alpha\tau,\alpha^\prime\tau^\prime}} \breve{U}_{\alpha,\tau}^{r} \breve{U}_{\alpha^\prime,\tau^\prime}^{r}.
	\end{equation}
	Hence, we have completed the proof that $\breve{U}_{\alpha,\tau}^{r}$ operators form the same commutation relations as the $\Lcal_{\alpha,\tau}^{r}$ does. Similarly, we can prove that the $\breve{U}_{\alpha,\tau}^{r}$ operators form the same commutation relations as the $\Rcal_{\alpha,\tau}^{r}$ does. 
	
	We further discuss the case where the MPS matrix ${T}^{g^{r}_{1}\ldots g^{r}_q} $ is not canonical. Since ${T}^{g^{r}_{1}\ldots g^{r}_q} $ is related to its canonical form via a similarity transformation Eq.~\eqref{eq.SimilarityTransformation}, the virtual $U$ operator is related to $\breve{U}$ via the same similarity transformation, $S$, i.e., $U^{r}_{\alpha, \tau}= S\cdot \breve{U}^{r}_{\alpha, \tau}\cdot S^{-1}$. Hence 
	\begin{equation}
		\begin{split}
			{U}_{\alpha^\prime,\tau^\prime}^{r} {U}_{\alpha,\tau}^{r} &= S\cdot \breve{U}_{\alpha^\prime,\tau^\prime}^{r} \cdot S^{-1}\cdot S\cdot \breve{U}_{\alpha,\tau}^{r}\cdot S^{-1}\\&=S\cdot \breve{U}_{\alpha^\prime,\tau^\prime}^{r} \breve{U}_{\alpha,\tau}^{r}\cdot S^{-1}
			\\&=
			(-1)^{\mathbf{t}^{r}_{\alpha\tau,\alpha^\prime\tau^\prime}} S\cdot \breve{U}_{\alpha,\tau}^{r} \breve{U}_{\alpha^\prime,\tau^\prime}^{r}\cdot S^{-1}\\&=
			(-1)^{\mathbf{t}^{r}_{\alpha\tau,\alpha^\prime\tau^\prime}} {U}_{\alpha,\tau}^{r} {U}_{\alpha^\prime,\tau^\prime}^{r}.
		\end{split}
	\end{equation}
	So the virtual $U$ operators (associated to the non-canonical MPS) also satisfy the same commutation relation as the physical $\Lcal$ operators. 
\end{proof}

\section{Linear Equations for Local Tensors}
\label{app.Linearizing}

In this appendix, we prove that Eq.~\eqref{Eq.11} (and its generalization Eq.~\eqref{Eq.LFinial}) is a necessary and sufficient condition of Eqs.~\eqref{Eq.LRsingle1}, \eqref{Eq.LRsingle2} and \eqref{Eq.LRsingle3} (and their generalization Eq.~\eqref{Eq.LRgeneral}). 

\begin{theorem}\label{TheoremG1}
	Eq.~\eqref{Eq.LFinial} is a necessary and sufficient condition of Eq.~\eqref{Eq.LRgeneral}.
\end{theorem}

\begin{proof}
	It is trivial to show that Eq.~\eqref{Eq.LFinial} is a sufficient condition for Eq.~\eqref{Eq.LRgeneral}. Our focus in this proof is to show that Eq.~\eqref{Eq.LFinial} is also a necessary condition for Eq.~\eqref{Eq.LRgeneral}. We start with the leftmost $\Lcal$ operator in Eq.~\eqref{Eq.LRgeneral}.

	\begin{figure}[t]
		\centering
		\includegraphics[width=1\columnwidth]{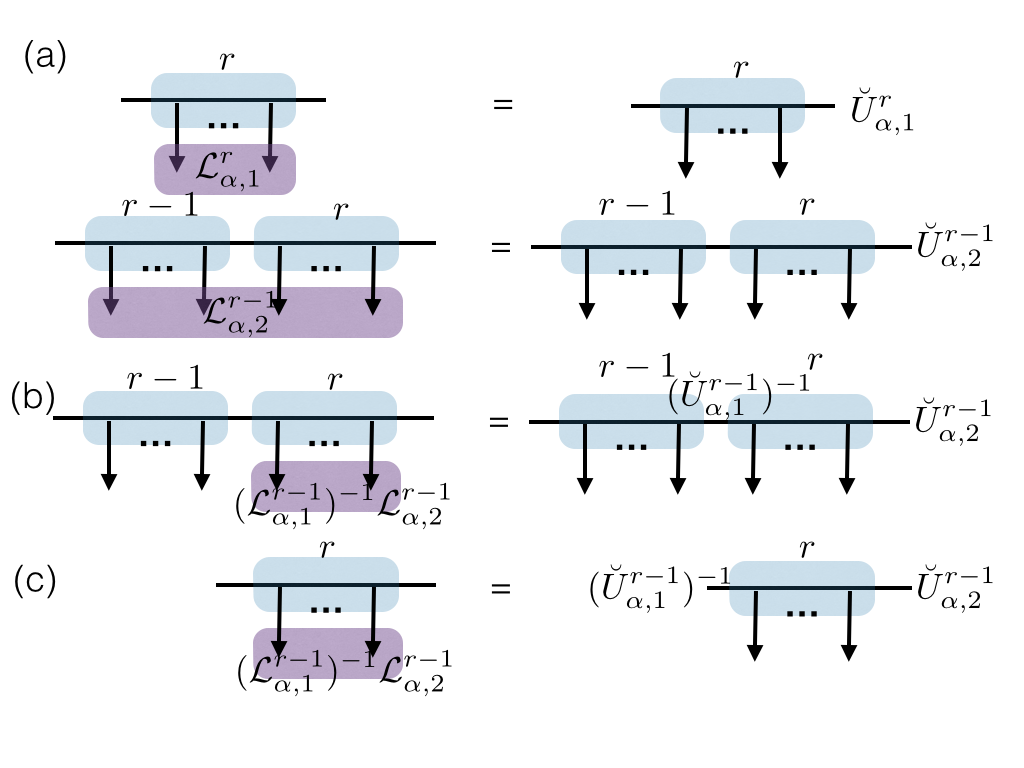}
		\caption{Graphical representation of (a) Eq.~\eqref{eq.2eqsforLRgenearl} and (b) Eq.~\eqref{Eq.F3}, and (c) Eq.~\eqref{Eq.F6}. }
		\label{Fig_AppF}
	\end{figure}
	
	By shifting the positions of Eq.~\eqref{Eq.LRgeneral}, we can obtain Eq.~\eqref{Eq.Lequations}, and we will mainly use Eq.~\eqref{Eq.Lequations}. We first consider the case when the MPS is canonical, and then discuss the general case. To prove that Eq.~\eqref{Eq.LFinial} is necessary of Eq.~\eqref{Eq.Lequations}, we use an recursive method. In particular, let us focus on the first two equations of Eq.~\eqref{Eq.Lequations} when the tensor is the canonical one $\breve{T}$: (See Fig.~\ref{Fig_AppF} (a) for the graphical representation)
	\begin{equation}\label{eq.2eqsforLRgenearl}
		\begin{split}
			\Lcal^r_{\alpha, 1} \circ \left( \breve{T}^{g^{r}_{1}\ldots g^{r}_q}\right) 
			&= \breve{T}^{g^{r}_{1}\ldots g^{r}_q} \breve{U}^r_{\alpha, 1}	\\
			\Lcal^{r-1}_{\alpha, 2} \circ \left( \breve{T}^{g^{r-1}_{1}\ldots g^{r-1}_q} \breve{T}^{g^{r}_{1}\ldots g^{r}_q} \right) 
			&= \breve{T}^{g^{r-1}_{1}\ldots g^{r-1}_q} \breve{T}^{g^{r}_{1}\ldots g^{r}_q} \breve{U}^{r-1}_{\alpha, 2}.	\\
		\end{split}
	\end{equation}
	We can apply $(\Lcal^{r-1}_{\alpha, 1})^{-1}$ to the second equation:
	\begin{equation}
		\begin{split}
			& (\Lcal^{r-1}_{\alpha, 1})^{-1} \Lcal^{r-1}_{\alpha, 2} \circ \left( \breve{T}^{g^{r-1}_{1}\ldots g^{r-1}_q} \breve{T}^{g^{r}_{1}\ldots g^{r}_q} \right) 	\\
			=& (\Lcal^{r-1}_{\alpha, 1})^{-1} \circ \left( \breve{T}^{g^{r-1}_{1}\ldots g^{r-1}_q} \breve{T}^{g^{r}_{1}\ldots g^{r}_q} \breve{U}^{r-1}_{\alpha, 2} \right).	\\
		\end{split}
	\end{equation}
	Using the first equation of Eq.~\eqref{eq.2eqsforLRgenearl} at the $(r-1)$-th site, we continue to simplify: (See Fig.~\ref{Fig_AppF} (b) for the graphical representation )
	\begin{equation}\label{Eq.F3}
		\begin{split}
			& (\Lcal^{r-1}_{\alpha, 1})^{-1} \Lcal^{r-1}_{\alpha, 2} \circ \left( \breve{T}^{g^{r-1}_{1}\ldots g^{r-1}_q} \breve{T}^{g^{r}_{1}\ldots g^{r}_q} \right) 	\\
			=& \breve{T}^{g^{r-1}_{1}\ldots g^{r-1}_q} (\breve{U}^{r-1}_{\alpha, 1})^{-1} \breve{T}^{g^{r}_{1}\ldots g^{r}_q} \breve{U}^{r-1}_{\alpha, 2}.	\\
		\end{split}
	\end{equation}
	Notice that the physical operator $(\Lcal^{r-1}_{\alpha, 1})^{-1} \Lcal^{r-1}_{\alpha, 2}$ only acts on the $r$-th site. We can rewrite this equation as:
	\begin{equation}
		\begin{split}
			& \breve{T}^{g^{r-1}_{1}\ldots g^{r-1}_q} \left( (\Lcal^{r-1}_{\alpha, 1})^{-1} \Lcal^{r-1}_{\alpha, 2} \circ \breve{T}^{g^{r}_{1}\ldots g^{r}_q} \right) 	\\
			=& \breve{T}^{g^{r-1}_{1}\ldots g^{r-1}_q} (\breve{U}^{r-1}_{\alpha, 1})^{-1} \breve{T}^{g^{r}_{1}\ldots g^{r}_q} \breve{U}^{r-1}_{\alpha, 2}.	\\
		\end{split}
	\end{equation}
	Multiply both sides with $\left( \breve{T}^{g^{r-1}_{1}\ldots g^{r-1}_q} \right)^\dagger$ and sum over the physical indices:
	\begin{widetext}
		\begin{equation}
			\begin{split}
				\sum_{g^{r-1}_{1}\ldots g^{r-1}_q} \left( \breve{T}^{g^{r-1}_{1}\ldots g^{r-1}_q} \right)^\dagger 
				\breve{T}^{g^{r-1}_{1}\ldots g^{r-1}_q} 
				\left( (\Lcal^{r-1}_{\alpha, 1})^{-1} \Lcal^{r-1}_{\alpha, 2} \circ \breve{T}^{g^{r}_{1}\ldots g^{r}_q} \right) 
				= \sum_{g^{r-1}_{1}\ldots g^{r-1}_q} \left( \breve{T}^{g^{r-1}_{1}\ldots g^{r-1}_q} \right)^\dagger \breve{T}^{g^{r-1}_{1}\ldots g^{r-1}_q} (\breve{U}^{r-1}_{\alpha, 1})^{-1} \breve{T}^{g^{r}_{1}\ldots g^{r}_q} \breve{U}^{r-1}_{\alpha, 2}.
			\end{split}
		\end{equation}
	\end{widetext}
	Now we can apply Eq.~\eqref{Eq.Canonical} at the $(r-1)$-th site: (See Fig.~\ref{Fig_AppF} (c) for the graphical representation)
	\begin{equation}\label{Eq.F6}
		(\Lcal^{r-1}_{\alpha, 1})^{-1} \Lcal^{r-1}_{\alpha, 2} \circ \breve{T}^{g^{r}_{1}\ldots g^{r}_q} 	
		= (\breve{U}^{r-1}_{\alpha, 1})^{-1} \breve{T}^{g^{r}_{1}\ldots g^{r}_q} \breve{U}^{r-1}_{\alpha, 2}.
	\end{equation}
	Hence, we have proved the following equations for the canonical MPS $\breve{T}$:
	\begin{equation}
		\begin{split}
			\Lcal^r_{\alpha, 1} \circ \breve{T}^{g^{r}_{1}\ldots g^{r}_q}
			&=\breve{T}^{g^{r}_{1}\ldots g^{r}_q} \cdot \breve{U}^r_{\alpha, 1}\\
			\left((\Lcal^{r-1}_{\alpha, 1})^{-1}\Lcal^{r-1}_{\alpha, 2}\right) \circ \breve{T}^{g^r_1\cdots g^r_q}
			&=(\breve{U}^{r-1}_{\alpha, 1})^{-1}\cdot \breve{T}^{g^r_1\cdots g^r_q}\cdot \breve{U}^{r-1}_{\alpha, 2}.\\
		\end{split}
	\end{equation}
	
	By iterating the process, we can prove the rest of the equations in Eq.~\eqref{Eq.LFinial} for the canonical MPS with tensor $\breve{T}$. The same statement is true for a general MPS with a tensor $T$, since the tensor $T$ and $\breve{T}$ are related by the similarity transformation in Eq.~\eqref{eq.SimilarityTransformation}. Therefore, we have completed our proof.
\end{proof}

Applying the theorem \ref{TheoremG1} to the $ZZXZZ$ model, we find that Eq.~\eqref{Eq.11} is the necessary and sufficient condition for Eqs.~\eqref{Eq.LRsingle1}, \eqref{Eq.LRsingle2} and \eqref{Eq.LRsingle3}.

\begin{theorem}\label{theoremF.sameLsameU}
	If $\Lcal^r_{\alpha, 1}, U^r_{\alpha, 1}$ and $ V^r_{\alpha, 1}$ satisfy
	\begin{equation}\label{Eq.21}
		\begin{split}
			\Lcal^r_{\alpha, 1}\circ T^{g^r_1\ldots g^r_q}&= T^{g^r_1\ldots g^r_q}\cdot U^r_{\alpha, 1}\\
			\Lcal^r_{\alpha, 1}\circ T^{g^r_1\ldots g^r_q}&= T^{g^r_1\ldots g^r_q}\cdot V^r_{\alpha, 1},
		\end{split}
	\end{equation}
	then $U^r_{\alpha, 1}=V^r_{\alpha, 1}$. 
\end{theorem}

\begin{proof}
	We first prove when the $T$ matrix is canonical. Since $\Lcal^r_{\alpha, 1}$ and $\Lcal^r_{\beta, 1}$ are identical physical operators, LHS of Eq.~\eqref{Eq.21} are the same. Hence
	\begin{equation}
		\breve{T}^{g^r_1\ldots g^r_q}\cdot \breve{U}^r_{\alpha, 1}=\breve{T}^{g^r_1\ldots g^r_q}\cdot \breve{V}^r_{\alpha, 1},
	\end{equation}
	where $\breve{T}^{g^r_1\ldots g^r_q}$ is the canonical MPS matrix, and $\breve{U}^r_{\alpha, 1}$ and $\breve{V}^r_{\alpha, 1}$ are the associated virtual operator. In components, 
	\begin{equation}
		\sum_{\mathbf{h}_2}(\breve{T}^{g^r_1\ldots g^r_q})_{\mathbf{h}_1, \mathbf{h}_2} (\breve{U}^r_{\alpha, 1})_{\mathbf{h}_2, \mathbf{h}_3}=\sum_{\mathbf{h}_2}(\breve{T}^{g^r_1\ldots g^r_q})_{\mathbf{h}_1, \mathbf{h}_2} (\breve{V}^r_{\alpha, 1})_{\mathbf{h}_2, \mathbf{h}_3}.
	\end{equation}
	Multiplying $(\breve{T}^{g^r_1\ldots g^r_q})^*_{\mathbf{h}_1, \mathbf{h}_4}$ on both sides, and summing over $\mathbf{h}_1$ as well as the physical indices $g^r_1, \ldots , g^r_q$, and using the canonical condition Eq.~\eqref{Eq.Canonical}, we find
	\begin{equation}
		(\breve{U}^r_{\alpha, 1})_{\mathbf{h}_4, \mathbf{h}_3}=(\breve{V}^r_{\alpha, 1})_{\mathbf{h}_4, \mathbf{h}_3}.
	\end{equation}
	
	When the MPS is not canonical, we apply the similarity transformation Eq.~\eqref{eq.SimilarityTransformation}: 
	\begin{equation}
		U^r_{\alpha, 1}=S\cdot \breve{U}^r_{\alpha, 1}\cdot S^{-1}, ~ V^r_{\alpha, 1}=S\cdot \breve{V}^r_{\alpha, 1}\cdot S^{-1}.
	\end{equation}
	So
	\begin{equation}
		{U}^r_{\alpha, 1}=S\cdot \breve{U}^r_{\alpha, 1}\cdot S^{-1}=S\cdot \breve{V}^r_{\alpha, 1}\cdot S^{-1}={V}^r_{\alpha, 1}.
	\end{equation}
	This completes the proof. 
\end{proof}

\section{Virtual $U$ Operators as Tensor Products of Pauli Matrices}
\label{App. VirtualUOperator}

In this appendix, we show that the virtual $U$ operators can be constructed as tensor products of Pauli matrices. 

As discussed in the paragraph before Eq.~\eqref{Eq.29} in Sec.~\ref{Sec.GeneralMPS} and proved in Ref.~\onlinecite{2010arXiv1005.4300J}, the anti-symmetric integer matrix $\mathbf{t}$ can be block diagonalized by a unimodular integer matrix $V$, such that each nontrivial block is a $2\times 2$ anti-symmetric matrix with integer off-diagonal matrix elements.
Consider a general set of operators $\{U_i\}$ $(i=1, ..., N)$ which either commute or anti-commute, 
\begin{equation}\label{Eq.30}
	U_i U_j =(-1)^{\mathbf{t}_{ij}}U_{j} U_i.
\end{equation}
Let us define a new set of operators using the unimodular integer matrix $V$ as follows
\begin{equation}
	\widetilde{U}_i=U_{1}^{V_{i1}}U_{2}^{V_{i2}}...U_{N}^{V_{iN}},
\end{equation}
where $V_{ij}$ are the entries of the unimodular integer matrix $V$. It is straightforward to compute the commutation relations of $\{\widetilde{U}_i\}$, 
\begin{equation}
	\begin{split}
		\widetilde{U}_i \widetilde{U}_j &=(-1)^{\sum_{k, l}V_{ik}\mathbf{t}_{kl} (V^T)_{lj}}\widetilde{U}_{j} \widetilde{U}_i\\
		&=(-1)^{(V\cdot \mathbf{t}\cdot V^T)_{ij}}\widetilde{U}_{j} \widetilde{U}_i.
	\end{split}
\end{equation}
Due to Eq.\eqref{Eq.29}, $V\cdot \mathbf{t}\cdot V^T$ is block diagonalized. Since $V\cdot \mathbf{t}\cdot V^T$ appears on the exponent of $(-1)$, only the modulo $2$ values of the matrix elements matter. Hence the nontrivial $2\times 2$ blocks have off-diagonal elements $\pm 1$ where we keep the minus signs to make the anti-symmetry manifest. Suppose $n$ is the number of nontrivial blocks of the $V\cdot \mathbf{t}\cdot V^T$. Then one can find the representations of $\widetilde{U}_i$ by using the Pauli matrices, because each $2 \times 2$ block corresponds to a pair of anti-commuting operators. For an irreducible representation, we can assign for instance
\begin{equation}\label{Eq.31}
	\begin{split}
		\widetilde{U}_i=
		\begin{cases}
			\underbrace{I\otimes ...\otimes I }_{\frac{i-1}{2}}\otimes X\otimes \underbrace{I\otimes ...\otimes I}_{\frac{2n-i-1}{2}}&, ~~ i~\mathrm{is~ odd}, 1\leq i\leq 2n\\
			\underbrace{I\otimes ...\otimes I}_{\frac{i-2}{2}} \otimes Z\otimes \underbrace{I\otimes ...\otimes I}_{\frac{2n-i}{2}}&, ~~ i~\mathrm{is~ even}, 1\leq i\leq 2n\\
			\underbrace{I\otimes ...\otimes I \otimes I\otimes I\otimes ...\otimes I}_{n}&, ~~ 2n+1\leq i\leq N,\\
		\end{cases}
	\end{split}
\end{equation}
where $n = \frac{\mathrm{rank}(\mathbf{t})}{2}$, and each $\widetilde{U}_i$ is a tensor product of $n$ Pauli matrices, forming a $2^{\frac{\mathrm{rank}(\mathbf{t})}{2}}=2^{n}$ dimensional representation. 
Since $V$ is unimodular, we can do an inverse transformation from $\{\widetilde{U}_i\}$ to $\{U_i\}$. 
\begin{equation}
	U_i= \widetilde{U}_1^{(V^{-1})_{i1}}... \widetilde{U}_N^{(V^{-1})_{iN}}.
\end{equation}
Since $\{\widetilde{U}_i\}$ are tensor product of Pauli matrices, $\{U_i\}$ are also tensor product of Pauli matrices. This generalizes the construction of Sec.~\ref{Sec.Example}. 

\section{Projective Representations and 1D Symmetry Protected Topological Phases}
\label{app.SPT}

\subsection{Projective Representations and Cocycles}
\label{app.SubsectionA1}

In this section, we describe projective representations and cocycles. Suppose $G$ is a discrete group and $\rho(g)$ is a matrix representation of the group element $g\in G$. $\rho$ is the projective representation of $G$ if 
\begin{equation}\label{Eq.algebra}
	\rho(g_1)\rho(g_2)= \omega_2(g_1,g_2) \rho(g_1g_2), \quad \forall\; g_1,g_2\in G,
\end{equation}
where $\omega_2(g_1,g_2)$ is a $U(1)$ phase. As a result of Eq.~\eqref{Eq.algebra} being associative, i.e., 
\begin{equation}
	\Big(\rho(g_1)\rho(g_2)\Big)\rho(g_3)= \rho(g_1)\Big(\rho(g_2)\rho(g_3)\Big).
\end{equation}
$\omega_2(g_1,g_2)$ satisfies:
\begin{equation}\label{eq.cocyclecondition}
	\omega_2(g_1,g_2)\omega_2(g_1g_2,g_3) = \omega_2(g_2,g_3)\omega_2(g_1,g_2g_3).
\end{equation}
We further require that ${\rho}(g)$ and ${\rho}(g)\mu_1(g)$  belongs to the same class of the projective representation, where $ \mu_1(g)$ is a $U(1)$ phase. This yields that if two cocycles, $\omega_2$ and $\tilde{\omega}_2$, are related by $\mu_1$ as follows: 
\begin{equation}\label{eq.identification}
	\tilde{\omega}_2(g_1,g_2) = \mu_1(g_1) \mu_1(g_2)\mu_1(g_1g_2)^{-1} \omega_2(g_1,g_2),
\end{equation}
then they give rise to the same projective representation.
The conditions Eqs.~\eqref{eq.cocyclecondition} and \eqref{eq.identification} require the $U(1)$ phase $\omega_2$ belongs to the group cohomology $\mathcal{H}^2(G,U(1))$ and is a cocycle.\cite{propitius1995topological,dijkgraaf1990,chen2013symmetry}. 

Throughout the paper, $G$ is an Abelian group of the form $(\mathbb{Z}_2)^{q}$, and the group element $g$ is parametrized by $g=(g_1,g_2,\ldots,g_q)$ with $g_i \in \mathbb{Z}_2$ for $i=1,2,\ldots,q$. All the cocycles in $\mathcal{H}^2(G,U(1))$ are parametrized as in Eq.~\eqref{eq.cocycleexpression}\cite{wang2015field,propitius1995topological}.

\subsection{Cocycle States}

In this subsection, we summarize the construction of a class of short range entangled states which we dub as the the cocycle states,  following Ref.~\onlinecite{chen2013symmetry}. These states are interesting because they are the states describing the \emph{symmetry protected topological} (SPT) phase, protected by the on-site unitary symmetry $G$. We first set up the notations, and then review their results with Abelian groups for simplicity. 

Consider a 1D lattice with $L$ unit cells. In each unit cell, the local Hilbert space basis can be labeled by the elements of $G$: $\ket{g^r}, \forall\; {g}^r\in G$, ($r=0, 1, ..., L-1$). Besides the group elements $\{{g}^r\}$, Ref.~\onlinecite{chen2013symmetry} also introduced an auxiliary group element ${g}^\star\in G$ which does not belong to the Hilbert space, but nevertheless enables one to cons.  The cocycle state is constructed as follows (see Eq.~(54) of Ref.~\onlinecite{chen2013symmetry})
\begin{equation}\label{Eq.23}
	|\psi\rangle_{G, \omega_2}= \sum_{\{{g}^r\}} \bigg(\sum_{g^\star} \prod_{r=0}^{L-1} \omega_2({g}^{r}-{g}^{r-1}, {g}^{\star}-{g}^{r})\bigg)|\{{g}^r\}\rangle.
\end{equation}

We further restrict Eq.~\eqref{Eq.23} to the $(\Z_2)^q$ group. As introduced in App.~\ref{app.SubsectionA1}, each unit cell contains $q$ number of $\Z_2$ group elements/spins, i.e., ${g}^r=(g^r_1, ..., g^r_q)$. A generic $\omega_2$ is in Eq.~\eqref{eq.cocycleexpression}, i.e.,
\begin{equation}\label{Eq.cocycle}
	\begin{split}
		&\omega_2({g}^{r}-{g}^{r-1}, {g}^{\star}-{g}^{r})\\&=\exp\bigg( -i \pi \sum_{1\leq i<j\leq q} P_{ij} (g^{r}_j-g^{r-1}_j)(g^{\star}_i-g^{r}_i) \bigg).
	\end{split}
\end{equation}
Plugging Eq.~\eqref{Eq.cocycle} to \eqref{Eq.23}, the cocycle state of $(\Z_2)^q$ global symmetry becomes:
\begin{equation}\label{Eq.8}
	\begin{split}
		|\psi\rangle_{(\Z_2)^q, \omega_2}= &\sum_{\{g_i^r\}} \Bigg(\sum_{\{g_i^\star\}} \exp\bigg( -i\pi \sum_{r=0}^{L-1}\sum_{1\leq i<j\leq q}\\& P_{ij}(g^{r}_j-g^{r-1}_j)(g^{\star}_i-g^{r}_i) \bigg)\Bigg) |\{g^r_i\}\rangle.
	\end{split}
\end{equation}
Notice that in the exponent,  the coefficient of $g^\star_i$ with fixed $j$, i.e., $-i\pi \sum_{r=0}^{L-1}P_{ij}(g^{r}_j-g^{r-1}_j)$, vanishes due to PBC. This further simplifies the cocycle state Eq.~\eqref{Eq.8} to
\begin{equation}\label{Eq.SPTGS}
	\begin{split}
		&|\psi\rangle_{(\Z_2)^q, \omega_2}
		=\\& \sum_{\{g_i^r\}} \exp\left( i\pi \sum_{r=0}^{L-1}\sum_{1\leq i<j\leq q} P_{ij}(g^{r}_j-g^{r-1}_j)g^r_i \right) |\{g^r_i\}\rangle.
	\end{split}
\end{equation}

\subsection{Cocycle Hamiltonians}

We now construct a cocycle Hamiltonian $H_{(\Z_2)^q, \omega_2}$ whose ground state is Eq.~\eqref{Eq.SPTGS}. The cocycle Hamiltonian has been constructed in Refs.~\onlinecite{hu2013twisted,wan2015twisted}. We present a simplified construction.

\begin{lemma}\label{LemmaH1}
	There exist $qL$ operators $\Ocal^r_\alpha$ defined by 
	\begin{equation}\label{Eq.SPTOperator}
		\Ocal^{r}_\alpha\equiv \prod_{1\leq k< \alpha}(Z^{r+1}_k Z^{r}_k)^{P_{k\alpha}} X^r_\alpha \prod_{\alpha<l\leq q}(Z^r_lZ^{r-1}_l)^{P_{\alpha l}}
	\end{equation}
	satisfying
	\begin{equation}\label{Eq.StabilizerCondition}
		\Ocal^{r}_{\alpha}|\psi\rangle_{(\Z_2)^q, \omega_2}=|\psi\rangle_{(\Z_2)^q, \omega_2}, \forall r\in [0, L-1], \alpha\in \{1, \ldots , q\}.
	\end{equation}
\end{lemma}

In the main text, we adopt a slightly different but equivalent convention to label all the operators $\Ocal^{r}_\alpha$ using translation symmetry. See Eq.~\eqref{Eq.10}. In the main text, the convention adopted in Eq.~\eqref{Eq.10} is consistent with the discussion of general stabilizer code Eq.~\eqref{Eq.LR}. In this appendix, $\Ocal^{r}_\alpha$ in Eq.~\eqref{Eq.SPTOperator} shares the same label with $X^r_\alpha$ in its expression. The convention in Eq.~\eqref{Eq.SPTOperator}  will simplify the proof without repeating the same equations for different labels.

\begin{proof}
	We first act $X^r_\alpha$ on $|\psi\rangle_{(\Z_2)^q, \omega_2}$ \eqref{Eq.SPTGS}, 
	\begin{widetext}
		\begin{equation}
			\begin{split}
				X^r_\alpha|\psi\rangle_{(\Z_2)^q, \omega_2}&=\sum_{\{g^{\widehat{r}}_k\}}\exp\bigg( i\pi \sum_{\widehat{r}=0}^{L-1}\sum_{1\leq k<l\leq q} P_{kl}(g^{\widehat{r}}_l-g^{\widehat{r}-1}_l)g^{\widehat{r}}_k \bigg) X^r_\alpha|\{g^{\widehat{r}}_k\}\rangle\\
				&= \sum_{\{g^{\widehat{r}}_k\}}\exp\bigg( i\pi \sum_{\widehat{r}=0}^{L-1}\sum_{1\leq k<l\leq q} P_{kl}(g^{\widehat{r}}_l-g^{\widehat{r}-1}_l)g^{\widehat{r}}_k \bigg) |\{g^{\widehat{r}}_k+\delta^{\widehat{r}r}\delta_{k\alpha}\}\rangle.
			\end{split}
		\end{equation}
		In the second line, we used the fact that since the group element $g^r_\alpha$ is defined mod 2, $g^{\widehat{r}}_k+\delta^{\widehat{r}r}\delta_{k\alpha}$ is equivalent to flipping the value of the spin $g^{r}_\alpha$. We further redefine the spins as $\widetilde{g}^{\widehat{r}}_k=g^{\widehat{r}}_k+\delta^{\widehat{r}r}\delta_{k\alpha}$, and rewrite the equation as 
		\begin{equation}
			\begin{split}
				X^r_\alpha|\psi\rangle_{(\Z_2)^q, \omega_2}&=\sum_{\{\widetilde{g}^{\widehat{r}}_k\}}\exp\bigg( i\pi \sum_{\widehat{r}=0}^{L-1}\sum_{1\leq k<l\leq q} P_{kl}(\widetilde{g}^{\widehat{r}}_l-\widetilde{g}^{\widehat{r}-1}_l-\delta^{\widehat{r} r}\delta_{l\alpha}+\delta^{(\widehat{r}-1)r}\delta_{l\alpha})(\widetilde{g}^{\widehat{r}}_k-\delta^{\widehat{r}r}\delta_{k\alpha}) \bigg) |\{\widetilde{g}_{\widehat{r}}^k\}\rangle\\
				&=\sum_{\{\widetilde{g}^{\widehat{r}}_k\}}\exp\bigg( i\pi \sum_{\widehat{r}=0}^{L-1}\sum_{1\leq k<l\leq q} P_{kl}(\widetilde{g}^{\widehat{r}}_l-\widetilde{g}^{\widehat{r}-1}_l)\widetilde{g}^{\widehat{r}}_k-i\pi\sum_{1\leq k<\alpha} P_{k\alpha}(\widetilde{g}^{r}_k-\widetilde{g}^{r+1}_k)-i\pi \sum_{\alpha<l\leq q}P_{\alpha l}(\widetilde{g}_l^{r}-\widetilde{g}_l^{r-1}) \bigg) |\{\widetilde{g}_{\widehat{r}}^k\}\rangle\\
				&=\prod_{1\leq k< \alpha}(Z^{r+1}_k Z^{r}_k)^{P_{k\alpha}} \prod_{\alpha<l\leq q}(Z^r_lZ^{r-1}_l)^{P_{\alpha l}}\sum_{\{g^{\widehat{r}}_k\}}\exp\bigg( i\pi \sum_{\widehat{r}=0}^{L-1}\sum_{1\leq k<l\leq q} P_{kl}(g^{\widehat{r}}_l-g^{\widehat{r}-1}_l)g^{\widehat{r}}_k \bigg) |\{g^{\widehat{r}}_k\}\rangle\\
				&=\prod_{1\leq k< \alpha}(Z^{r+1}_k Z^{r}_k)^{P_{k\alpha}} \prod_{\alpha<l\leq q}(Z^r_lZ^{r-1}_l)^{P_{\alpha l}}|\psi\rangle_{(\Z_2)^q, \omega_2}.
			\end{split}
		\end{equation} 
	\end{widetext}
	In the second line, the first term on the exponent has exactly the same form as the original $|\psi\rangle_{(\Z_2)^q, \omega_2}$, while the second and the third terms on the exponent are extra terms. They can be reproduced by the acting with the product of Pauli $Z$ operators, $\prod_{1\leq k< \alpha}(Z_{r+1}^k Z_{r}^k)^{P_{k\alpha}} \prod_{\alpha <l\leq q}(Z_r^lZ_{r-1}^l)^{P_{\alpha l}}$. This observation directly leads to the third line. Hence we find the following combination leaves $|\psi\rangle_{(\Z_2)^q, \omega_2}$ invariant:
	\begin{equation}
		\Ocal^{r}_\alpha\equiv \prod_{1\leq k< \alpha}(Z^{r+1}_k Z^{r}_k)^{P_{k\alpha}} X^r_\alpha \prod_{\alpha<l\leq q}(Z^r_lZ^{r-1}_l)^{P_{\alpha l}}.
	\end{equation}
	This completes the proof. 
\end{proof}

\begin{lemma}\label{LemmaH2}
	The operators $\Ocal^r_\alpha$ in Lemma~\ref{LemmaH1} mutually commute, i.e,
	\begin{eqnarray}
		[\Ocal^r_\alpha, \Ocal^{r'}_{\alpha'}]=0, ~~~\forall r, r', \alpha, \alpha'.
	\end{eqnarray}
\end{lemma}

\begin{proof}
	Without loss of generality, we assume $r=1$. Then $\Ocal^{1}_{\alpha}$ only acts on the unit cells at $0, 1$ and $2$. $\Ocal^1_{\alpha}$ and $\Ocal^{r'}_{\alpha'}$ trivially commute unless $X^1_{\alpha}$ overlap with a Pauli $Z$ operator of $\Ocal^{r'}_{\alpha'}$ and/or $X^{r'}_{\alpha'}$ overlap with a Pauli $Z$ operator of $\Ocal^1_{\alpha}$. It is straightforward to check that the Pauli $X$ and $Z$ operators overlap when 
	\begin{enumerate}
		\item $r'=2$ and $\alpha>\alpha'$. 
		\item $r'=1$ and $\alpha\neq \alpha'$. 
		\item $r'=0$ and $\alpha<\alpha'$. 
	\end{enumerate}
	When $r'=2$ and $\alpha>\alpha'$, 
	\begin{eqnarray}
		\Ocal^1_\alpha\Ocal^{2}_{\alpha'}= (-1)^{P_{\alpha'\alpha}}(-1)^{P_{\alpha'\alpha}}\Ocal^{2}_{\alpha'}\Ocal^1_\alpha= \Ocal^{2}_{\alpha'}\Ocal^1_\alpha.
	\end{eqnarray}
	When $r'=1$ and $\alpha> \alpha'$, 
	\begin{eqnarray}
		\Ocal^1_\alpha\Ocal^{1}_{\alpha'}=(-1)^{P_{\alpha'\alpha}}(-1)^{P_{\alpha'\alpha}}\Ocal^{1}_{\alpha'}\Ocal^1_\alpha=\Ocal^{1}_{\alpha'}\Ocal^1_\alpha.
	\end{eqnarray}
	When $r'=1$ and $\alpha< \alpha'$, 
	\begin{eqnarray}
		\Ocal^1_\alpha\Ocal^{1}_{\alpha'}=(-1)^{P_{\alpha\alpha'}}(-1)^{P_{\alpha\alpha'}}\Ocal^{1}_{\alpha'}\Ocal^1_\alpha=\Ocal^{1}_{\alpha'}\Ocal^1_\alpha.
	\end{eqnarray}
	When $r'=0$ and $\alpha<\alpha'$, 
	\begin{eqnarray}
		\Ocal^1_\alpha\Ocal^{0}_{\alpha'}= (-1)^{P_{\alpha\alpha'}}(-1)^{P_{\alpha\alpha'}}\Ocal^{0}_{\alpha'}\Ocal^1_\alpha= \Ocal^{0}_{\alpha'}\Ocal^1_\alpha.
	\end{eqnarray}
	
	In summary, we have proven that for any $r', \alpha, \alpha'$, $[\Ocal^1_\alpha, \Ocal^{r'}_{\alpha'}]=0$. By translational invariance, $[\Ocal^r_\alpha, \Ocal^{r'}_{\alpha'}]=0, ~\forall r, r',\alpha, \alpha'$. This completes the proof. 
\end{proof}

\begin{lemma}\label{LemmaH3}
	The operators $\Ocal^r_\alpha$ in Lemma~\ref{LemmaH1} are all independent for different $r=0, ..., L-1$ and $\alpha=1, ..., q$. 
\end{lemma}

\begin{proof}
	The observation is that each $\Ocal^r_\alpha$ involves only one Pauli $X$ operator, $X^r_{\alpha}$. Then all operators $\Ocal^r_\alpha$ are independent.
\end{proof}

\begin{lemma}\label{LemmaH4}
	The commuting Hamiltonian 
	\begin{eqnarray}\label{Eq.34}
		H_{(\Z_2)^q, \omega_2}=-\sum_{r=0}^{L-1} \sum_{\alpha=1}^{q} \Ocal^r_{\alpha}.
	\end{eqnarray}
	has only one ground state. 
\end{lemma}

\begin{proof}
	We prove by counting the degrees of freedom and the number of independent constraints. Since each unit cell contains $q$ spins and there are $L$ unit cells, the total dimension of the Hilbert space is $2^{qL}$. From Lemma \ref{LemmaH2}, all the operators in the Hamiltonian commute. Thus the ground state $|\psi\rangle_{(\Z_2)^q, \omega_2}$ must be stabilized by all the operators satisfying 
	\begin{eqnarray}\label{Eq.33}
		\Ocal^r_\alpha|\psi\rangle_{(\Z_2)^q, \omega_2}=|\psi\rangle_{(\Z_2)^q, \omega_2}.
	\end{eqnarray}
	From Lemma \ref{LemmaH3}, all the operators $\Ocal^r_\alpha$ are independent. Hence each Eq.~\eqref{Eq.33} provides one independent constraint for the ground state Hilbert space. Because $\Ocal^r_\alpha$ is a product of Pauli operators, each equation in Eq.~\eqref{Eq.33} eliminates half of the Hilbert space dimension. Since there are $qL$ independent equations, the number of ground state is $2^{qL-qL}=1$. Hence there is only one ground state. 
\end{proof}

Summarizing Lemma \ref{LemmaH1}, \ref{LemmaH2}, \ref{LemmaH3} and \ref{LemmaH4}, we have constructed the cocycle Hamiltonian:
\begin{theorem}
	The cocycle state Eq.~\eqref{Eq.SPTGS} is stabilized by the cocycle Hamiltonian
	\begin{eqnarray}\label{Eq.CocycleHamiltonian}
		H_{(\Z_2)^q, \omega_2}=-\sum_{r=0}^{L-1} \sum_{\alpha=1}^{q} \Ocal^r_{\alpha},
	\end{eqnarray}
	where
	\begin{equation}
		\Ocal^{r}_\alpha\equiv \prod_{1\leq k< \alpha}(Z^{r+1}_k Z^{r}_k)^{P_{k\alpha}} X^r_\alpha \prod_{\alpha<l\leq q}(Z^r_lZ^{r-1}_l)^{P_{\alpha l}}.
	\end{equation}
	The Hamiltonian satisfies
	\begin{enumerate}
		\item All the operators $\Ocal^{r}_\alpha$ are products of Pauli operators, and mutually commute. 
		\item There is a unique ground state $|\psi\rangle_{(\Z_2)^q, \omega_2}$ with PBC. 
	\end{enumerate}
\end{theorem}

\section{Some Useful Identities}
\label{app.UsefulIden}

In this appendix, we prove that Eq.~\eqref{Eq.Identities_n} holds. We first prove a Lemma which turns out to be useful in proving Eq.~\eqref{Eq.Identities_n}. 
\begin{lemma}\label{LemmaI1}
	If $x$ is an integer, then the following equation holds. 
	\begin{equation}\label{Eq.appD3}
		\begin{split}
			\exp\bigg(i\pi \frac{1}{2} x^2 \bigg)=\frac{1+\exp\left(i\pi x\right)}{2}+i \frac{1-\exp\left(i\pi x\right)}{2}.
		\end{split}
	\end{equation}
\end{lemma}

\begin{proof}
	When $x$ is an even integer, both sides are $1$. When $x$ is an odd integer, both sides are $i$. Hence Eq.~\eqref{Eq.appD3} holds. 
\end{proof}

\begin{lemma}
	Eq.~\eqref{Eq.Identities_n} holds. 
\end{lemma}

\begin{proof}
	We start with the LHS of Eq.~\eqref{Eq.Identities_n}. Using $\sum_{i<j} g_i g_j=\frac{1}{2} \left( (\sum_{i}g_i)^2 - \sum_{i}g_i^2 \right)$, we reduce the LHS to 
	\begin{equation}\label{Eq.appD2}
		\exp\bigg(i\pi \sum_{i<j} g_i g_j\bigg)=\exp\bigg(i\pi \frac{1}{2} \left( (\sum_{i}g_i)^2 - \sum_{i}g_i^2 \right)\bigg).\\
	\end{equation}
	If we further restrict the value of $g_i$ as $g_i\in \{0,1\}$, we have $g_i^2=g_i$, hence $\sum_i g_i^2=\sum_i g_i$. Applying Lemma~\ref{LemmaI1} with $x=\sum_{i}g_i$, we further reduce Eq.~\eqref{Eq.appD2} to 
	\begin{eqnarray}
		\begin{split}
			&\bigg(\frac{1+e^{i \pi \sum_{i=1}^{n}g_i}}{2}+i \frac{1-e^{i \pi \sum_{i=1}^{n}g_i}}{2}\bigg)e^{-\frac{i\pi}{2}\sum_{i=1}^{n}g_i}\\&=\sqrt{2} \cos\left( \frac{\pi}{2} \left( \sum_{i=1}^{n} g_i - \frac{1}{2} \right)\right).
		\end{split}
	\end{eqnarray}
	Introducing a hidden variable $h$ to write the RHS in the RBM form, we find the RHS is precisely
	\begin{eqnarray}
		\frac{1}{\sqrt{2}}\sum_{h=0}^{1} \exp\bigg(i\frac{\pi}{2}(1-2h)\sum_{i=1}^n g_i-i\frac{\pi}{4}(1-2h)\bigg).
	\end{eqnarray}
	This completes the proof. 
\end{proof}

Two simple examples of Eq.~\eqref{Eq.Identities_n} are:
\begin{equation}
	\begin{split}
		&\exp\bigg(i\pi g_1 g_2\bigg)	
		\\&=\frac{1}{\sqrt{2}}\sum_{h=0}^{1} \exp\bigg(i\frac{\pi}{2}(1-2h)(g_1+g_2)-i\frac{\pi}{4}(1-2h)\bigg)
	\end{split}
\end{equation}
for $n=2$ and 
\begin{equation}
	\begin{split}
		&\exp\bigg(i\pi (g_1 g_2+ g_1 g_3+g_2g_3)\bigg)	
		\\&=\frac{1}{\sqrt{2}}\sum_{h=0}^{1} \exp\bigg(i\frac{\pi}{2}(1-2h)(g_1+g_2+g_3)-i\frac{\pi}{4}(1-2h)\bigg)
	\end{split}
\end{equation} 
for $n=3$.

\section{More Examples of RBM for Cocycle Model}
\label{app.J}


In this appendix, we exemplify the construction of the RBM state in Sec.~\ref{Sec.RBMCocycle} by the cocycle model with $P_{12}=P_{13}=\cdots = P_{1q}=1$ and $P_{ij}=0$ with $i\geq 2$ and $j>i$. 

The Hamiltonian of the model is 
\begin{equation}
	\begin{split}
		H_{(\Z_2)^q, \omega_2}=&-\sum_{r=0}^{L-1} \Bigg(\prod_{i=2}^{q}Z^r_{i} X^{r+1}_1 \prod_{i=2}^{q}Z^{r+1}_i+ \sum_{i=2}^{q} Z^{r}_1 Z^{r+1}_1 X^r_i~ \Bigg).
	\end{split}
\end{equation}
The ground state is 
\begin{equation}\label{Eq.GSExample2}
	|\mathrm{GS}\rangle_{(\Z_2)^q, \omega_2}=\sum_{\{g^r_i\}}\prod_{r=0}^{L-1}\exp\bigg(i\pi \sum_{i=2}^q (g_i^{r}-g^{r-1}_{i})g^r_{1}\bigg)|\{g^r_i\}\rangle.
\end{equation}
The $q\times q$ $\Gamma$ matrix (defined in Eq.~\eqref{Eq.Gamma}) is
\begin{equation}\label{Eq.GammaMatrixExample2}
	\Gamma=
	\begin{pmatrix}
		0 & 0 & \cdots & 0 & 0\\
		1 & 0& \cdots& 0& 0\\
		1 & 0 & \cdots & 0& 0\\
		\vdots &\vdots & \ddots & \vdots& 0 \\
		1 & 0 & \cdots &0 & 0
	\end{pmatrix}.
\end{equation}
Applying the procedures introduced in the proof of Lemma.~\ref{lemma.GaussianElimination}, we first use row operations to set the all the rows of Eq.~\eqref{Eq.GammaMatrixExample2} to zero except the first row. Recall $G_1$ and $G_2$ defined in Eq.~\eqref{Eq.G1G2}. The row operation is
\begin{eqnarray}
	G^T= G_2(1, q-1)G_2(1,q-2)\cdots G_{2}(1,2) G_1(1,q).
\end{eqnarray}
The visible spins transform as
\begin{eqnarray}
	\begin{pmatrix}
		g^r_{1}\\
		g^r_{2}\\
		\vdots\\
		g^r_{q-1}\\
		g^r_{q}
	\end{pmatrix}
	\to 
	\begin{pmatrix}
		\widehat{g}^r_{1}\\
		\widehat{g}^r_{2}\\
		\vdots\\
		\widehat{g}^r_{q-1}\\
		\widehat{g}^r_{q}
	\end{pmatrix}
	= 
	G^{-1}\cdot 
	\begin{pmatrix}
		g^r_{1}\\
		g^r_{2}\\
		\vdots\\
		g^r_{q-1}\\
		g^r_{q}
	\end{pmatrix}
	=
	\begin{pmatrix}
		\sum_{i=2}^{q}g^r_{i}\\
		g^r_{2}\\
		\vdots\\
		g^r_{q-1}\\
		g^r_{1}
	\end{pmatrix}.
\end{eqnarray}
The $\Gamma$ matrix is transformed to
\begin{eqnarray}
	\Gamma\to \widehat{\Gamma}=G^T\cdot \Gamma \cdot G= 
	\begin{pmatrix}
		0 & 0 & \cdots & 0 & 1\\
		0 & 0& \cdots& 0& 0\\
		0 & 0 & \cdots & 0& 0\\
		\vdots &\vdots & \ddots & \vdots& 0 \\
		0 & 0 & \cdots &0 & 0
	\end{pmatrix}.
\end{eqnarray}
Hence the rank of the $\Gamma$ matrix is
\begin{eqnarray}
	\mathrm{rank}(\Gamma)=\mathrm{rank}(\widehat{\Gamma})=1.
\end{eqnarray}
Using the identity Eq.~\eqref{Eq.Identities_n}, we only need to introduce one hidden spin of type $h$ and type $\widetilde{h}$ respectively to express the exponent in Eq.~\eqref{Eq.GSExample2} in terms of RBM, 
\begin{eqnarray}
	\begin{split}
		&\sum_{r=0}^{L-1}\sum_{i=2}^q (g_i^{r}-g^{r-1}_{i})g^r_{1}
		=\\&\sum_{r=0}^{L-1}\bigg(-\Sym(g_1^r, \sum_{i=2}^{q}g_i^{r-1})	
		+\Sym(g_1^r, \sum_{i=2}^{q}g_i^{r})\bigg).
	\end{split}
\end{eqnarray}

The ground state Eq.~\eqref{Eq.GSExample2} can be written as an RBM state
\begin{widetext}
	\begin{equation}
		\begin{split}
			|\mathrm{GS}\rangle_{(\Z_2)^q, \omega_2}
			&=\sum_{\{g^r_i\}, \{h^r_1\}, \{\widetilde{h}^r_1\}} \prod_{r=0}^{L-1}\exp\bigg(-i\frac{\pi}{2}(1-2h_1^r)(g_1^r+\sum_{i=2}^q g^{r-1}_i)+i\frac{\pi}{4}(1-2h_1^r)+i\frac{\pi}{2}(1-2\widetilde{h}^r_1)\sum_{i=1}^q g^r_i-i\frac{\pi}{4}(1-2\widetilde{h}^r_1)\bigg) |\{g^r_i\}\rangle.
		\end{split}
	\end{equation}
	This RBM can be casted into an MPS with bond dimension 2, and the matrix elements of the RBM-MPS are:
	\begin{equation}
		\begin{split}
			T^{g^r_1, \ldots , g^r_q}_{h^r_1, h^{r+1}_1}
			&=\exp\bigg(-i\frac{\pi}{2}(1-2h_1^r)g_1^r- i\frac{\pi}{2}(1-2h_1^{r+1})\sum_{i=2}^q g^r_i+i\frac{\pi}{4}(1-2h_1^r)\bigg)\sum_{\widetilde{h}^r_1=0}^1\exp\bigg(i\frac{\pi}{2}(1-2\widetilde{h}^r_1)\sum_{i=1}^q g^r_i-i\frac{\pi}{4}(1-2\widetilde{h}^r_1)\bigg).
		\end{split}
	\end{equation}
\end{widetext}
We also present the RBM for two examples in Fig.~\ref{FigRBMExample2q=3} and \ref{FigRBMExample2q=4} corresponding to $q=3$ and $q=4$. 

\begin{figure}[H]
	\centering
	\includegraphics[width=1\columnwidth]{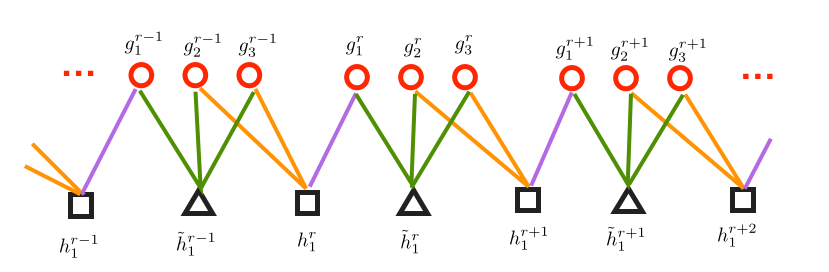}
	\caption{RBM network for cocycle model with $q=3, P_{12}=P_{13}=1, P_{23}=0$.}
	\label{FigRBMExample2q=3}
\end{figure}

\begin{figure}[H]
	\centering
	\includegraphics[width=1\columnwidth]{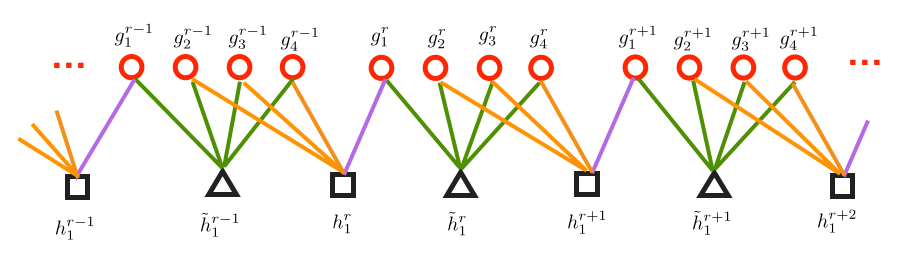}
	\caption{RBM network for cocycle model with $q=4, P_{12}=P_{13}=P_{14}=1, P_{23}=P_{24}=P_{34}=0$.}
	\label{FigRBMExample2q=4}
\end{figure}

\bibliography{RBM_MPS}

\end{document}